\documentclass[12pt]{article}

\usepackage[margin=1in]{geometry}

\usepackage{setspace}
\usepackage{float}
\usepackage{amsmath,amssymb,amsthm}
\usepackage{lscape}
\usepackage{bm}
\usepackage{tabularx}
\usepackage{longtable}
\usepackage{color}
\usepackage{graphicx}
\usepackage{caption}
\usepackage{subcaption}
\usepackage{changepage}
\usepackage{amsfonts}
\usepackage{amstext}
\usepackage{mathrsfs}
\usepackage{rotating}
\usepackage{multirow}
\usepackage{xcolor}
\usepackage{ulem}
\usepackage{thrmappendix}
\usepackage{threeparttable}
\usepackage{enumitem}
\usepackage{adjustbox}
\usepackage{bbm}
\usepackage[longnamesfirst]{natbib}
\usepackage{mathtools}
\usepackage{microtype}
\setlist{nosep}
\setlist[itemize]{itemsep=2pt,topsep=2pt}
\setlist[enumerate]{itemsep=2pt,topsep=2pt}

\newcounter{IDsec}
\newcommand{\aAssump}{A\arabic{IDsec}}
\newcounter{estsec}
\newcommand{\bAssump}{B\arabic{estsec}}

\newtheoremstyle{theoremSuppressedNumber}{}{}{}{}{\bfseries}{.}{ }{\thmname{#1}\thmnote{ (\mdseries #3)}}

\makeatletter
\renewcommand{\theequation}{%
\thesection.\arabic{equation}}
\@addtoreset{equation}{section}

\theoremstyle{definition}
\newtheorem{theorem}{Theorem}[section]
\newtheorem{lemma}{Lemma}[section]

\newtheorem{proposition}{Proposition}[section]
\theoremstyle{definition}	
\newtheorem{definition}{Definition}[section]
\newtheorem{remark}{Remark}[section]
\newtheorem{algorithm}{Algorithm}[section]
\newtheorem{assumption}{Assumption}[section]

\newcommand{\indep}{\perp \! \! \! \perp}
\newcommand{\avg}{\frac{1}{n}\sum_{i=1}^n}
\newcommand{\pconv}{\xrightarrow{p}}
\newcommand{\dconv}{\xrightarrow{d}}
\newcommand{\R}{\mathbb{R}}
\newcommand{\E}{\mathbb{E}}
\newcommand{\amax}{a_{\max}}
\newcommand{\alphamax}{\alpha_{\max}}
\newcommand{\Prob}{\mathbb{P}}
\renewcommand{\P}{\mathbb{P}}
\def\avgt{\frac{1}{T} \sum_{t=1}^T}

\DeclareMathOperator{\var}{var}
\DeclareMathOperator{\cov}{cov}

\newcommand{\1}{\mathbbm{1}}
\newcommand{\supp}{\text{supp}}
\DeclareMathOperator*{\argmax}{arg\,max}

\newcommand{\hMIV}{\widehat{\mathrm{MIV}}}
\newcommand{\hMR}{\widehat{\mathrm{MR}}}

\usepackage{xr}
\usepackage{titlesec}

\titleformat{\section}
  {\normalfont\normalsize\bfseries\centering}
  {\thesection}{1em}{}

\titleformat{\subsection}
  {\normalfont\normalsize\bfseries\centering}
  {\thesubsection}{1em}{}

\titleformat{\subsubsection}
  {\normalfont\normalsize\bfseries\centering}
  {\thesubsubsection}{1em}{}

\titlespacing*{\section}{0pt}{0.7\baselineskip}{0.5\baselineskip}
\titlespacing*{\subsection}{0pt}{0.65\baselineskip}{0.45\baselineskip}
\titlespacing*{\subsubsection}{0pt}{0.6\baselineskip}{0.4\baselineskip}

\usepackage{etoolbox}
\makeatletter
\patchcmd{\@startsection}
  {\@afterindenttrue}
  {\@afterindentfalse}
  {}{}
\makeatother

\begin{document}

\title{\textbf{Quantifying the Internal Validity of Weighted Estimands}\footnote{First arXiv draft:~April 22, 2024. This version:~July 28, 2026. This paper was presented at LMU Munich, University of Bonn, Brown University, Brandeis University, University of Pittsburgh, McMaster University, Boston College, University of Oslo, Norwegian School of Economics, Universidad Carlos III de Madrid, University of North Carolina at Chapel Hill, Syracuse University, University of Warwick, University of Oxford, University of Notre Dame, the 2024 Annual Congress of the European Economic Association, the 2024 DC-MD-VA Econometrics Workshop, the 2024 Midwest Econometrics Group Meeting, the 2024 Southern Economic Association Meeting, the 2024 EC$^2$ Conference, the 2024 BU/BC Greenline Econometrics Workshop, the 2025 Econometrics Mini-Conference at the University of Iowa, the 2025 Annual Conference of the International Association for Applied Econometrics, the 2025 Aarhus Workshop in Econometrics, the 2026 North American Winter Meeting of the Econometric Society, the 2026 Annual Meetings of the Society of Labor Economists, and the 2026 Interactions Conference at the University of Chicago. We thank audiences at those seminars and conferences, as well as Young Ahn, St\'ephane Bonhomme, Greg Caetano, Brant Callaway, Kevin Chen, Cl\'ement de Chaisemartin, Juan Carlos Escanciano, Joachim Freyberger, Paul Goldsmith-Pinkham, Christian Hansen, Peter Hull, Shakeeb Khan, Toru Kitagawa, Matt Masten, Tomasz Olma, Guillaume Pouliot, Jonathan Roth, Pedro Sant'Anna, Andres Santos, Alex Torgovitsky, and Daniel Wilhelm for helpful conversations and comments.}}

\author{Alexandre Poirier\thanks{Department of Economics, Georgetown University, \texttt{alexandre.poirier@georgetown.edu}}
\qquad
Tymon S{\l}oczy\'nski\thanks{Department of Economics, Brandeis University, \texttt{tslocz@brandeis.edu}}
}

\date{}

\maketitle

\singlespacing

\newsavebox{\tablebox} \newlength{\tableboxwidth}

\vspace{-0.5cm}

\begin{abstract}
\begin{spacing}{1}
\noindent In this paper we study a class of weighted estimands, which we define as parameters that can be expressed as weighted averages of the underlying heterogeneous treatment effects. The popular ordinary least squares (OLS), two-stage least squares (2SLS), and two-way fixed effects (TWFE) estimands are all special cases within our framework. Our focus is on answering two questions concerning weighted estimands. First, under what conditions can they be interpreted as the average treatment effect for some (possibly latent) subpopulation? Second, when these conditions are satisfied, what is the upper bound on the size of that subpopulation, either in absolute terms or relative to a target population of interest? We argue that this upper bound provides a valuable diagnostic for empirical research. When a given weighted estimand corresponds to the average treatment effect for a small subset of the population of interest, we say its internal validity is low. Our paper develops practical tools to quantify the internal validity of weighted estimands. We also implement these tools in our companion Stata package \texttt{causalrep} and use them to revisit four prominent empirical studies.
\end{spacing}
\end{abstract}

\vspace{-0.6cm}

\begin{small}
 \noindent

\end{small}

\thispagestyle{empty}

\onehalfspacing

\pagebreak

\section{Introduction}\label{sec:intro}

Estimating average treatment effects is a central goal in many areas of empirical research. Even though applied researchers usually believe that treatment effects are heterogeneous, they often favor well-established estimation methods that were not designed with treatment effect heterogeneity in mind. In many cases, these methods lead to estimands that can be represented as weighted averages of the underlying treatment effects of interest. The econometrics literature has highlighted the problems that arise when some of these weights are negative. In this paper, we build a framework to study a different problem that is often overlooked: What happens when these weights are nonuniform, even if they are nonnegative? Uniform, or constant, weights imply that the resulting estimand is equal to the average treatment effect for all units under consideration. By contrast, we show that nonuniform weights can sharply limit the internal validity of an estimand, defined as the size of the largest subpopulation whose average treatment effect is equal to the estimand. We develop diagnostic tools to quantify this loss of internal validity in a given application.

To illustrate this problem, consider a scenario in which unconfoundedness holds given covariates $X$. Let treatment $D$ be binary, let $(Y(1),Y(0))$ be potential outcomes, and let $\tau_0(X) \coloneqq \E[Y(1) - Y(0)\mid X]$ be the conditional average treatment effect (CATE)\@. Suppose the researcher is interested in the average treatment effect (ATE), defined as $\E[Y(1) - Y(0)]$. Many estimators are consistent for the ATE\@. Yet, in practice, it is common to estimate a linear regression of $Y$ on $D$ and $X$\@. Following \cite{Angrist1998}, if we also assume that $\P(D=1 \mid X)$ is linear in $X$, the resulting ordinary least squares (OLS) estimand is equal to
\begin{align}\label{eq:OLS_estimand}
	\beta_\text{OLS} = \E\left[\frac{\var(D\mid X)}{\E[\var(D\mid X)]} \cdot \tau_0(X)\right],
\end{align}
a weighted average of CATEs that generally differs from the ATE\@. The underlying weights, $\var(D \mid X)/\E[\var(D \mid X)]$, are nonnegative. This is encouraging because negative weights may lead the resulting estimand to have the opposite sign from all the CATEs.\footnote{Such an estimand is not \textit{weakly causal}, following the terminology of \cite{BlandholBonneyMogstadTorgovitsky2026}.} We further show that, under unrestricted treatment effect heterogeneity, nonnegativity is necessary and sufficient for a weighted estimand to correspond to the average treatment effect for some subpopulation. At the same time, the weights in equation \eqref{eq:OLS_estimand} need not be uniform, and the largest such subpopulation may be small and nonrepresentative. Whether it is small, however, depends on the data-generating process: it may instead be large or even coincide with the entire population. Our primary aim is thus to quantify the size of this implicit subpopulation, so that applied researchers can distinguish between these possibilities.

More generally, in this paper we study a broad class of \emph{weighted estimands} that can be expressed as follows:
\begin{align}
	\mu(a,\tau_0) &\coloneqq \E[a(X)\tau_0(X) \mid W_0 = 1],\label{eq:weighted_est_def}
\end{align}
where $W_0 \in \{0,1\}$ is an indicator for a subpopulation, $\tau_0(X) = \E[Y(1) - Y(0)\mid W_0 = 1, X]$ is the CATE function given $X$ in the same subpopulation $W_0$, and $a(X)$ is an identified weight function that integrates to 1, i.e., $\E[a(X) \mid W_0 = 1] = 1$. The OLS estimand in equation \eqref{eq:OLS_estimand} belongs to this class, which can be seen by letting $W_0 = 1$ with probability 1, and letting the weight function $a(X)$ be the normalized variance of $D$ given $X$\@. Under standard assumptions, this class also includes the two-stage least squares (2SLS) and two-way fixed effects (TWFE) estimands in instrumental variables and difference-in-differences settings. Here, the leading cases of $W_0$ are compliers in the case of 2SLS and treated units in the case of TWFE\@.

We build our framework by asking under what conditions the weighted estimand in \eqref{eq:weighted_est_def} corresponds to an average treatment effect of the form $\E[Y(1)-Y(0)\mid W^*=1]$, where $W^*\in \{0,1\}$ is an indicator for a (possibly latent) subpopulation of $W_0$. An affirmative answer to this question endows a particular estimand with some degree of validity as a causal parameter. It also allows us to address the primary question that we seek to answer: How can we \emph{quantify} the degree of validity of $\mu(a,\tau_0)$ as a causal parameter?

To do this, we characterize the size, and the size relative to $W_0$, of subpopulations $W^*$ associated with the estimand in \eqref{eq:weighted_est_def}. In other words, we ask how large $\Prob(W^*=1)$ and $\Prob(W^*=1\mid W_0=1)$ can be in the representation $\mu(a,\tau_0) = \E[Y(1) - Y(0)\mid W^* = 1]$. If $\E[Y(1) - Y(0)\mid W_0 = 1]$ is the parameter of interest, we interpret a large value of $\Prob(W^*=1\mid W_0=1)$ as evidence of high \emph{internal validity} of $\mu(a,\tau_0)$ with respect to the target. For brevity, we refer to the largest possible value of $\Prob(W^*=1\mid W_0=1)$, our measure of internal validity, as the MIV\@. Similarly, if the corresponding marginal probability, $\Prob(W^*=1)$, is large, we say that $\mu(a,\tau_0)$ is highly \emph{representative} of the underlying population. We refer to the largest possible value of $\Prob(W^*=1)$, our measure of representativeness, as the MR\@.

The answer to our questions about subpopulation existence and size depends on the information we have about the CATE function, $\tau_0$. When this function is unrestricted, we formally show that $\mu(a,\tau_0)$ can be written as the average treatment effect for a subpopulation of $W_0$ \textit{if and only if} the weights $a(X)$ are nonnegative. The contrapositive of this statement is that the presence of negative weights implies that $\mu(a,\tau_0)$ cannot be represented as an average treatment effect for a subpopulation uniformly in $\tau_0$. This result provides a novel justification for the common requirement that the weights underlying a suitable estimand must be nonnegative. Our main contribution, however, is to quantify the size of the largest implicit subpopulation for which such a representation is possible. In particular, when $\tau_0$ is unrestricted and the weights are nonnegative, we show that the MIV is simply $\amax^{-1}$, where $\amax$ denotes the largest value of the normalized weight $a(X)$.

A low value of $\amax^{-1}$ is therefore a cause for concern. If $\E[Y(1) - Y(0)\mid W_0 = 1]$ is the parameter of interest and $\mu(a,\tau_0)$ corresponds to the average effect for at most a fraction $\amax^{-1}$ of $W_0$, then $\mu(a,\tau_0)$ is directly informative only about that fraction of the target population. The average effect for the remaining fraction, $1-\amax^{-1}$, can drive the gap between the weighted estimand and the target parameter. Our paper derives bounds that formalize this intuition: when CATEs are known only to lie in $[\underline{\tau},\overline{\tau}]$, the width of the sharp bounds on $\E[Y(1) - Y(0)\mid W_0 = 1]$ is equal to $(1-\amax^{-1})(\overline{\tau}-\underline{\tau})$, and these bounds are valid even in the presence of negative weights. Thus, regardless of whether the weights are nonnegative, a low value of $\amax^{-1}$, reflecting highly nonuniform weights, reduces the informativeness of the weighted estimand about the target parameter. At the other extreme, when the weights are uniform, $\amax^{-1}=1$, and the weighted estimand equals the parameter of interest, $\E[Y(1) - Y(0)\mid W_0 = 1]$, for any $\tau_0$. Relatedly, a high value of $\amax^{-1}$ limits how much the covariate distribution can differ between the target population, $W_0$, and the implicit subpopulation, $W^*$.

We also consider the possibility that the CATE function is known. If $\tau_0$ is known, we show that $\mu(a,\tau_0)$ can be written as an average treatment effect whenever it lies in the convex hull of CATE values, a weaker criterion than having nonnegative weights. The MIV and MR now depend on $\tau_0$, and can be obtained via linear programming when $X$ is discrete. We show that this linear program has a closed-form solution even when the support of $X$ includes continuous components, which would usually lead to an infinite-dimensional linear program. If the researcher estimates and compares the MIVs and MRs in the cases where $\tau_0$ is unrestricted and where it is known, they can evaluate the importance of treatment effect heterogeneity for the interpretation of $\mu(a,\tau_0)$ in a given application.

We further specialize our framework to several popular estimands in additive linear models: OLS under unconfoundedness, 2SLS under local average treatment effect (LATE) assumptions, and TWFE under parallel trends and variation in treatment timing. For each, we provide closed-form expressions for the MIV\@. We also propose an analog estimator of $\amax^{-1}$, establish its nonstandard limiting distribution, and provide inference procedures. We compute our diagnostics in four empirical applications, which draw on papers published in the \textit{American Economic Review} and \textit{Quarterly Journal of Economics}. These applications illustrate how our framework can diagnose problems with negative weights and quantify the loss of internal validity caused by highly nonuniform but nonnegative weights.

\subsection*{Implications for Empirical Practice}

Our results have a simple practical implication: we recommend that researchers using weighted estimands report an estimate of $\amax^{-1}$ alongside their main estimates. When weights are nonnegative, $\amax^{-1}$ is exactly the MIV, the largest fraction of the target population $W_0$ for which the estimand can be interpreted as an average treatment effect under unrestricted treatment effect heterogeneity. Thus, small values of $\amax^{-1}$ indicate that the estimand may directly represent only a small and potentially selected subset of the target population. By contrast, large values of $\amax^{-1}$ imply large implicit subpopulations, at least relative to $W_0$. At the same time, if some weights are negative and treatment effect heterogeneity remains unrestricted, the MIV is zero, whereas $\amax^{-1}$ still enters our sharp bounds on the target average effect and bounds the possible fraction of negative weights. For this reason, researchers should report an estimate of $\amax^{-1}$ itself, possibly alongside an estimate of the MIV\@. To facilitate practical use, we implement this diagnostic in our companion Stata package, \texttt{causalrep}.

This diagnostic is practical because it does not use outcome data. It depends only on the weighting structure induced by the research design and right-hand-side specification, so in empirical applications with multiple outcomes, it can be estimated once and reported as a single design-stage diagnostic. When weights are nonnegative, researchers can also characterize the implicit subpopulation $W^*$ by estimating and reporting means of $X$ for units with $W^*=1$ and comparing them to those in the target population $W_0$\@.

\subsection*{Literature Review}

This paper contributes to the literature on the causal interpretation of weighted estimands, including OLS, 2SLS, and TWFE in additive linear models.\footnote{Contributions to this literature include \cite{Angrist1998}, \cite{AronowSamii2016}, \cite{Sloczynski2022}, \cite{ZhaoDing2022}, \cite{Chen2024}, \cite{Goldsmith-PinkhamHullKolesar2024}, \cite{Humphreys2025}, and \cite{BugniCanayMcBride2026} for OLS; \cite{ImbensAngrist1994}, \cite{AngristImbens1995}, \cite{Kolesar2013}, \cite{BlandholBonneyMogstadTorgovitsky2026}, and \cite{Sloczynski2026} for 2SLS; and \cite{ChaisemartinDHaultfoeuille2020}, \cite{Goodman-Bacon2021}, \cite{SunAbraham2021}, \cite{AI2022}, \cite{BJS2024}, \cite{CaetanoCallaway2024}, and \cite{CallawayGoodman-BaconSantAnna2024} for TWFE\@.} A common view in this literature, which goes back at least to \cite{ImbensAngrist1994}, is that causal interpretability of weighted estimands requires all weights to be nonnegative. \cite{BlandholBonneyMogstadTorgovitsky2026} show that the absence of negative weights and level dependence is necessary and sufficient for an estimand to be ``weakly causal,'' that is, to guarantee sign preservation when all treatment effects have the same sign. In this paper, we study the related problem of whether a weighted estimand can be written as the average treatment effect over a subpopulation. While the absence of negative weights is essential for subpopulation existence, nonuniform weights may sharply reduce its size. This relates to the skeptical view of both negative and nonuniform weights in \cite{CallawayGoodman-BaconSantAnna2024}.

Some papers focus on weighted averages of heterogeneous treatment effects as legitimate targets rather than merely as probability limits of existing estimators. \cite{HIR2003} introduce the class of weighted average treatment effects, a subclass of the more general class in \eqref{eq:weighted_est_def}. \cite{LMZ2018} discuss the connection between weighted average treatment effects and implicit subpopulations. However, the internal validity and representativeness of weighted estimands have received little attention to date.

One exception is work by \cite{Chaisemartin2012,Chaisemartin2017} on the interpretation of the instrumental variables (IV) estimand. First, \citet[][Section 3]{Chaisemartin2012} studies whether the local average treatment effect (LATE) can be extrapolated beyond compliers and derives bounds on the size of the largest subpopulation to which it applies. The sharp lower bound equals the complier share, so the LATE may apply only to compliers. The sharp upper bound can be below one, which in some applications rules out extrapolation to the entire population. Unlike in our paper, however, the maximum size of this subpopulation is not point identified. Second, when monotonicity fails, \cite{Chaisemartin2012,Chaisemartin2017} uses a restriction on the CATE function to reinterpret the IV estimand as the average treatment effect for a subset of compliers. In our framework, this is an existence result in an intermediate case between the setting in which $\tau_0$ is unrestricted and the setting in which it is fixed.

Finally, there is precedent for prioritizing larger implicit subpopulations. \cite{MT2024} argue that ``[t]arget parameters that reflect larger subpopulations of the population of interest are more interesting than those that reflect smaller and more specific subpopulations.'' In a setting with multiple instrumental variables, \cite{HLM2025} suggest that the largest subpopulation of compliers is generally more interesting than other complier subpopulations.

\subsection*{Plan of the Paper and Notation}

We organize the paper as follows. Section \ref{sec:example1} briefly discusses our motivating example of the OLS estimand. Section \ref{sec:causal_represent} develops our theoretical framework and examines the conditions under which the estimand in \eqref{eq:weighted_est_def} has a causal representation as an average treatment effect over a subpopulation. Section \ref{sec:quantifying_int_valid} establishes our main results on subpopulation size. Section~\ref{sec:diagnostics_bounds} develops the corresponding diagnostics. Section \ref{sec:examples2} applies our results to OLS, 2SLS, and TWFE estimands. Sections \ref{sec:estimation} and \ref{sec:empirical} discuss estimation and inference and present our empirical applications. Section \ref{sec:conclusion} concludes. The appendices contain our proofs and additional results.

We use $\supp(\cdot)$ to denote the support of a random vector and $\text{conv}(\cdot)$ to denote the convex hull of a set. For brevity, we denote some expectations, probabilities, variances, and standard deviations under the conditional distribution given $W_0 = 1$ by $\E_{W_0}[\cdot]$, $\P_{W_0}(\cdot)$, $\var_{W_0}(\cdot)$, and $\text{SD}_{W_0}(\cdot)$, respectively.

\section{Motivating Example}\label{sec:example1}

Here we provide further discussion of the OLS estimand. We postpone the discussion of the 2SLS and TWFE estimands to Section \ref{sec:examples2}. In the initial example, we have a binary treatment $D \in \{0,1\}$, potential outcomes $(Y(1),Y(0))$, covariate vector $X$, and realized outcome $Y = Y(D)$. We make the following assumption.

\begin{assumption}[Unconfoundedness]\label{assn:CIA}\hfill
\begin{enumerate}
	\item Selection on observables: $(Y(1),Y(0)) \indep D\mid X$;
	\item Overlap: $p(X) \coloneqq \P(D = 1 \mid X) \in (0,1)$ almost surely.
\end{enumerate}
\end{assumption}

\noindent
Following \cite{Angrist1998}, we can establish that $\beta_\text{OLS}$, the coefficient on $D$ in the linear projection of $Y$ on $(1,D,X)$, satisfies the representation in \eqref{eq:weighted_est_def} under a restriction on the propensity score. The following proposition summarizes \citeauthor{Angrist1998}'s \citeyearpar{Angrist1998} result.

\begin{proposition}
\label{prop:angrist1998}
	Suppose Assumption \ref{assn:CIA} holds. Suppose $p(X)$ is linear in $X$\@. Then 
	\begin{align*}
		\beta_\text{OLS} &=  \E\left[\frac{\var(D \mid X)}{\E[\var(D \mid X)]} \cdot \E[Y(1) - Y(0)\mid X]\right]. 
	\end{align*}
\end{proposition}

\noindent
The linearity assumption can be removed by instead regressing $Y$ on $(1,D,h(X))$, where $h(X)$ is a vector of functions of $X$ whose linear span contains $p(X)$\@. Overlap is not required for this representation to hold: the assumption that $\P(p(X) \in (0,1)) > 0$ is sufficient.

Proposition \ref{prop:angrist1998} implies that $\beta_\text{OLS}$ is a weighted estimand that satisfies the representation in \eqref{eq:weighted_est_def}, with $a(X) = \var(D \mid X)/\E[\var(D \mid X)]$ and $\tau_0(X) = \E[Y(1) - Y(0)\mid X]$. Here we implicitly set $W_0 = 1$ with probability 1. Thus, the regression coefficient $\beta_\text{OLS}$ is a weighted average of CATEs with weights proportional to $p(X)(1-p(X))$. Note that $\beta_\text{OLS} = \text{ATE} \coloneqq \E[Y(1) - Y(0)]$ if and only if $a(X)$ and $\tau_0(X)$ are uncorrelated, which is the case, for example, when $p(X)$ or $\tau_0(X)$ are constant.

An alternative representation of $\beta_\text{OLS}$ follows by focusing on the subpopulation of treated units, $D=1$. Let $W_0 = D$, $a(X) = (1-p(X)) / (1 - \E[p(X) \mid D = 1])$, and $\tau_0(X) = \E[Y(1) - Y(0)\mid D=1, X] = \E[Y(1) - Y(0)\mid X]$, which is implied by unconfoundedness. Then, we can write
\begin{align*}
	\beta_\text{OLS} &= \E\left[\frac{1-p(X)}{1 - \E[p(X) \mid D = 1]} \cdot \E[Y(1) - Y(0)\mid D=1, X] \mid D = 1\right].
\end{align*}
Yet another representation can be obtained when focusing on the subpopulation of untreated units by letting $W_0 = 1-D$. We omit details for brevity.

We revisit this example in Section \ref{sec:examples2} after establishing conditions under which weighted estimands admit a causal representation (Section \ref{sec:causal_represent}), identifying the sizes of the resulting implicit subpopulations (Section \ref{sec:quantifying_int_valid}), and discussing the associated diagnostics (Section \ref{sec:diagnostics_bounds}).

\section{Causal Representation of Weighted Estimands}\label{sec:causal_represent}

In this section, we consider a general class of weighted estimands. We derive necessary and sufficient conditions for an estimand in this class to have a causal representation as an average treatment effect over a subpopulation. We provide these conditions under various restrictions on treatment effect heterogeneity, including the case of no restrictions.

\subsection{Preliminaries}

Recall the earlier setting where we let $D \in \{0,1\}$ denote a treatment variable, and let $(Y(0),Y(1))$ denote the corresponding potential outcomes. Let $X \in \supp(X) \subseteq \R^{d_X}$ denote a $d_X$-vector of covariates. We suppose that $(Y(1),Y(0),D,X)$ are drawn from a common population distribution $F_{Y(1),Y(0),D,X}$.

Let $W_0 \in \{0,1\}$ be an indicator variable used to denote a subpopulation $\{W_0 = 1\}$ and let $\tau_0(X) = \E_{W_0}[Y(1) - Y(0)\mid X]$ denote the CATE given $X$ in that subpopulation. For example, this subpopulation can be the entire population by setting $W_0 = 1$ almost surely, in which case $\tau_0$ denotes the usual CATE function. It can also denote the subpopulation of treated units by setting $W_0 = D$\@. In the presence of a binary instrument $Z$, the complier subpopulation is defined by setting $W_0 = \1(D(1) > D(0))$, where $D(1)$ and $D(0)$ are potential treatments. In this case, $\tau_0$ denotes the conditional local average treatment effect.

Note that $\tau_0$ is defined for all values of $X$ such that $w_0(X) = \Prob(W_0=1\mid X) > 0$.\footnote{While $\tau_0(X)$ is only defined when $w_0(X) > 0$, we set $\tau_0(X)w_0(X) = 0$ when $w_0(X) = 0$.} Throughout this paper, we assume $\P(W_0 = 1) > 0$, so that this subpopulation has a positive mass, which avoids technical issues associated with conditioning on zero-probability events.

Recall the weighted estimands of equation \eqref{eq:weighted_est_def}:
\begin{align*}
	\mu(a,\tau_0) &=\E_{W_0}[a(X)\tau_0(X)].
\end{align*}
The estimands we consider have the above representation and satisfy the following regularity condition to ensure they are well defined.
\begin{assumption}\label{assn:reg}
Suppose $\E[\tau_0(X)^2]<\infty$ and $\E[a(X)^2]<\infty$. Also suppose $\E_{W_0}[a(X)] = 1$.
\end{assumption}

\noindent
The estimands in \eqref{eq:weighted_est_def} can also be written as a weighted sum when $X$ is discrete or an integral when it is continuous. In the discrete case, let $\supp(X) = \{x_1,\ldots,x_K\}$ and $p_k \coloneqq \Prob(X=x_k) > 0$ for $k = 1,\ldots,K$, and assume $W_0 = 1$ almost surely for simplicity. Then, 
\begin{align}\label{eq:weighted_est_def2}
	\mu(a,\tau_0) &=  \sum_{k=1}^K \omega_k \tau_0(x_k) \qquad \text{ where } \qquad \omega_k = a(x_k)p_k \quad \text{and} \quad \sum_{k=1}^K \omega_k = 1.
\end{align}
Thus, the weights $\omega_k$ are the normalized weights $a(x_k)$ adjusted by the population shares $p_k$. Moreover, $\frac{a(x_k)}{a(x_{k'})} = \frac{\omega_k}{\omega_{k'}} \big/ \frac{p_k}{p_{k'}}$, so $a(x_k)>a(x_{k'})$ means that cell $\{X = x_k\}$ is overweighted by the estimand relative to $\{X = x_{k'}\}$, compared with their population shares. Similar algebra can express the estimand as an integral when $X$ is continuously distributed. We focus on the representation in \eqref{eq:weighted_est_def} since it accommodates discrete, continuous, and mixed covariates.

\subsection{Regular Subpopulations}\label{subsec:subpops}

The first question we address is whether an estimand defined by \eqref{eq:weighted_est_def} can be represented as $\E[Y(1) - Y(0) \mid W^* = 1]$, where $W^* \in \{0,1\}$ and $\{W^* = 1\}$ characterizes a subpopulation of $\{W_0 = 1\}$ in the sense that $\{W^*= 1\} \subseteq \{W_0 = 1\}$  almost surely. 

We impose some structure on this problem by restricting how these subpopulations may be formed. We only consider what we call ``regular subpopulations,'' defined here.
\begin{definition}\label{def:subpop}
	Let $W^* \in \{0,1\}$ such that $\Prob(W^* = 1) > 0$. Say $\{W^* = 1\}$ is a \textit{regular subpopulation} of $\{W_0 = 1\}$ if
\begin{enumerate}
	\item (Inclusion) $\Prob(W_0 = 1\mid W^* = 1) = 1$;
	\item (Conditional independence) $W^* \indep (Y(1),Y(0))\mid X,W_0 = 1$.
\end{enumerate}
\end{definition}
\noindent
For convenience, we will abbreviate this as ``$W^*$ is a regular subpopulation of $W_0$''. We denote the set of regular subpopulations of $W_0$ as 
\begin{align*}
	\text{SP}(W_0) &\coloneqq \{W^* \in \{0,1\}: W^* \text{ is a regular subpopulation of } W_0\}.
\end{align*}
These subpopulations have positive masses and are subsets of $\{W_0 = 1\}$. The other substantive requirement is that they are independent of potential outcomes when conditioning on $X$ and the original subpopulation $W_0 = 1$. While this may seem restrictive, it allows for rich and natural classes of subpopulations. For example, consider the unconfoundedness assumption of Section \ref{sec:example1} and let $W_0$ be the entire population, i.e., $\Prob(W_0 =1) = 1$. In this case, regular subpopulations must satisfy $W^* \indep (Y(1),Y(0))\mid X$, or be unconfounded. Regular subpopulations include the population of all treated (or untreated) individuals, i.e., $W^* = D$ (or $W^* = 1-D$), and any subpopulation characterized by a subset of $\supp(X)$. More generally, they include any subpopulation that can be described through a combination of $(D,X,U)$ where $U$ is independent from $(Y(1),Y(0),D,X)$. For example, a subpopulation described by ``a fraction $q(x)$ of units with covariate $X=x$ for all $x \in \supp(X)$'' can be constructed as $W^* = \1(U \leq q(X))$ where $U \sim \text{Unif}(0,1)$ is independent from $(Y(1),Y(0),X)$. 

The conditional independence requirement rules out subpopulations that directly depend on the potential outcomes such as $W^* = \1(Y(1) \geq Y(0))$, i.e., the subpopulation of those who benefit from treatment. Note that $\Prob(W^*=1\mid X) = \Prob(Y(1) \geq Y(0)\mid X)$ and $\Prob(W^*=1) = \Prob(Y(1) \geq Y(0))$ are not point identified under unconfoundedness. Another way to view this requirement is that regular subpopulations are policy relevant in the sense that we could design a policy that targets a regular subpopulation. Indeed, a policy maker may observe $X$ and can use $U$ to randomly target a fraction of units with specific values of $X$, but does not observe the pair $(Y(1),Y(0))$ for any unit.

Regular subpopulations enjoy a number of useful properties. Two of them are characterized in the following proposition. 

\begin{proposition}[Properties of Regular Subpopulations]\label{prop:subpops_properties}
	Suppose that $W^* \in \text{SP}(W_0)$. Suppose Assumption \ref{assn:reg} holds. Let $\underline{w}^*(x) \coloneqq \Prob_{W_0}(W^* = 1\mid X = x)$. Then,
	\begin{enumerate}
		\item $\E[Y(1) - Y(0)\mid W^* = 1, X] = \E_{W_0}[Y(1) - Y(0)\mid X]$ when $\underline{w}^*(X) > 0$;
		\item $\E[Y(1) - Y(0) \mid W^* = 1] = \mu(\underline{w}^*/\E_{W_0}[\underline{w}^*(X)],\tau_0)$.
	\end{enumerate}
\end{proposition}

\noindent
The first part of this proposition shows that average treatment effects within the original population $W_0$ and regular subpopulation $W^*$ are the same when conditioning on $X$. For example, this holds under unconfoundedness for the subpopulation of treated units, $W^* = D$. The second part of this proposition allows us to write the average treatment effect for $W^* = 1$ using the same functional $\mu(\cdot,\cdot)$ that was used to characterize the class of estimands we analyze. This property will be used when studying the mapping between weighted estimands and average treatment effects for regular subpopulations of $W_0$.

\subsection{Existence of a Causal Representation for Weighted Estimands}

We now consider necessary and sufficient conditions for the weighted estimand $\mu(a,\tau_0)$ to be written as the average treatment effect within a regular subpopulation of $W_0$. As we show, these depend on what is assumed about the function $\tau_0(\cdot) = \E_{W_0}[Y(1) - Y(0)\mid X = \cdot]$. 

For example, if $\tau_0$ is constant in $X$, then any weighted estimand satisfying \eqref{eq:weighted_est_def} equals $\E_{W_0}[Y(1) - Y(0)]$, the average treatment effect in the population $\{W_0 = 1\}$. This is the case even when the weight function $a(X)$ changes sign with $X$. However, if $\tau_0$ is nonconstant, the existence of causal representations will depend on the weight function $a(X)$. Among other cases, we will consider one in which no restrictions are placed on $\tau_0$. In this setting, the existence of a causal representation of $\mu(a,\tau_0)$ will require the sign of $a(X)$ to be constant.

To formalize this, let $\mathcal{T}$ denote a class of functions such that $\tau_0 \in \mathcal{T}$. Recall from Proposition \ref{prop:subpops_properties} that $\E[Y(1) - Y(0) \mid W^* = 1] = \mu(\underline{w}^*/\E_{W_0}[\underline{w}^*(X)],\tau_0)$ and define
\begin{align*}
	\mathcal{W}(a;W_0,\mathcal{T}) &\coloneqq \{W^* \in \text{SP}(W_0): \mu(a,\tau) = \mu(\underline{w}^*/\E_{W_0}[\underline{w}^*(X)],\tau) \text{ for all } \tau \in \mathcal{T}\}.
\end{align*}
This is the set of regular subpopulations of $W_0$ such that $\mu(a,\tau) = \mu(\underline{w}^*/\E_{W_0}[\underline{w}^*(X)],\tau)$ for every CATE function $\tau \in \mathcal{T}$\@. If the set $\mathcal{W}(a;W_0,\mathcal{T})$ is empty, then no regular subpopulation of $W_0$ provides such a representation uniformly in $\tau \in \mathcal{T}$\@. We use this set to formally define a notion of uniform causal representation.

\begin{definition} \label{def:caus_rep} An estimand $\mu(a,\tau)$ has a \textit{causal representation uniformly in} $\tau \in \mathcal{T}$ if
\begin{align*}
	\mathcal{W}(a;W_0,\mathcal{T}) \neq \emptyset.
\end{align*}
\end{definition}

\noindent
We now examine several cases for the set $\mathcal{T}$.

\subsubsection{Existence Uniformly in $\tau_0$}

We begin by considering the largest class of functions in which $\tau_0$ lies: the class of all functions, subject to the moment condition in Assumption \ref{assn:reg} that ensures the existence of $\mu(a,\tau_0)$. We denote this class by
\begin{align*}
	\mathcal{T}_\text{all} &\coloneqq \{\tau: \E[\tau(X)^2] < \infty\}.
\end{align*}
In this function class, we show the existence of a causal representation is equivalent to the weights $a(X)$ being nonnegative. In what follows, let $\amax \coloneqq \sup(\supp(a(X)\mid W_0=1))$ be the essential supremum of $a(X)$ given $W_0=1$.

\begin{theorem}\label{thm:existence1}
	Let $\mu(a,\tau_0)$ be an estimand satisfying equation \eqref{eq:weighted_est_def}. Suppose Assumption~\ref{assn:reg} holds and that $\amax < \infty$.\footnote{The condition $\amax < \infty$ restricts attention to subpopulations with positive mass. It holds in each of our theoretical examples in Sections \ref{sec:example1} and \ref{sec:examples2}, where we study the OLS, 2SLS, and TWFE estimands.} Then, $\mu(a,\tau_0)$ has a causal representation uniformly in $\tau_0 \in \mathcal{T}_\text{all}$ if and only if
\vspace*{-8pt}
\begin{align*}	
	\Prob_{W_0}(a(X) \geq 0) = 1.
\end{align*}
\end{theorem}

\noindent
A uniform (in $\mathcal{T}_\text{all}$) causal representation exists if and only if $a(X)$ is nonnegative when $W_0 = 1$. To see why $a(X)$ must be nonnegative with probability 1, suppose that $\P_{W_0}(a(X) < 0) > 0$ and consider the ``adversarial'' CATE function $\tau^-(X) \coloneqq \1(a(X) < 0)$. This CATE function is nonnegative for all $X$, and implies a positive average effect only for units with negative weights. Yet it yields a strictly negative weighted estimand, $\mu(a,\tau^-) = \E_{W_0}[a(X)\1(a(X) < 0)] < 0$. Thus, $\mu(a,\tau^-) $ cannot be the average treatment effect for any subpopulation of $W_0$, since averaging a nonnegative CATE over any subpopulation cannot yield a negative average.

Conversely, if $a(X) \geq 0$, our proof constructively defines a regular subpopulation $W^*$ whose average treatment effect equals the estimand $\mu(a,\tau_0)$ uniformly in $\tau_0 \in \mathcal{T}_\text{all}$. Let
\begin{align}
	W^* = \1\left(U \leq \frac{a(X)}{\amax}\right)\cdot W_0, \label{eq:rep_subpop_constr}
\end{align}
where $U \sim \text{Unif}(0,1) \indep (Y(1),Y(0),X,W_0)$. This is a regular subpopulation of $W_0$ for which the probability of inclusion, conditional on $X$ and $W_0=1$, is proportional to $a(X)$. We can also interpret $\mu(a,\tau_0)$ as the average effect of an intervention in which units with covariate value $X$ are treated with probability $a(X)/\amax$ given $W_0 = 1$. From this construction, we can see that $\underline{w}^*(X) = \Prob_{W_0}(W^*=1\mid X)$ is proportional to $a(X)$, which implies that $\E[Y(1) - Y(0)\mid W^*=1] = \mu(\underline{w}^*/\E_{W_0}[\underline{w}^*(X)],\tau_0) = \mu(a,\tau_0)$ uniformly in $\tau_0 \in \mathcal{T}_\text{all}$.

As noted earlier, the condition that the weights are nonnegative is well established. \cite{BlandholBonneyMogstadTorgovitsky2026} show that it is equivalent to an estimand being ``weakly causal,'' meaning that it matches the sign of $\tau_0$ whenever that sign is the same across all units. Thus, in the class of weighted estimands we consider, estimands have a causal representation uniformly in $\mathcal{T}_\text{all}$ if and only if they are weakly causal. Appendix \ref{appsec:weaklycausal} formally establishes this connection.

\subsubsection{Existence for a Given $\tau_0$}

We now provide an existence result that requires the causal representation to exist only for the \textit{given} $\tau_0$, rather than uniformly for $\tau_0$ in the larger set $\mathcal{T}_\text{all}$. The following result depends on the CATE function $\tau_0$ in the population, whereas the condition in Theorem \ref{thm:existence1} depended only on the weight function $a(X)$. Thus, the distribution of the potential outcomes will have an impact on the existence of a causal representation given $\tau_0$. Using the notation from Definition \ref{def:caus_rep}, a causal representation exists if and only if $\mathcal{W}(a;W_0,\{\tau_0\}) \neq \emptyset$.

\begin{theorem}\label{thm:existence2}
Let $\mu(a,\tau_0)$ be an estimand satisfying equation \eqref{eq:weighted_est_def}. Suppose Assumption~\ref{assn:reg} holds.  Then, $\mu(a,\tau_0)$ has a causal representation for $\tau_0$ \mbox{if and only if}
\begin{align*}
	\mu(a,\tau_0) \in \mathcal{S}(\tau_0;W_0)\coloneqq  \{t \in \R: \; &\Prob_{W_0}(\tau_0(X) \leq t) > 0 \text{ and } \Prob_{W_0}(\tau_0(X)\geq t) > 0\}.
\end{align*}
\end{theorem}

\noindent
The existence condition in Theorem \ref{thm:existence2} is weaker than the one in Theorem \ref{thm:existence1} since we no longer require this representation to be valid for any CATE function, but rather just for the one that is identified from the population. The necessary and sufficient condition in this theorem only requires that the estimand is in the convex hull of the support of the CATEs. This means $\mu(a,\tau_0)$ has a causal representation even with negative weights, as long as there are CATEs smaller and greater than $\mu(a,\tau_0)$. We can see this support condition holds for all $\tau_0 \in \mathcal{T}_\text{all}$ if and only if $\mu(a,\tau_0)$ is between the upper and lower support points of $\tau_0(X)$ for any $\tau_0 \in \mathcal{T}_\text{all}$. This is precisely the case when the weights $a(X)$ are nonnegative since it guarantees $\inf(\supp(\tau_0(X)\mid W_0=1)) \leq \mu(a,\tau_0) \leq \sup(\supp(\tau_0(X)\mid W_0=1))$ for any $\tau_0$.

\subsubsection{Existence in Intermediate Cases}
\label{remark:intermediate_cases}

Analyzing the causal representation of an estimand under no restrictions on $\tau_0$ could be viewed as unnecessarily conservative in some settings. At the other extreme, assuming knowledge of $\tau_0$ may be unrealistic, especially when $X$ has many components, making $\tau_0$ more difficult to estimate. For example, some shape constraints may be known to hold for $\tau_0$. In some economic applications one may posit that $\tau_0$ is monotonic or convex in some components of $X$, or positive/negative over a subset of $\supp(X\mid W_0=1)$. In these cases, a causal representation may exist under weaker conditions than those in Theorem \ref{thm:existence1}, but stronger than those in Theorem \ref{thm:existence2}. In particular, one may be able to relax the requirement that $a(X) \geq 0$ without assuming that $\tau_0$ be given or fixed. The following proposition shows that this is the case when $\tau_0(X)$ is assumed to be linear in $X$. 

\begin{proposition}\label{prop:linear_CATEs}
Let $\mu(a,\tau_0)$ be an estimand satisfying equation \eqref{eq:weighted_est_def}. Suppose $X$ has finite support. Suppose Assumption \ref{assn:reg} holds and define 
\begin{align*}
	\mathcal{T}_\text{lin} &\coloneqq \{\tau \in \mathcal{T}_\text{all}: \tau(X) = c + d'X, \; (c,d) \in \R^{1 + d_X}\}.
\end{align*} 
Then, $\mu(a,\tau_0)$ has a causal representation uniformly in $\tau_0 \in\mathcal{T}_\text{lin}$ if and only if 
\begin{align*}
	\E_{W_0}[a(X)X] \in \text{conv}(\supp(X\mid W_0=1)).
\end{align*}
\end{proposition}

\noindent
This proposition illustrates that placing restrictions on $\mathcal{T}$ may obviate the requirement that $a(X) \geq 0$ for a uniform causal representation of an estimand. In particular, the requirement here is that $\E_{W_0}[a(X)X]$ lies in the convex hull of the support of $X$ given $W_0=1$. For scalar $X$, this is an interval. This condition does not require $a(X)$ to be nonnegative. For example, if $\supp(X) = \{0,1,2\}$ and $W_0 = 1$ almost surely, any values of $(a(0),a(1),a(2))$ such that $\E[a(X)X] \in [0,2]$ and $\E[a(X)] = 1$ imply a causal representation. Let $\Prob(X = x) = 1/3$ for $x \in \{0,1,2\}$ and $(a(0),a(1),a(2)) = (2,-1,2)$. Here units with $X=1$ have negative weights, but $\E[a(X)] = \E[X a(X)] = 1 \in [0,2]$, implying that the estimand has a causal representation uniformly in $\tau_0 \in\mathcal{T}_\text{lin}$. The result is stated for discrete $X$, but an estimand with negative weights can have a causal representation even when $X$ has continuous components.

We consider another class of CATE functions that restricts their heterogeneity. For $\Delta \geq 0$, let
\vspace*{-10pt}
\begin{align*}
	\mathcal{T}_\text{BD}(\Delta) \coloneqq \left\{\tau \in \mathcal{T}_\text{all} : \sup_{x,x' \in \supp(X \mid W_0=1)}|\tau(x) - \tau(x')| \leq \Delta\right\}.
\end{align*}
This function class uniformly bounds differences of the CATE function. When $\Delta=0$, the CATE function is constant, and thus equal to $\E_{W_0}[Y(1) - Y(0)]$. When $\Delta > 0$, CATEs may differ in value, but the maximum discrepancy between any two CATEs is bounded above by $\Delta$\@. We show that restricting the CATEs to satisfy this bounded difference assumption does not remove the requirement that $a(X)$ be nonnegative, unless $\Delta = 0$, in which case all $a(\cdot)$ functions yield a causal representation uniformly in $\mathcal{T}_\text{BD}(0)$.

\begin{proposition}\label{prop:BV_CATEs}
Let $\mu(a,\tau_0)$ be an estimand satisfying equation \eqref{eq:weighted_est_def}. Suppose Assumption~\ref{assn:reg} holds. Then, $\mu(a,\tau_0)$ has a causal representation uniformly in $\mathcal{T}_\text{BD}(\Delta)$ when $\Delta > 0$ if and only if 
\begin{align*}
	\Prob_{W_0}(a(X) \geq 0) = 1.
\end{align*}
The estimand $\mu(a,\tau_0)$ has a causal representation uniformly in $\mathcal{T}_\text{BD}(0)$ for any $a(\cdot)$.
\end{proposition}

\noindent
For intuition, consider the adversarial CATE function $\tau^-(X) = \Delta \cdot \1(a(X) < 0)$, a member of $\mathcal{T}_\text{BD}(\Delta)$, and assume $\P_{W_0}(a(X) \geq 0) < 1$. Then we obtain the same contradiction we discussed after Theorem \ref{thm:existence1}, where $\tau^-$ is nonnegative everywhere but $\mu(a,\tau^-)$ is negative.

The last two propositions show that the impact of restrictions on $\tau_0$ on the requirement that $a(X)$ be nonnegative critically depends on the nature of these restrictions. Generalizations to additional or empirically motivated function classes are left for future work.

\section{Quantifying the Internal Validity of Weighted Estimands}\label{sec:quantifying_int_valid}

Many estimands will admit causal representations, but their implicit subpopulations $\{W^* = 1\}$ will differ. Also, a weighted estimand may not correspond to the \textit{target estimand} a researcher is interested in. For example, a researcher may take $\E_{W_0}[Y(1) - Y(0)]$, the average effect in the population $\{W_0 = 1\}$, as their target. In general, this parameter differs from $\mu(a,\tau_0)$.

However, the set of subpopulations corresponding to a weighted estimand can be used to understand how representative the weighted estimand is of the target. For example, we may seek estimands for which $\Prob_{W_0}(W^*=1)$ attains values closest to 1, since they have a higher degree of internal validity with respect to the target $\E_{W_0}[Y(1) - Y(0)]$. At one extreme, a weighted estimand for which $\Prob_{W_0}(W^*=1) =1$ would be deemed to have the highest degree of internal validity for the target estimand since it would be equal to it. We formalize this interpretation through a measure of internal validity (MIV), which we define here.

\begin{definition}[Measure of Internal Validity] \label{def:internal_validity}
Let
\begin{align*}
	\text{MIV}(a,W_0;\mathcal{T}) \coloneqq \sup_{W^* \in \mathcal{W}(a;W_0,\mathcal{T})} \Prob_{W_0}(W^* = 1)
\end{align*}
denote the MIV of the weighted estimand $\mu(a,\tau_0)$ over the function class $\mathcal{T}$.
\end{definition}

\noindent
Formally, $\text{MIV}(a,W_0;\mathcal{T})$ is the sharp upper bound on $\Prob_{W_0}(W^* = 1)$ for any regular subpopulation $W^*$ of $W_0$ such that the weighted estimand $\mu(a,\tau_0)$ has a causal representation as the average treatment effect over the subpopulation $W^*$. Note that we set $\text{MIV}(a,W_0;\mathcal{T}) = 0$ when $\mathcal{W}(a;W_0,\mathcal{T})$ is empty. This object depends on the chosen function class $\mathcal{T}$, as did Theorems \ref{thm:existence1} and \ref{thm:existence2} in the previous section. Given the above terminology and assuming that $\E_{W_0}[Y(1) - Y(0)]$ is the target, we call $\text{MIV}(a,W_0;\mathcal{T})$ the MIV of the estimand $\mu(a,\tau_0)$, and we use this definition in the remainder of the paper.

We can also compute the maximum value of $\Prob(W^* = 1)$ across $W^* \in  \mathcal{W}(a;W_0,\mathcal{T})$, which measures the largest share of the entire population for which the weighted estimand has a causal representation. We refer to this measure as a measure of representativeness (MR)\@.

\begin{definition}[Measure of Representativeness] \label{def:representativeness}
Let
\begin{align*}
	\text{MR}(a,W_0;\mathcal{T}) \coloneqq \sup_{W^* \in \mathcal{W}(a;W_0,\mathcal{T})} \Prob(W^* = 1)
\end{align*}	
denote the MR of the weighted estimand $\mu(a,\tau_0)$ over the function class $\mathcal{T}$.
\end{definition}

\noindent
The MR gives the internal validity of $\mu(a,\tau_0)$ with respect to the target estimand $\E[Y(1) - Y(0)]$, the average treatment effect in the population from which the sample is drawn. Our measures of internal validity and representativeness are closely linked and a subpopulation will maximize $\P(W^* = 1)$ if and only if it maximizes $\P_{W_0}(W^* = 1)$. We can also see that $\text{MR}(a,W_0;\mathcal{T}) = \text{MIV}(a,W_0;\mathcal{T}) \cdot \Prob(W_0=1)$, since $W^*$ is a subpopulation of $W_0$. The MIV and MR are identical when $W_0 = 1$ almost surely. In Section \ref{sec:diagnostics_bounds}, we use these measures to obtain simple bounds on target estimands.

We now derive explicit expressions for $\text{MIV}(a,W_0;\mathcal{T})$. We focus on two cases, the first being when $\tau_0$ is unrestricted.

\subsection{Quantifying Internal Validity Uniformly in $\tau_0$}

Without imposing any restrictions on the CATE function, except for the existence of second moments, the MIV of a weighted estimand is given by the following theorem.

\begin{theorem}\label{thm:unc_upperbounds1}
Let $\mu(a,\tau_0)$ be an estimand satisfying equation \eqref{eq:weighted_est_def}. Suppose Assumption~\ref{assn:reg} holds and that $\amax < \infty$. If $\Prob_{W_0}(a(X) \geq 0) = 1$, then 
\begin{align*}
	\text{MIV}(a,W_0;\mathcal{T}_\text{all}) = \amax^{-1}. 
\end{align*}
If $\Prob_{W_0}(a(X) \geq 0) < 1$, then $\text{MIV}(a,W_0;\mathcal{T}_\text{all}) = 0$.
\end{theorem}

\noindent
The maximum size of a subpopulation characterizing the weighted estimand $\mu(a,\tau_0)$ depends on $a(X)$ only through its supremum $\amax$. This bound can be computed at what \cite{ImbensRubin2015} call the ``design stage'' of the study, that is, without using outcome data.

For intuition about the role of the supremum in this expression, recall that $\mu(a,\tau_0) = \E[Y(1) - Y(0)\mid W^* = 1]$ is equivalent to writing
\begin{align}
	\mu(a,\tau_0) &= \mu(\underline{w}^*/\E_{W_0}[\underline{w}^*(X)],\tau_0)\label{eq:mu(a)=mu(wstar)}
\end{align}
for all $\tau_0 \in \mathcal{T}_{\text{all}}$. For equation \eqref{eq:mu(a)=mu(wstar)} to hold for all $\tau_0$, $\underline{w}^*(X)$ must be proportional to $a(X)$. While the range of $a(X)$ is unconstrained, $\underline{w}^*(X)$ must lie in $[0,1]$ to be a conditional probability. Since we seek to maximize $\Prob_{W_0}(W^*=1) = \E_{W_0}[\underline{w}^*(X)]$, we let $\underline{w}^*(X) = a(X)/\amax$, which is the largest multiple of $a(X)$ that lies in $[0,1]$ with probability 1. 

We can define a corresponding regular subpopulation $W^*$ as in equation \eqref{eq:rep_subpop_constr}. The population $\{W^* = 1\}$ contains a random subset of $\{W_0 = 1\}$ where the probability of inclusion is proportional to $a(X)$. Thus, units with larger values of $a(X)$ are more likely to be included in $W^*$. All units in $\{W_0 = 1\}$ with $X$ such that $a(X) = \amax$ are included in $W^*$, whereas no units where $a(X) = 0$ are included. 

The construction of this subpopulation is illustrated in Figure \ref{fig:unif} for continuous $x$, omitting the conditioning on $W_0 = 1$ for simplicity. We seek to maximize $\P(W^* = 1) = \int \underline{w}^*(x) f_X(x) \, dx$ with the requirement that $\underline{w}^*(x) \leq 1$ (or, equivalently, $\underline{w}^*(x)f_X(x) \leq f_X(x)$) and that $\underline{w}^*(x)$ is a multiple of $a(x)$. In the figure, we see that $\amax > 1$ and thus the largest multiple of $a(x)$ that is weakly smaller than 1 is illustrated by the gray curve. The area under this curve is precisely $\P(W^* = 1)$. Note that the area under $f_X(x)$ is one, so closer alignment of the gray curve and the density $f_X(x)$ corresponds to more representative estimands.

\begin{figure}[!tb]
\begin{adjustwidth}{-1in}{-1in}
\centering
\caption{Characterizing an Implicit Subpopulation Uniformly in $\mathcal{T}_{\text{all}}$}\label{fig:unif}
\includegraphics[width=15cm]{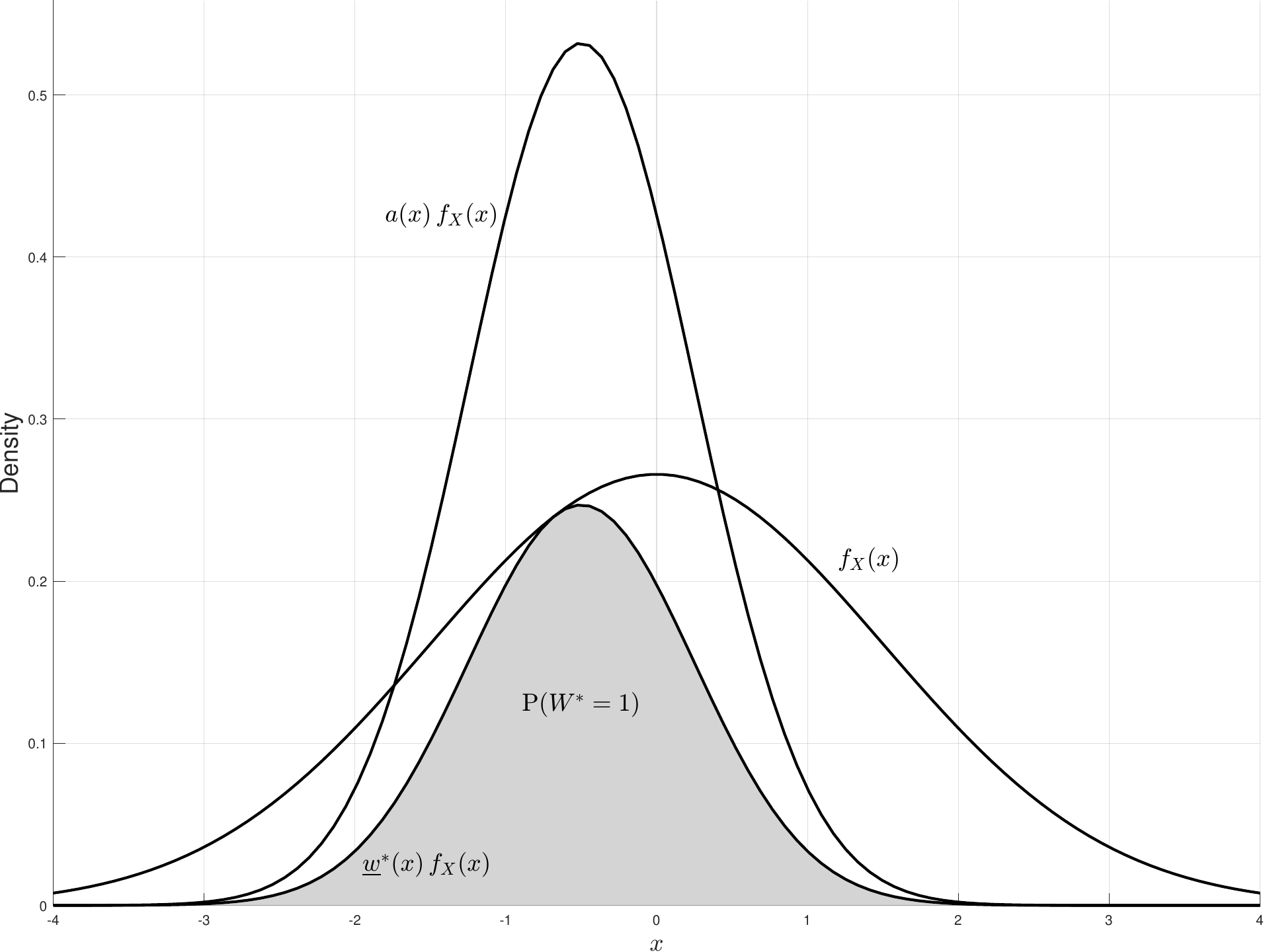}
\begin{footnotesize}
\begin{tabular}{p{15cm}}
\textit{Note:} $X$ is a single continuously distributed covariate with density $f_X$.
\end{tabular}
\end{footnotesize}
\end{adjustwidth}
\end{figure}

Several further comments about Theorem \ref{thm:unc_upperbounds1} are in order. 

\begin{remark}[Estimands and Their Corresponding Interventions]
We can link the estimand $\mu(a,\tau_0)$ to an intervention in which a fraction $a(X)/\amax \in [0,1]$ of units with covariate value $X$ and $W_0=1$ are treated. This link can be seen through
\begin{align*}
	\mu(a,\tau_0) = \E_{W_0}[a(X) \tau_0(X)] = \E_{W_0}\left[\frac{a(X)/\amax}{\E_{W_0}[a(X)]/\amax} \cdot \tau_0(X)\right].
\end{align*}
For example, if $W_0=1$ and $a(X) = \amax$ almost surely, then $\mu(a,\tau_0)$ is the ATE, the average effect of treatment among all units. Under unconfoundedness, we also note that the average treatment effect on the treated (ATT) can be written as $\E[Y(1) - Y(0)\mid D=1] = \mu(a,\tau_0)$, a weighted estimand with weights $a(X) = \P(D=1\mid X)/\P(D=1)$. This means that it can be interpreted either as the effect of an intervention where a fraction $\P(D=1\mid X)$ of units with covariate $X$ are treated or as the effect of an intervention where all treated units are treated. In our setting, we can interpret any weighted estimand with nonnegative weights as the effect of treatment for a feasible intervention defined only in terms of $X$, $W_0$, and independent noise $U \sim \text{Unif}(0,1)$ via $W^* = \1(U \leq a(X)/\amax)\cdot W_0$.
\end{remark}

\begin{remark}[Uniqueness of Implicit Subpopulations]
The subpopulation maximizing the level of internal validity is generally not unique. The population $W^* =  \1\left(U \leq a(X)/\amax\right)\cdot W_0$ will generally change if $U$ is replaced by another draw from a uniform distribution. The probability (conditional on $X$) of any unit being part of $W^*$ does not change with the draw of $U$, but whether any given unit is included in subpopulation $\{W^* = 1\}$ cannot be determined.
\end{remark}

\begin{remark}[Meaning of Internal Validity and Representativeness]
Suppose we consider $\E[Y(1) - Y(0) \mid W^* = 1]$ to be the target parameter, where $W^* = \1\left(U \leq a(X)/\amax \right)\cdot W_0$ is the subpopulation for which $\mu(a,\tau_0) = \E[Y(1) - Y(0) \mid W^* = 1]$ uniformly in $\tau_0 \in \mathcal{T}_\text{all}$. For example, in the case of the OLS estimand in Section \ref{sec:example1}, this is a subpopulation with inclusion probability proportional to the conditional variance of treatment. If this subpopulation is the target, it would be reasonable to infer that the MIV of the estimand $\mu(a,\tau_0)$ is 1. Indeed, this is the case because the estimand can be written as $\mu(a,\tau_0) = \E[Y(1) - Y(0) \mid W^* = 1]$, and $\{W^* = 1\}$ is trivially the largest regular subpopulation of itself. This example illustrates how the MIV of $\mu(a,\tau_0)$ depends entirely on the target estimand. This is appropriate because we associate the term ``internal validity'' with whether the probability limit of an estimator equals the parameter of interest. Thus, the MIV of $\mu(a,\tau_0)$ should naturally depend on the parameter choice. Meanwhile, the MR of $\mu(a,\tau_0)$ targets the ATE\@. It will be less than one unless $W^* = 1$ almost surely, so that $\mu(a,\tau_0)$ actually equals the ATE\@. If we also wish to operationalize external validity in our framework, we need to consider whether $\mu(a,\tau_0)$ can be written as the average treatment effect within a subpopulation of another, possibly arbitrary population. We leave this extension for future work.
\end{remark}

\noindent
We now consider a simple example to give further intuition for Theorem \ref{thm:unc_upperbounds1}.

\subsubsection*{Illustrative Example: A Single Binary Covariate}
Consider an estimand $\mu(a,\tau_0)$ where $W_0 = 1$ almost surely, $a(X) \geq 0$, and where $X$ is binary with support $\supp(X) = \{1,2\}$. Let $p_x \coloneqq \Prob(X=x)$ for $x \in \{1,2\}$. Equation \eqref{eq:weighted_est_def2} and the definition of the ATE yield
\begin{align*}
    \mu(a,\tau_0)
        &= a(1)p_1\tau_0(1)+a(2)p_2\tau_0(2)  \qquad \text{ and } \qquad
    \text{ATE}
        &= p_1\tau_0(1)+p_2\tau_0(2).
\end{align*}
If $a(1) = a(2) = 1$, the estimand equals the ATE and thus clearly has the maximum degree of internal validity with respect to the ATE\@. Applying Theorem \ref{thm:unc_upperbounds1}, we can directly see that $\amax = 1$ and thus $\text{MIV}(a,1;\mathcal{T}_\text{all}) = 1$.

However, when $a(1) \neq a(2)$, the weights used by the estimand differ from $(p_1,p_2)$, the population weights for the two covariate cells. For concreteness, let $(p_1,p_2) = (0.2,0.8)$, and let $(a(1),a(2)) = (2,0.75)$, which correspond to the OLS weights of Proposition \ref{prop:angrist1998} when the propensity score is $(p(1),p(2)) = (0.4,0.1)$. In this case
\begin{align*}
	\mu(a,\tau_0) &= 0.4 \cdot \tau_0(1) + 0.6 \cdot \tau_0(2) \qquad \text{ and } \qquad \text{ATE} &= 0.2 \cdot \tau_0(1) + 0.8 \cdot \tau_0(2).
\end{align*}
Relative to the ATE, the weighted estimand overrepresents the subpopulation with $X=1$ and underrepresents the subpopulation with $X=2$\@. The largest subpopulation $\{W^* = 1\}$ that causally represents the estimand can be constructed by combining subsets of the subpopulations defined by $\{X=1\}$ and $\{X=2\}$. Specifically, let 
\begin{align*}
W^* = \1(X=1) + \1\left(U \leq \frac{a(2)}{a(1)},X=2\right) = \1(X=1) + \1\left(U \leq \frac{3}{8},X=2\right),
\end{align*}
where $U \sim \text{Unif}(0,1)$ is independent of $(Y(1),Y(0),X)$. This regular subpopulation contains all units with $X=1$ and a uniformly random three eighths of units with $X=2$. Therefore
\begin{align*}
	\Prob(W^*=1 \mid X = 1) = \underline{w}^*(1) = 1 \qquad \text{ and } \qquad \P(W^* = 1 \mid X=2) = \underline{w}^*(2) = 3/8,
\end{align*}
which yields $\P(W^* = 1) = p_1 \underline{w}^*(1) + p_2 \underline{w}^*(2) = 0.5$. The same quantity can be obtained from Theorem \ref{thm:unc_upperbounds1}, which implies that $\text{MIV}(a,W_0;\mathcal{T}_\text{all}) = 1 / \left( \sup_{x \in \{1,2\}} a(x) \right) = 1/2$. The average effect in this subpopulation is given by
\begingroup
\allowdisplaybreaks
\begin{align*}
	\E[Y(1) - Y(0) \mid W^* = 1] &= \E_{W_0}[ \tau_0(X) \mid W^* = 1]\\
	&= \left(1 \cdot \tau_0(1) \cdot 0.2 + 3/8 \cdot \tau_0(2)\cdot 0.8 \right)/(1/2)\\
	&= 0.4 \cdot \tau_0(1) + 0.6 \cdot \tau_0(2),
\end{align*}
\endgroup
which equals $\mu(a,\tau_0)$ for any $\tau_0$. The relative weights on $\{X=1\}$ and $\{X=2\}$ in subpopulation $\{W^* = 1\}$ are $\frac{\Prob_{W_0}(X=1 \mid W^* = 1)}{\Prob_{W_0}(X=2 \mid W^* = 1)} = 0.4/0.6$, matching the ratio of the weights assigned by the estimand. The subpopulation $\{W^* = 1\}$ cannot expand while preserving this ratio since it already includes all units with $X=1$. Therefore, $W^*$ is the largest subpopulation for which $\mu(a,\tau_0) = \E[Y(1) - Y(0) \mid W^* = 1]$ for any $\tau_0$.

\subsection{Quantifying Internal Validity Given $\tau_0$}

We can also ask how internally valid a weighted estimand can be, given knowledge of the CATE function. In this case, the object of interest is
\begin{align}\label{eq:HC_optimum}
	\text{MIV}(a,W_0;\{\tau_0\}) = \sup_{W^* \in \mathcal{W}(a;W_0,\{\tau_0\})} \Prob_{W_0}(W^* = 1),
\end{align}
where $\tau_0$ is a given CATE function. Since $\tau_0$ is known, the condition $W^* \in \mathcal{W}(a;W_0,\{\tau_0\})$ can be written as $\mu_0 = \mu(\underline{w}^*/\E_{W_0}[\underline{w}^*(X)],\tau_0)$, where we let $\mu_0 \coloneqq \mu(a,\tau_0)$ to simplify the notation. This condition is equivalent to $\E_{W_0}[(\tau_0(X)-\mu_0) \underline{w}^*(X)] = 0$, a linear constraint on $\underline{w}^*$. Additionally, the objective function $\Prob_{W_0}(W^*=1) = \E_{W_0}[\underline{w}^*(X)]$ is linear in $\underline{w}^*$. Thus, the optimization in \eqref{eq:HC_optimum} can be cast as a linear program. To see this, consider as an example the case where $X$ has finite support, i.e., $\supp(X \mid W_0 = 1) = \{x_1,\ldots,x_K\}$. Let $f_k \coloneqq \Prob_{W_0}(W^*=1, X = x_k)$ and note that $f_k \in [0,p_k]$ where $p_k = \Prob_{W_0}(X=x_k)$.

We can write the above optimization problem as
\begin{align*}
	\max_{(f_1,\ldots,f_K) \geq \mathbf{0}_K} \sum_{k=1}^K f_k \; \text{ s.t. } f_k \leq p_k \text{ for } k \in \{1,\ldots,K\} \; \text{ and } \; \sum_{k=1}^K (\tau_0(x_k) - \mu_0)f_k = 0,
\end{align*}
a finite-dimensional linear program. This program has a feasible solution if $\tau_0(x_k) - \mu_0$ is not strictly positive or strictly negative for all $k$, meaning that $\mu(a,\tau_0)$ lies in the convex hull of CATE values, which is precisely the condition stated in Theorem \ref{thm:existence2}. While there exist many methods for solving linear programs, the value function can be obtained through an algorithm that is simple to describe analytically. 

\begin{algorithm}[Internal Validity for Fixed $\tau_0$] Without loss of generality, let $\tau_0(x_1) \leq \ldots \leq \tau_0(x_K)$.

\begin{enumerate}
	\item Set $(f_1,\ldots,f_K) = (p_1,\ldots,p_K)$.
	\item If $\sum_{k=1}^K (\tau_0(x_k) - \mu_0)f_k = 0$, end the algorithm and report $\sum_{k=1}^K f_k$.
	\item If $\sum_{k=1}^K (\tau_0(x_k) - \mu_0)f_k \neq 0$:
	\begin{enumerate}
		\item If $\sum_{k=1}^K (\tau_0(x_k) - \mu_0)f_k > 0$, let $k^* = \max \{k \in \{1,\ldots,K\}: f_k = p_k\}$ and set $f_{k^*} = \max\left\{0,-\sum_{k=1}^{k^*-1} (\tau_0(x_k) - \mu_0)p_k/(\tau_0(x_{k^*}) - \mu_0)\right\}$.
		\item If $\sum_{k=1}^K (\tau_0(x_k) - \mu_0)f_k < 0$, let $k^* = \min \{k \in \{1,\ldots,K\}: f_k = p_k\}$ and set $f_{k^*} = \max\left\{0,-\sum_{k=k^* + 1}^{K} (\tau_0(x_k) - \mu_0)p_k/(\tau_0(x_{k^*}) - \mu_0)\right\}$.
	\end{enumerate}
	\item Go to step 2.
\end{enumerate}
\end{algorithm}

\noindent
When $\mu_0$ exceeds $\E_{W_0}[Y(1) - Y(0)]$, this algorithm reduces the weights associated with smallest CATEs until $\mu_0$ equals $\E[Y(1) - Y(0) \mid W^* = 1]$ for some subpopulation. When $\mu_0 < \E_{W_0}[Y(1) - Y(0)]$, the same procedure is instead applied to the largest CATEs. The support assumption of Theorem \ref{thm:existence2} guarantees that this algorithm ends.

When $X$ is not discretely supported, the problem can still be cast as a linear program, but its dimension is infinite, which leads to implementation difficulties. However, we show this program has an exact, analytical solution even when the components of $X$ are allowed to be discrete, continuous, and mixed, as is often the case in practice.

\begin{theorem}\label{thm:unc_upperbounds2}
	Let $\mu(a,\tau_0)$ be an estimand satisfying equation \eqref{eq:weighted_est_def}. Suppose Assumption~\ref{assn:reg} holds. If $\mu_0 \notin \mathcal{S}(\tau_0;W_0)$, then  $\text{MIV}(a,W_0;\{\tau_0\}) = 0$. If $\mu_0 \in \mathcal{S}(\tau_0;W_0)$, then
\begin{align} \label{eq:HC_Uppbound}
	\text{MIV}(a,W_0;\{\tau_0\}) =& \begin{cases}
			\Prob_{W_0}(T_\mu \leq \alpha^+)  - \frac{\E_{W_0}[T_\mu \1(T_\mu \leq \alpha^+)]}{\alpha^+}\1(\alpha^+ > 0) &\text{ if } \mu_0 < E_0 \\
			\Prob_{W_0}(T_\mu \geq \alpha^-)  - \frac{\E_{W_0}[T_\mu \1(T_\mu \geq \alpha^-)]}{\alpha^-}\1(\alpha^- < 0)
			  &\text{ if }\mu_0 > E_0 \\
			1 &\text{ if } \mu_0 = E_0,
			\end{cases}
\end{align}
where $T_\mu \coloneqq \tau_0(X) - \mu_0$, $E_0 \coloneqq \E_{W_0}[Y(1) - Y(0)]$, $\alpha^+ \coloneqq \inf\{\alpha \geq 0: \E_{W_0}[T_\mu \1(T_\mu \leq \alpha)] \geq 0\}$, and $\alpha^- \coloneqq \sup\{\alpha \leq 0: \E_{W_0}[T_\mu \1(T_\mu \geq \alpha)] \leq 0\}$.  
\end{theorem}

\noindent
The computation of these bounds can be done using a linear programming algorithm when $X$ is discrete, or through plug-in estimators of the terms in equation \eqref{eq:HC_Uppbound} regardless of the nature of the support of $X$.

To illustrate this theorem, let $\tau_0(X)$ be continuously distributed with support $[\underline{\tau},\overline{\tau}]$, and suppose $\mu_0 \in (\underline{\tau},\overline{\tau})$. Without loss of generality, assume $E_0 \geq \mu_0$. If $E_0 = \mu_0$, then the estimand is perfectly representative of the population since it equals the average treatment effect over it. In the case where $E_0 > \mu_0$, the estimand is not representative of the entire population. 
We are searching for the largest subpopulation $\{W^* = 1\}$ such that $\E_{W_0}[Y(1) - Y(0)\mid W^*=1] = \mu_0$. Initializing $W^*$ at $W_0$, removing the subpopulation with the largest values of $\tau_0(x)$ yields the steepest decrease in $\E_{W_0}[Y(1) - Y(0)\mid W^* = 1]$. Therefore, $\text{MIV}(a,W_0;\{\tau_0\})$ is obtained by removing a subpopulation of the kind $W^-(\alpha) \coloneqq \1(\tau_0(X) > \alpha) \cdot W_0$ for a given threshold $\alpha$. This threshold is determined by solving the constraint 
\begin{align}\label{eq:example_HC_bound}
	\E_{W_0}[\tau_0(X)\mid \tau_0(X) \leq \alpha] = \mu_0.
\end{align}
Thus, the remaining subpopulation $W^*$ corresponds to $W^* = (1-W^-(\alpha^*)) \cdot W_0$ where $\alpha^*$ is the unique solution to \eqref{eq:example_HC_bound}.  In this setting, the value $\text{MIV}(a,W_0;\{\tau_0\})$ is larger when the truncated subpopulations are smaller. In particular, this is the case when there are a few units with extreme values of $\tau_0$ whose removal has a large impact on the estimand, but a small impact on the share of the population.

\section{Diagnostics and Bounds}\label{sec:diagnostics_bounds}

We now focus on the measure $\amax^{-1}$, the MIV under unrestricted heterogeneity, and use it as a diagnostic. Beyond indicating the size of the implicit subpopulation, it yields simple bounds on target parameters. These bounds remain valid even with negative weights, making the measure informative even when the estimand lacks a causal representation. We next show that the fraction of negative weights is bounded above by $1 - \amax^{-1}$. We also describe how the characteristics of the subpopulation $W^*$ can be recovered using the weight function and suggest reporting them as additional diagnostics in empirical applications.

\subsection{Using the MIV to Bound Average Effects}\label{subsec:bounds}

We first show how to bound average treatment effects using the weighted estimand and our diagnostic $\amax^{-1}$. Consider a scenario where only the weighted estimand $\mu(a,\tau_0)$ and its maximum weight $\amax$ are known. For example, this could be the case if a researcher uses a weighted estimand (e.g., OLS) and reports the measure we propose in Definition \ref{def:internal_validity} to quantify its degree of internal validity. We can decompose the target estimand, $\E_{W_0}[Y(1) - Y(0)]$, as
\begin{align}
	\E_{W_0}[Y(1) - Y(0)] \; &= \; \E_{W_0}[(a(X)\amax^{-1} + (1 - a(X)\amax^{-1}))\tau_0(X)]\notag\\
	&= \; \mu(a,\tau_0) \amax^{-1} \; + \; \E_{W_0}[(1 - a(X)\amax^{-1})\tau_0(X)].\label{eq:bounds_eq1}
\end{align}
The first term in equation \eqref{eq:bounds_eq1} is determined from the estimand and the diagnostic $\amax^{-1}$, while the second term depends on the CATE\@. Now suppose we have knowledge of bounds on the CATE in the sense that $\tau_0(X)$ lies in $[\underline{\tau},\overline{\tau}]$, a known interval, with probability 1. This knowledge could be derived from support restrictions on the outcomes: for example, we could set $[\underline{\tau},\overline{\tau}] = [-1,1]$ when outcomes are binary.  Alternatively, $(\underline{\tau},\overline{\tau})$ could be viewed as sensitivity parameters. 

With this knowledge, we can bound the second term in \eqref{eq:bounds_eq1} and obtain a simple, sharp bound for $\E_{W_0}[Y(1) - Y(0)]$.

\begin{proposition}[Bounds on Target Estimand]\label{prop:ATE_bounds_amax}
Suppose the weights $a(X)$ satisfy Assumption~\ref{assn:reg}. Suppose $\mu(a,\tau_0)$ and $\amax^{-1}$ are known, and that $\tau_0(X) \in [\underline{\tau},\overline{\tau}]$ a.s., where $[\underline{\tau},\overline{\tau}]$ is a known interval containing $\mu(a,\tau_0)$. Then, 
\begin{align}\label{eq:bounds1}
	\E_{W_0}[Y(1) - Y(0)] \; \in \; \left[\mu(a,\tau_0)\amax^{-1} \, + \, \underline{\tau}(1 - \amax^{-1}), \; \mu(a,\tau_0)\amax^{-1} \, + \, \overline{\tau}(1 - \amax^{-1})\right],
\end{align}
and these bounds are sharp.
\end{proposition}

\noindent
The bounds in \eqref{eq:bounds1} collapse to a point whenever $\amax = 1$ or $\underline{\tau} = \overline{\tau}$. If $\amax = 1$, the weights are constant and equal to 1, whereas $\underline{\tau} = \overline{\tau}$ implies that $\tau_0$ is constant. In both cases, the weighted estimand equals the target estimand. The width of the bounds is equal to $(1 - \amax^{-1}) \cdot (\overline{\tau} - \underline{\tau})$, so higher values of $\amax^{-1}$ yield narrower bounds for given values of $(\underline{\tau},\overline{\tau})$. Thus, reporting an estimate of $\amax^{-1}$ allows the researcher to compute simple bounds for the target estimand by combining it with an estimate of $\mu(a,\tau_0)$ and posited bounds on $\tau_0$.

These bounds are sharp in the sense that knowledge of only $(\mu(a,\tau_0),\amax,\underline{\tau},\overline{\tau})$ does not further constrain the target estimand. The bounds could be tightened by assuming knowledge of other aspects of the joint distribution of $(Y(1),Y(0),D,X)$, but we focus on adding a single piece of additional information to the estimand, namely our diagnostic $\amax^{-1}$.

When the weights are nonnegative, the law of iterated expectations helps clarify this result by allowing us to write the target estimand, $\E_{W_0}[Y(1) - Y(0)]$, as
\begingroup
\allowdisplaybreaks
\begin{align}
		&\E_{W_0}[Y(1) - Y(0)\mid W^*=1] \cdot \Prob_{W_0}(W^*=1) \; + \; \E_{W_0}[Y(1) - Y(0)\mid W^* = 0] \cdot  \Prob_{W_0}(W^*=0)\notag\\
		&\qquad = \mu(a,\tau_0) \cdot \text{MIV} \; + \; \E_{W_0}[Y(1) - Y(0)\mid W^* = 0] \cdot (1 - \text{MIV}).\label{eq:bound_terms}
\end{align}
\endgroup
All terms in \eqref{eq:bound_terms} can be estimated, except for the average effect $\E_{W_0}[Y(1) - Y(0)\mid W^* = 0]$. However, it can be bounded from bounds on $\tau_0(X)$  since it also equals $\E_{W_0}[\tau_0(X) \mid W^* = 0]$. Substituting $\underline{\tau}$ and $\overline{\tau}$ as the ``worst-case'' bounds yields the bounds of Proposition \ref{prop:ATE_bounds_amax}.

We obtain a similar result if, instead of bounding the support of $\tau_0(X)$, we bound its variance.

\begin{proposition}[Bounds on Target Estimand]\label{prop:ATE_variance_bound}
Suppose the weights $a(X)$ satisfy Assumption \ref{assn:reg} and $\P_{W_0}(a(X) \geq 0) = 1$. Suppose $\mu(a,\tau_0)$ and $\amax^{-1}$ are known. Suppose that $\var_{W_0}(\tau_0(X)) \leq \overline{\sigma}_\tau^2$ for some known $\overline{\sigma}_\tau \geq 0$. Then,
\begin{align}\label{eq:bounds2}
	\E_{W_0}[Y(1) - Y(0)] \; \in \; \left[\mu(a,\tau_0) - \sqrt{\amax - 1} \cdot \overline{\sigma}_\tau, \; \mu(a,\tau_0) + \sqrt{\amax - 1} \cdot \overline{\sigma}_\tau\right],
\end{align} 
and these bounds are sharp.
\end{proposition}

\noindent
Once again, these bounds collapse to a point when $\amax = 1$, or when we assume the CATEs are constant, i.e., $\overline{\sigma}_\tau = 0$. The width of these bounds increases as $\amax^{-1}$ decreases, showcasing the link between the MIV and the information about the target estimand. One caveat of Proposition \ref{prop:ATE_variance_bound} is that it requires the weights to be nonnegative, whereas Proposition \ref{prop:ATE_bounds_amax} does not impose restrictions on the sign of the weights. These bounds are sharp in the sense that their width cannot be reduced further using only knowledge of $\mu(a,\tau_0)$, $\amax$, and $\overline{\sigma}_\tau$.

\subsection{Fraction of Negative Weights}

Our proposed measure can also be used to conduct further diagnostics on the weights. For example, \cite{ChaisemartinDHaultfoeuille2020} recommend reporting the \textit{fraction} of weights that are negative in TWFE estimation. While this fraction could be estimated directly, one can use our diagnostic $\amax^{-1}$ to obtain a simple sharp bound on this fraction.

\begin{proposition}[Bounds on Fraction of Negative Weights]\label{prop:negative_weights_bound}
Suppose the weights $a(X)$ satisfy Assumption \ref{assn:reg}. Suppose $\amax$ is known. Then,
\begin{align*}
\P_{W_0}(a(X) < 0) \, \leq \, 1 - \amax^{-1},
\end{align*}
and this bound is sharp.
\end{proposition}

\noindent
For example, if $\amax^{-1} = 95\%$, then at most 5\% of weights are negative. This bound is sharp, so we cannot say more about the fraction of negative weights using only knowledge of $\amax$.

\subsection{Covariate Distribution in the Implicit Subpopulation}

A separate diagnostic measure is linked to the implicit subpopulations $W^*$. While we focus on measures of the \textit{size} of these subpopulations, it may also be useful to report their characteristics. Specifically, we recommend reporting features of the covariate distribution within $W^*$, such as mean covariate values. Indeed, this distribution is identified, since we can write, for a generic integrable function $g$,
\begin{align*}
	\E[g(X)\mid W^* = 1] = \frac{\E_{W_0}[\underline{w}^*(X)g(X)]}{\E_{W_0}[\underline{w}^*(X)]} = \mu(\underline{w}^*/\E_{W_0}[\underline{w}^*(X)],g) = \mu(a,g) = \E_{W_0}[a(X)g(X)],
\end{align*}
a simple function of weights $a(\cdot)$ and the marginal distribution of $X$\@. We can recover the average covariate values in $\{W^* = 1\}$ by setting $g(X) = X$, or the entire distribution by setting $g(X) = \1(X \leq x)$ for all $x \in \R^{d_X}$. Reporting the average covariate values of units within and outside of $W^*$ may help assess the representativeness of $\mu(a,\tau_0)$. One may also compare $\E_{W_0}[g(X)]$ to $\E[g(X) \mid W^* = 1]$. Their difference is generally larger whenever $\amax^{-1}$ is smaller, analogously to the results of Propositions \ref{prop:ATE_bounds_amax} and \ref{prop:ATE_variance_bound}.

\section{Applications to Common Estimands}\label{sec:examples2}

Here we consider three identification strategies where commonly used estimands follow the structure of equation \eqref{eq:weighted_est_def}. We show how the results in Sections \ref{sec:causal_represent} and \ref{sec:quantifying_int_valid} apply in each of these cases. For simplicity, we assume that $\amax = \sup_{x \in \supp(X\mid W_0=1)} a(x)$ in this section. This condition is satisfied when $a(\cdot)$ is continuous or when $X$ has finite support, among other cases. We also note that our assumption $\amax < \infty$ holds trivially in every case considered below.

\subsection{Unconfoundedness}\label{subsec:unconf2}

In Section \ref{sec:example1}, we provided the expression for the coefficient on $D$ in a population regression of $Y$ on $(1,D,X)$:
\begin{align*}
	\beta_\text{OLS} &=\E\left[\frac{p(X)(1-p(X))}{\E[p(X)(1-p(X))]} \cdot \tau_0(X)\right].
\end{align*}
Suppose the target estimand is the ATE, i.e., $W_0 = 1$ almost surely. By Theorem \ref{thm:existence1}, there exists a regular subpopulation $W^*$ such that $\beta_\text{OLS}$ equals the average treatment effect over $W^*$ since the weight function $a_\text{OLS}(X) \coloneqq p(X)(1-p(X))/\E[p(X)(1-p(X))]$ is nonnegative. By Theorem \ref{thm:unc_upperbounds1}, the upper bound on the size of the subpopulation $W^*$ is 
\begin{align*}
	\text{MIV}(a_\text{OLS},1;\mathcal{T}_\text{all}) &= \frac{\E[p(X)(1-p(X))]}{\sup_{x \in \supp(X)}p(x)(1-p(x))}.
\end{align*}
A corresponding subpopulation $W^*$ can be written as
\begin{align*}
	W^* &= \1\left(U \leq \frac{p(X)(1-p(X))}{\sup_{x \in \supp(X)}p(x)(1-p(x))}\right),
\end{align*}
where $U \indep (Y(1),Y(0),X)$ and $U \sim \text{Unif}(0,1)$. This subpopulation is more likely to include units with larger treatment variation given covariates. Its size is largest when $\var(D\mid X) = p(X)(1-p(X))$ is constant, so that $\Prob(W^* = 1\mid X) = 1$. This occurs if and only if $p(X)$ has support contained in $\{b,1-b\}$ for some $b \in (0,1)$. Whenever $\var(p(X)(1-p(X))) > 0$, $\{W^* = 1\}$ is a strict subpopulation.

The size of this subpopulation is the expectation of $\var(D\mid X)$ divided by its maximum value. This expression can be further simplified or bounded. Its numerator is bounded above by $\var(D) = \Prob(D=1) \cdot \Prob(D=0)$, which is particularly simple to estimate. The denominator is a nonsmooth functional of $p(\cdot)$. However, if $X$ has continuous components, $p(X)$ may be continuously distributed, making it likely that $1/2 \in \supp(p(X))$. In this case, $\sup_{x \in \supp(X)} p(x)(1-p(x)) = 1/4$. Combining these two approximations yields
\begin{align*}
	\text{MIV}(a_\text{OLS},1;\mathcal{T}_\text{all}) \; &\leq \; 4 \var(D),
\end{align*}
when the support of $p(X)$ includes 1/2. This bound is trivial when $\Prob(D=1) = 1/2$, but informative when the treatment probability is close to 0 or 1. For example, if $\Prob(D=1) = 0.1$, the OLS estimand cannot causally represent more than 36\% of the population. This is consistent with \cite{Sloczynski2022}, who shows that the OLS estimand is more similar to the ATE when $\Prob(D=1)$ is closer to 1/2.

When $1/2 \in \supp(p(X))$, we can also compute bounds on the ATE using the MIV as in Proposition \ref{prop:ATE_bounds_amax}. Assuming that $\tau_0(X) \in [\underline{\tau},\overline{\tau}]$, bounds on the ATE are given by
\begin{align*}
	&\left[(\beta_\text{OLS} - \underline{\tau}) \cdot 4\E[\var(D\mid X)] + \underline{\tau}, \; (\beta_\text{OLS} - \overline{\tau}) \cdot 4\E[\var(D\mid X)] + \overline{\tau}\right].
\end{align*}
Estimating these bounds requires the estimation of one additional quantity beyond the OLS estimand, which is the expectation of $\var(D\mid X)$. The width of these bounds depends crucially on $\overline{\tau} - \underline{\tau}$, the width of the support of $\tau_0$.

To obtain the MIV of $\beta_\text{OLS}$ when the ATT is the target parameter, we use the following representation:
\begin{align*}
	\beta_\text{OLS} \; &= \; \E\left[\frac{1-p(X)}{1 - \E[p(X) \mid D = 1]} \cdot \tau_0(X)\mid D = 1\right],
\end{align*}
with weights $a_\text{OLS,2}(X) \coloneqq (1-p(X))/(1 - \E[p(X) \mid D = 1])$. 
Applying Theorem \ref{thm:unc_upperbounds1} shows that the MIV is given by
\begin{align*}
	\text{MIV}(a_\text{OLS,2},D;\mathcal{T}_\text{all}) &= \frac{\E[1-p(X) \mid D = 1]}{\sup_{x \in \supp(X\mid D=1)} (1 - p(x))} = \frac{\E[p(X)(1-p(X))]}{\Prob(D=1) \cdot  (1 -  \inf_{x \in \supp(X\mid D=1)}p(x))}.
\end{align*}
Once again, this bound depends only on $p(X)$ and the distribution of $X$. The implicit subpopulation satisfies $\Prob(W^*=1\mid D= 1, X) \propto a_\text{OLS,2}(X)$, so units with smaller propensity scores are more likely to be included in $W^*$, given that they are treated. The MIV is maximized at 1 when $p(X)$ is constant, or if $D \indep X$\@. On the other hand, if $p(X)$ takes values close to 0, the bound satisfies
\begin{align*}
	\text{MIV}(a_\text{OLS,2},D;\mathcal{T}_\text{all}) \; &\approx \; \frac{\E[p(X)(1-p(X))]}{\Prob(D=1)} \; \leq \; \frac{\Prob(D=1) \cdot \Prob(D=0)}{\Prob(D=1)} \; = \; \Prob(D=0).
\end{align*}
This suggests that the OLS estimand has greater internal validity with respect to the ATT when the treated share is \emph{smaller}. This again echoes the results in \cite{Sloczynski2022} on the relationship between $\Prob(D=1)$ and the interpretation of the OLS estimand.

Returning to the ATE as the target parameter, we can also assess the internal validity of $\beta_\text{OLS}$ when $\tau_0(X)$ is known. For simplicity, assume that $\tau_0(X)$ is continuously distributed and, without loss of generality, that $\text{ATE} > \beta_\text{OLS}$\@. Then, using Theorem \ref{thm:unc_upperbounds2}, we obtain
\begin{align*}
	\text{MIV}(a_{\text{OLS}},1;\{\tau_0\}) \, &= \, \Prob(\tau_0(X) \leq \alpha^*),
\end{align*}
where $\alpha^*$ satisfies $\E[\tau_0(X)\mid \tau_0(X) \leq \alpha^*] = \beta_\text{OLS}$. The MIV is largest when the least amount of trimming needs to be applied. This is the case when the trimmed values are largest, or when $\E[\tau_0(X)\mid \tau_0(X) \geq \alpha]$ is large for large $\alpha$.

\subsection{Instrumental Variables}\label{subsec:IV}

Now consider a binary treatment $D \in \{0,1\}$ and a binary instrument $Z \in \{0,1\}$. Potential treatments, denoted by $(D(1),D(0))$, are linked to the realized treatment through $Z$, that is, $D = D(Z)$. Potential outcomes, $Y(d,z)$ for $d,z\in \{0,1\}$, may depend on both $d$ and $z$ in the absence of an exclusion restriction. Let $Y = Y(D,Z)$ be the realized outcome. As before, let $X$ denote covariates. We make the following assumptions.

\begin{assumption}[Instrument Validity]\label{assn:IV}
Almost surely, the following hold:
\begin{enumerate}
	\item Exogeneity: $(Y(0,0),Y(1,0),Y(0,1),Y(1,1),D(1),D(0)) \indep Z\mid X$;
	\item Exclusion: $\Prob(Y(d,0) = Y(d,1) \mid X) = 1$ for $d \in \{0,1\}$;
	\item First stage: $\Prob(Z=1\mid X) \in (0,1)$ and $\Prob(D(1)=1\mid X) \neq \Prob(D(0)=1\mid X)$;
	\item Strong monotonicity: $\Prob(D(1) \geq D(0)\mid X) = 1$.
\end{enumerate}
\end{assumption}

\noindent
The first instrumental variables estimand we consider was originally studied by \cite{AngristImbens1995}. In addition to Assumption \ref{assn:IV}, suppose that the model for $X$ is saturated, with $K$ possible combinations of covariate values, i.e., let $\supp(X) = \{x_1,\ldots,x_K\}$. Let $X_S \coloneqq \left( 1, \1(X = x_1), \ldots, \1(X = x_{K-1}) \right)$ and $Z_S \coloneqq \left( Z, Z\cdot\1(X = x_1), \ldots, Z\cdot\1(X = x_{K-1}) \right) = ZX_S$, where $Z_S$ is the constructed instrument vector. The estimand of interest is given by
\begin{align*}
\beta_\text{2SLS} \coloneqq \left[ \left( \E \left[ W_S^{\prime} Q_S \right] \left( \E \left[ Q_S^{\prime} Q_S \right] \right) ^{-1} \E \left[ Q_S^{\prime} W_S \right] \right) ^{-1} \E \left[ W_S^{\prime} Q_S \right] \left( \E \left[ Q_S^{\prime} Q_S \right] \right) ^{-1} \E \left[ Q_S^{\prime} Y \right] \right] _1,
\end{align*}
where $W_S \coloneqq \left( D, X_S \right)$, $Q_S \coloneqq \left( Z_S, X_S \right)$, and $\left[ \cdot \right] _k$ denotes the $k$th element of the corresponding vector. This estimand has been studied by \cite{AngristImbens1995}, \cite{Kolesar2013}, \cite{BlandholBonneyMogstadTorgovitsky2026}, and \cite{Sloczynski2026}.

\begin{proposition}
\label{prop:ai1995}
	Suppose Assumption \ref{assn:IV} holds. Suppose $X$ is discrete with finite support. Then 
	\begin{align*}
		\beta_\text{2SLS} &= \E\left[ \frac{\cov(D,Z\mid X)}{\E[\cov(D,Z\mid X)\mid T_C = 1]} \cdot \tau_0(X) \mid T_C = 1\right]
	\end{align*}
where $T_C \coloneqq \1(D(1) > D(0))$ denotes the subpopulation of compliers and $\tau_0(X) \coloneqq \E[Y(1)-Y(0)\mid T_C = 1, X]$.
\end{proposition}

\noindent
Thus, $\beta_\text{2SLS}$ satisfies the representation in \eqref{eq:weighted_est_def} with $W_0 = T_C$. Note that $\beta_\text{2SLS} = \text{LATE} \coloneqq \E[Y(1) - Y(0) \mid T_C = 1]$ if and only if the weights, denoted by $a_\text{2SLS}(X)$, are uncorrelated with $\tau_0(X)$ among compliers; this holds when either the weights or $\tau_0$ is constant in $X$.

The practical limitation of focusing on $\beta_\text{2SLS}$ is that applied researchers rarely create multiple interacted instruments \citep[cf.][]{BlandholBonneyMogstadTorgovitsky2026}, which is how $Z_S$ is defined and used to obtain $\beta_\text{2SLS}$ above. A more practically relevant estimand is the ``noninteracted'' IV estimand,
\begin{align*}
	\beta_\text{IV} \coloneqq \left[ \left( \E \left[ Q^{\prime} W \right] \right) ^{-1} \E \left[ Q^{\prime} Y \right] \right]_1,
\end{align*}
where $Q \coloneqq \left( Z, X \right)$ and $W \coloneqq \left( D, X \right)$. We also impose the following assumption on the instrument propensity score, implied by the saturated specification in Proposition \ref{prop:ai1995}.
\begin{assumption}[Rich Covariates]
\label{assn:IPS}
	$\Prob(Z=1\mid X)$ is linear in $X$.
\end{assumption}

\noindent
Under Assumptions \ref{assn:IV} and \ref{assn:IPS}, \cite{Sloczynski2026} obtains the following representation of the ``noninteracted'' IV estimand.
\begin{proposition}
\label{prop:justid}
Suppose Assumptions \ref{assn:IV} and \ref{assn:IPS} hold. Then
\begingroup
\allowdisplaybreaks
\begin{align*}
	\beta_\text{IV} &= \E\left[ \frac{\var(Z\mid X)}{\E[\var(Z\mid X)\mid T_C = 1]} \cdot \tau_0(X) \mid T_C = 1 \right].
\end{align*}
\endgroup
\end{proposition}

\noindent
It follows that $\beta_\text{IV}$ is a weighted estimand satisfying \eqref{eq:weighted_est_def} with the same CATE function as in Proposition \ref{prop:ai1995}, and where the average is again taken over the complier subpopulation. Let $a_\text{IV}(X)$ denote the weights for this estimand.

\subsubsection*{Causal Representation and Internal Validity of IV Estimands}\label{subsubsec:IV2}

We first consider the estimand $\beta_\text{2SLS}$, which can be characterized as $\mu(a_\text{2SLS},\tau_0)$. Since $a_\text{2SLS}(X) \geq 0$, there exists a subpopulation of compliers such that $\beta_\text{2SLS}$ is an average treatment effect over that subpopulation. Its maximum size is given by
\begin{align*}
	\text{MIV}(a_\text{2SLS},T_C;\mathcal{T}_\text{all}) \; &= \; \frac{\E[\cov(D,Z\mid X)\mid T_C = 1]}{\sup_{x \in \supp(X \mid T_C = 1)}\cov(D,Z\mid X = x)}.
\end{align*}
The maximum value of the MIV is obtained when $\cov(D,Z\mid X)$ does not depend on $X$\@. This occurs, for example, when the instrument $Z$ and complier status $T_C$ are independent of $X$\@. In this case, the MIV is 1 and the estimand equals the LATE\@.

Under the representation in Proposition \ref{prop:justid}, the IV estimand has a different weight function, which is proportional to $\var(Z\mid X)$\@. Here, $a_\text{IV}(X) \geq 0$ and
\begin{align*}
	\text{MIV}(a_\text{IV},T_C;\mathcal{T}_\text{all}) &= \frac{\E[\var(Z\mid X)\mid T_C = 1]}{\sup_{x \in \supp(X\mid T_C = 1)} \var(Z\mid X = x)}.
\end{align*}
The MIV of $\beta_\text{IV}$ is maximized when $\var(Z\mid X)$ is constant, which occurs when $Z$ is independent of $X$\@. In this case, $\beta_\text{IV}$ equals the LATE\@. The quantities $\text{MIV}(a_\text{IV},T_C;\mathcal{T}_\text{all})$ and $\text{MIV}(a_\text{2SLS},T_C;\mathcal{T}_\text{all})$ are not ranked uniformly in the distributions of $(D(1),D(0),X,Z)$ as there are data-generating processes that make each of these two quantities larger than the other. For example, if $\var(Z \mid X)$ is constant but $\Prob(D(1) > D(0)\mid X)$ is not, then $\text{MIV}(a_\text{2SLS},T_C;\mathcal{T}_\text{all}) < \text{MIV}(a_\text{IV},T_C;\mathcal{T}_\text{all})$. This is plausible if $Z$ is randomly assigned and $X$ is a vector of pre-assignment characteristics. This inequality is reversed if $\cov(D,Z\mid X)$ is constant but $\Prob(D(1) > D(0)\mid X)$ is not. They have equal MIVs when $\Prob(T_C = 1\mid X)$ is constant. In this case, the estimands are equal, so this is not unexpected.

\subsection{Difference-in-Differences}\label{subsec:DiD}

Now suppose units are observed for $T$ periods and, for period $t \in \{1,\ldots,T\}$, let $D_t \in \{0,1\}$ denote binary treatment, $(Y_t(1),Y_t(0))$ potential outcomes, and $Y_t = Y_t(D_t)$ the realized outcome. Units are untreated prior to period $G \in \{2,3,\ldots,T\} \cup \{+\infty\}$, receive the treatment in period $G$, and remain treated thereafter; $G = +\infty$ denotes units never treated. Thus, $D_t = \1(G \leq t)$, and no units are treated in the first period. We assume the panel is balanced.

The TWFE estimand is often used in this setting. It consists of regressing the outcome on the treatment indicator, group indicators, and period indicators. By partitioned regression results, the coefficient on the treatment indicator is
\begin{align*}
\beta_\text{TWFE} \coloneqq \avgt\E\left[\ddot{D}_t Y_t\right]\big/\avgt \E\left[\ddot{D}_t^2 \right],
\end{align*}
where $\ddot{D}_t \coloneqq D_t - \frac{1}{T}\sum_{s=1}^T D_s - \E[D_t] + \frac{1}{T}\sum_{s=1}^T \E[D_s]$.

We assume a version of parallel trends most similar to the one in \cite{ChaisemartinDHaultfoeuille2020}.
\begin{assumption}[Difference-in-Differences]\label{assn:DiD}
Assume
\begin{enumerate}
	\item $\supp(G) = \{2,3,\ldots,T\} \cup \{+\infty\}$;
	\item For all $t \in \{2,\ldots,T\}$ and $g,g' \in \supp(G)$, we have that  $\E[Y_t(0) - Y_{t-1}(0)\mid G=g] = \E[Y_t(0) - Y_{t-1}(0)\mid G = g']$.
\end{enumerate}
\end{assumption}

\noindent
We use a proposition that is essentially a special case of Theorem 1 in \cite{ChaisemartinDHaultfoeuille2020} to obtain a representation of $\beta_\text{TWFE}$ as a weighted average.

\begin{proposition}
\label{prop:CDH}
	Suppose Assumption \ref{assn:DiD} holds. Then
\begin{equation*}
	\beta_\text{TWFE} = \avgt \E \left[ \frac{a_\text{TWFE}(G,t) \cdot \Prob(D_t=1\mid G)}{\frac{1}{T}\sum_{s = 1}^T \E[a_\text{TWFE}(G,s) \cdot \Prob(D_s=1\mid G)]} \cdot \E[Y_t(1) - Y_t(0)\mid G,D_t=1] \right],
\end{equation*}
where $a_\text{TWFE}(g,t) \coloneqq 1 - \frac{1}{T}\sum_{s=1}^T \E[D_s\mid G = g] - \E[D_t] + \frac{1}{T}\sum_{s=1}^T \E[D_s]$.
\end{proposition}

\noindent
We show that this representation satisfies equation \eqref{eq:weighted_est_def} by introducing an auxiliary variable $P$, uniformly distributed on $\{1,\ldots,T\}$ independently of $\{(Y_t(0),Y_t(1),G)\}_{t=1}^T$. This \textit{period} variable denotes the time period and is used to define $(Y(1),Y(0),Y,D) \coloneqq (Y_P(1),Y_P(0),Y_P,D_P)$, the potential outcomes, realized outcome, and treatment at random period $P$, respectively.

Letting $X \coloneqq (G,P)$, this means we can write $\beta_\text{TWFE}$ as
\begin{align*}
	\beta_\text{TWFE} &= \E\left[\frac{a_\text{TWFE}(X)}{\E[a_\text{TWFE}(X) \mid D = 1]} \cdot \tau_0(X) \mid D = 1\right],
\end{align*}
where $\tau_0(X) = \E[Y(1) - Y(0) \mid D=1, G, P]$, and $a_\text{TWFE}(X)$ is not generally nonnegative. A nonnegative weight function can be obtained if $\tau_0(X)$ is assumed constant over time. This property was described by \citet[Appendix 3.1]{ChaisemartinDHaultfoeuille2020} and \citet[Section 3.1.1]{Goodman-Bacon2021}, and the resulting representations of the TWFE estimand are given in their Theorem S2 and equation (16), respectively. The following proposition gives a simple expression for the weights in our setting.

\begin{proposition}
\label{prop:GB}
	Suppose Assumption \ref{assn:DiD} holds and that $\E[Y_t(1) - Y_t(0)\mid D=1, G] = \E[Y_s(1) - Y_s(0)\mid D=1, G]$ for any $s,t \in \{1,\ldots,T\}$. Then
	\begin{align*}
		\beta_\text{TWFE} &= \E\left[ \frac{a_\text{TWFE,H}(G)}{\E[a_\text{TWFE,H}(G) \mid D = 1]} \cdot \E[Y(1) - Y(0)\mid D= 1, G] \mid D = 1\right],
	\end{align*}
	where $a_\text{TWFE,H}(g) \coloneqq \Prob(D=0\mid G=g) \cdot (\Prob(D = 0\mid P \geq g) + \Prob(D=1\mid P < g)) \geq 0$ for $g \in \{2,\ldots,T\}$.
\end{proposition}

\noindent
As in the representation in Proposition \ref{prop:CDH}, the TWFE estimand in Proposition \ref{prop:GB} satisfies the representation in \eqref{eq:weighted_est_def}, with $X = G$, $W_0 = D$, $\tau_0(X) = \E[Y(1) - Y(0)\mid D=1,G]$, and the weight function $a_\text{TWFE,H}(G) \geq 0$. This weight function is derived in Appendix \ref{appsec:proofsforsec5}, with additional comparisons with the weights in \cite{Goodman-Bacon2021} in Appendix \ref{appsec:DiD}.

\subsubsection*{Causal Representation and Internal Validity of the TWFE Estimand}\label{subsubsec:DiD2}

We now consider the weights obtained in Proposition \ref{prop:GB} under its assumptions. These weights are nonnegative and therefore Theorem \ref{thm:existence1} guarantees the existence of a causal representation for $\beta_\text{TWFE}$ uniformly in $\tau_0 \in \mathcal{T}_\text{all}$.\footnote{In the context of Proposition \ref{prop:GB}, $\tau_0$ is a function of $G$ only, thus $\mathcal{T}_\text{all}$ denotes the set of all functions of $G$ with finite second moments. Note that this is a strict subset of all ``time-heterogeneous'' conditional average treatment effects, $\E[Y(1) - Y(0) \mid D = 1,G,P]$.} Using Theorem \ref{thm:unc_upperbounds1}, the MIV of $\beta_\text{TWFE}$ is 
\begingroup
\allowdisplaybreaks
\begin{align*}
	&\text{MIV}(a_\text{TWFE,H},D;\mathcal{T}_\text{all}) = \frac{\E[a_\text{TWFE,H}(G)\mid D =1]}{\sup_{g \in \supp(G\mid D=1)} a_\text{TWFE,H}(g)}\\
	&= \frac{\sum_{g=2}^T \var(D\mid G=g) \cdot (\Prob(D=0\mid P \geq g) + \Prob(D=1\mid P < g)) \cdot \Prob(G=g)}{\Prob(D=1) \cdot \sup_{g \in\{2,\ldots,T\}} \Prob(D=0\mid G=g) \cdot (\Prob(D=0\mid P \geq g) + \Prob(D=1\mid P < g))}.
\end{align*}
\endgroup
Due to the absorbing nature of the treatment in our setting, all expressions involving the distribution of $D$ given $P$ or $G$ can be derived as a function of the marginal distribution of $G$\@. Therefore, the MIV depends only on $\{\Prob(G=g)\}_{g \in \{2,\ldots,T\}}$.

To give some intuition, consider the case where $T = 3$ and therefore $G \in \{2,3,+\infty\}$. In this case, calculations yield that
\begin{align*}
	\text{MIV}(a_\text{TWFE,H},D;\mathcal{T}_\text{all}) &= \min\left\{\frac{2\Prob(G=2) + \omega \Prob(G=3)}{2\Prob(G = 2) + \Prob(G=3)},\frac{\omega^{-1}2\Prob(G=2) + \Prob(G=3)}{2\Prob(G = 2) + \Prob(G=3)}\right\},
\end{align*}
where $\omega \coloneqq \frac{4 - 2\Prob(G=2) - 4\Prob(G=3)}{2 - 2\Prob(G=2) - \Prob(G=3)} = \frac{a_\text{TWFE,H}(3)}{a_\text{TWFE,H}(2)}$. Therefore, $\beta_\text{TWFE}$ has perfect internal validity with respect to the ATT if and only if $\omega = 1$, which occurs if and only if $\P(G = 3) = 2/3$. Thus, the MIV is 1 when the probability of adoption at $t=3$ is $2/3$ and declines as $\P(G = 3)$ moves away from $2/3$. At that value, $a_\text{TWFE,H}(g)$ is constant in $g$, so $\beta_\text{TWFE}$ equals the ATT\@.

\section{Estimation and Inference} \label{sec:estimation}

We now consider estimation and inference for the MIV and MR\@. We focus here on the case where $\mathcal{T} = \mathcal{T}_\text{all}$ and briefly discuss the case where $\mathcal{T} = \{\tau_0\}$ in Appendix \ref{appsec:estimation}\@. We focus on estimating $\amax^{-1}$, which equals the MIV when the weights are nonnegative and remains a useful diagnostic when the MIV is zero, as shown in Section \ref{sec:diagnostics_bounds}.

Suppose we observe a random sample of size $n$, $\{(W_i,X_i)\}_{i=1}^n$, where $W_i$ is a set of variables needed to estimate $a(\cdot)$ and $w_0(\cdot)$. For example, under unconfoundedness we can let $W_i = D_i$ since the distribution of $(D,X)$ is sufficient to identify $a(\cdot)$; the outcome distribution does not affect $\amax^{-1}$. In our instrumental variables examples, we let $W_i = (D_i,Z_i)$.

Due to the normalization in Assumption \ref{assn:reg}, all of the weight functions we study are of the form
\begin{align*}
	a(X) &= \alpha(X)/\E_{W_0}[\alpha(X)]
\end{align*}
for some non-normalized weight function $\alpha(X)$ satisfying $\E_{W_0}[\alpha(X)] > 0$. Thus, we can write the MIV as
\begin{align*}
	\amax^{-1} &= \E_{W_0}[\alpha(X)]/\alpha_{\max},
\end{align*}
where $\alpha_{\max} \coloneqq \sup(\supp(\alpha(X) \mid W_0 = 1))$. We focus on the estimation of $\alpha(\cdot)$ since it allows us to separate the analysis of the expectation and essential supremum, which have different asymptotic properties.

Assuming the existence of estimators for $\alpha(\cdot)$ and $w_0(\cdot)$, we consider the following analog estimator of $\amax^{-1}$: 
\begin{align*}
	\widehat{a}_{\max}^{-1} &\coloneqq \frac{\avg \widehat{\alpha}(X_i)\widehat{w}_0(X_i)}{\avg \widehat{w}_0(X_i) \cdot \max_{i: \widehat{w}_0(X_i) > c_n} \widehat{\alpha}(X_i)}.
\end{align*}
Here $c_n$ is a tuning parameter with $c_n \to 0$ as $n \to \infty$. We first note that $\E[\alpha(X)\mid W_0=1]$ can be estimated by $\frac{\avg \widehat{\alpha}(X_i)\widehat{w}_0(X_i)}{\avg \widehat{w}_0(X_i)}$, consistently under standard conditions on $\widehat{\alpha}(\cdot)$ and $\widehat{w}_0(\cdot)$. Estimating $\alpha_{\max} = \sup(\supp(\alpha(X) \mid W_0 = 1))$ is more delicate. In some examples, this supremum is known or can be bounded above without data. For example, the OLS estimand under unconfoundedness has non-normalized weights $\alpha(X) = \var(D \mid X)$, bounded above by $1/4$. If $X$ is continuously distributed, then 1/2 may lie in the support of $p(X)$, allowing us to avoid estimating $\alpha_{\max}$. The IV estimand of Section \ref{subsec:IV} has non-normalized weights $\alpha_\text{IV}(X) = \var(Z\mid X)$ which are similarly bounded above by $1/4$. Likewise, $\alpha_\text{2SLS}(X) = \cov(D,Z\mid X) \leq 1/4$ by the Cauchy--Schwarz inequality. If knowledge of $\alpha_{\max}$ is not assumed, but $\supp(X\mid W_0=1)$ is known and $\alpha(x)$ is continuous,\footnote{Note that $\alpha(x)$ is trivially continuous on finite support.} then  $\sup_{x \in \supp(X\mid W_0=1)} \widehat{\alpha}(x)$ will be consistent for $\alpha_{\max}$ when $\widehat{\alpha}(x)$ is consistent for $\alpha(x)$ uniformly in $x \in \supp(X\mid W_0=1)$. Many parametric and nonparametric estimators for $\alpha(\cdot)$ satisfy this requirement.

In Appendix \ref{appsec:estimation}, we prove the consistency of $\widehat{a}_{\max}^{-1}$ for $\amax^{-1}$ and derive its limiting distribution. We also provide a step-by-step bootstrap algorithm that can be employed to conduct inference and prove its validity. This bootstrap approach is based on \cite{FangSantos2019} and is nonstandard, but yields valid inferences even when $\alpha_{\max}$ is estimated, as opposed to standard bootstrap approaches, such as the empirical bootstrap.

\section{Empirical Applications}\label{sec:empirical}

\defcitealias{Finkelsteinetal2012}{Finkelstein et al., 2012}

In this section, we use the proposed tools in four empirical applications, which focus on causal effects of cash transfers \citep{AEFLM2016}, well failure \citep{BFS2020}, health insurance \citepalias{Finkelsteinetal2012}, and unilateral divorce laws \citep{SW2006,Goodman-Bacon2021}.

\subsection{Causal Effects of Cash Transfers}

\defcitealias{AEFLM2016}{Aizer et al.~(2016)}
\defcitealias{BlandholBonneyMogstadTorgovitsky2026}{Blandhol et al., 2026}

Our first application revisits a study of the effects of the Mothers' Pension (MP) program on the log age at death of the children of its beneficiaries. In this study, \cite{AEFLM2016} construct a comparison group from the children of mothers who were initially deemed eligible for cash transfers from the MP program but were ultimately rejected. The main results, reported in Panel A of their Table 4, are based on four cross-sectional regression specifications, rely on the unconfoundedness assumption, and support the conclusion that cash transfers substantially increase children's longevity.

\begin{table}[!tb]
	\begin{adjustwidth}{-1in}{-1in}
		\centering
		\begin{threeparttable}
			\caption{Internal Validity of OLS Estimand for Effects of Cash Transfers}\label{tab:aizer}
			\begin{small}
				\begin{tabular}{c >{\centering\arraybackslash}m{1.35cm} >{\centering\arraybackslash}m{1.35cm} >{\centering\arraybackslash}m{0.05cm} >{\centering\arraybackslash}m{1.35cm} >{\centering\arraybackslash}m{1.35cm} >{\centering\arraybackslash}m{0.05cm} >{\centering\arraybackslash}m{1.35cm} >{\centering\arraybackslash}m{1.35cm} >{\centering\arraybackslash}m{0.05cm} >{\centering\arraybackslash}m{1.35cm} >{\centering\arraybackslash}m{1.35cm}}
					\hline\hline
          & \multicolumn{2}{c}{Specification \#1} &       & \multicolumn{2}{c}{Specification \#2} &       & \multicolumn{2}{c}{Specification \#3} &       & \multicolumn{2}{c}{Specification \#4} \\
					\cline{2-3}
					\cline{5-6}
					\cline{8-9}
					\cline{11-12}
          & ATE   & ATT   &       & ATE   & ATT   &       & ATE   & ATT   &       & ATE   & ATT \\
					\hline
          &       &       &       &       &       &       &       &       &       &       &  \\
    \multicolumn{12}{c}{A. Estimates of the effects of cash transfers}                            \\
					\hline
    OLS   & \multicolumn{2}{c}{0.0157} &       & \multicolumn{2}{c}{0.0158} &       & \multicolumn{2}{c}{0.0182} &       & \multicolumn{2}{c}{0.0167} \\
          & \multicolumn{2}{c}{(0.0058)} &       & \multicolumn{2}{c}{(0.0059)} &       & \multicolumn{2}{c}{(0.0062)} &       & \multicolumn{2}{c}{(0.0061)} \\
          &       &       &       &       &       &       &       &       &       &       &  \\
    Causal forest & 0.0092 & 0.0099 &       & 0.0133 & 0.0135 &       & 0.0121 & 0.0130 &       & 0.0121 & 0.0127 \\
          & (0.0057) & (0.0024) &       & (0.0058) & (0.0023) &       & (0.0060) & (0.0024) &       & (0.0059) & (0.0023) \\
					\hline
          &       &       &       &       &       &       &       &       &       &       &  \\
    \multicolumn{12}{c}{B. Internal validity of the OLS estimand}                                 \\
					\hline
    $\widehat{a}_{\max}^{-1}$  & 0.4675 & 0.3467 &       & 0.4126 & 0.2254 &       & 0.3745 & 0.1464 &       & 0.3745 & 0.1463 \\
    $\widehat{\P}_{W_0}(a(X) < 0)$ & 0     & 0     &       & 0.0378 & 0.0420 &       & 0.1120 & 0.1262 &       & 0.1117 & 0.1260 \\
    $\hMIV$ (uniformly in $\tau_0$) & 0.4675 & 0.3467 &       & 0     & 0     &       & 0     & 0     &       & 0     & 0 \\
    $\hMIV$ (given $\tau_0$) & 0.8276 & 0.8153 &       & 0.8889 & 0.8790 &       & 0.7134 & 0.7056 &       & 0.7774 & 0.7658 \\
				    \hline\hline
				\end{tabular}
			\end{small}
			\begin{footnotesize}
				\begin{tablenotes}[flushleft]
					\item \textit{Notes:} The data are a matched sample of 7,860 children of accepted and rejected MP applicants from \citetalias{AEFLM2016}. The outcome is log age at death, as reported in the MP records (specifications \#1 to \#3) or on the death certificate (specification \#4). The treatment is whether a mother was accepted for MP\@. The propensity score is estimated using the linear probability model, so the presence of estimated negative weights is evidence against the ``rich covariates'' assumption \citepalias{BlandholBonneyMogstadTorgovitsky2026}. The estimates of the MIV ``given $\tau_0$'' require an estimate of the CATE function, which we obtain using causal forests.
				\end{tablenotes}
			\end{footnotesize}
		\end{threeparttable}
	\end{adjustwidth}
\end{table}

Panel A of Table \ref{tab:aizer} replicates these results and additionally reports causal forest estimates of the ATE and ATT, which are on average nearly 30\% smaller than the original OLS estimates. How should we understand these estimates and their internal validity?

To answer this question, Panel B of Table \ref{tab:aizer} reports our main diagnostics. Across specifications, $\widehat{a}_{\max}^{-1}$ ranges from 37.45\% to 46.75\% when the ATE is the target parameter, and from 14.63\% to 34.67\% for the ATT\@. At the same time, when we do not place any restrictions on the CATE function, $\widehat{\text{MIV}}$ is zero in three out of four specifications. This is because the causal interpretation of the OLS estimand requires a linear model for the propensity score \citep[cf.][]{BlandholBonneyMogstadTorgovitsky2026}, and when we estimate this regression, the resulting weights are negative for 3.78\% to 12.62\% of observations across these three specifications. In other words, our diagnostics reveal, without using outcome data, that the OLS estimand may differ substantially from the ATE and ATT, and that it need not correspond to the average treatment effect for any implicit subpopulation.

The last row of Table \ref{tab:aizer} reports our estimates of the MIV that use the CATE function obtained using causal forests. This measure is not appropriate for an ``ex ante'' specification check, but instead allows the researcher to ``ex post'' evaluate the implicit subpopulation that is associated with the weighted estimand. Here, the OLS estimand represents between 71.34\% and 88.89\% of the population, or between 70.56\% and 87.90\% of the treated units. While these proportions are not negligible, they are much smaller than 1, which is the MIV of the probability limit of any consistent estimator that explicitly targets the ATE or ATT\@.

\subsection{Causal Effects of Well Failure}

Our second application focuses on causal effects of water loss in rural India. \cite{BFS2020} argue that the timing of borewell failure is as good as randomly assigned conditional on its age and location, which implies an unconfoundedness assumption. The authors estimate the impact of borewell failure on 50 outcomes, including agricultural practices, labor reallocation, child employment and schooling, household income and expenditures, assets and debt, and welfare. Because the authors primarily report OLS estimates, it is important to understand how the corresponding estimand should be interpreted. In particular, it is unclear whether it equals the average treatment effect for any subpopulation, how large that subpopulation could be, and how much the estimand may differ from the ATE and ATT\@.

\begin{table}[!b]
	\begin{adjustwidth}{-1in}{-1in}
		\centering
		\begin{threeparttable}
			\caption{Internal Validity of OLS Estimand for Effects of Well Failure}\label{tab:borewell}
			\begin{small}
				\begin{tabular}{c >{\centering\arraybackslash}m{1.9cm} >{\centering\arraybackslash}m{1.9cm} >{\centering\arraybackslash}m{0.05cm} >{\centering\arraybackslash}m{1.9cm} >{\centering\arraybackslash}m{1.9cm}}
					\hline\hline
          & \multicolumn{2}{c}{Specification \#1} &       & \multicolumn{2}{c}{Specification \#2} \\
					\cline{2-3}
					\cline{5-6}
          & ATE   & ATT   &       & ATE   & ATT \\
					\hline
    $\widehat{a}_{\max}^{-1}$  & 0.7855 & 0.3600 &       & 0.7251 & 0.2628 \\
    $\widehat{\P}_{W_0}(a(X) < 0)$ & 0.0088 & 0.0143 &       & 0.0413 & 0.0533 \\
    $\hMIV$ (uniformly in $\tau_0$) & 0     & 0     &       & 0     & 0 \\
				    \hline\hline
				\end{tabular}
			\end{small}
			\begin{footnotesize}
				\begin{tablenotes}[flushleft]
					\item \textit{Notes:} The data are a sample of 891 households in the Indian state of Karnataka from \cite{BFS2020}. The treatment is whether the household's first borewell had failed by the survey date. Specification \#1 controls for village indicators and additional covariates used by \cite{BFS2020}. Specification~\#2 adds fixed effects for the year in which each household drilled its first borewell. The propensity score is estimated using the linear probability model, so the presence of estimated negative weights is evidence against the ``rich covariates'' assumption \citepalias{BlandholBonneyMogstadTorgovitsky2026}.
				\end{tablenotes}
			\end{footnotesize}
		\end{threeparttable}
	\end{adjustwidth}
\end{table}

Given that \cite{BFS2020} use two identical right-hand-side regression specifications for every outcome, our framework allows us to address these questions with a single set of diagnostics and without using the data on the 50 outcomes one by one. Table \ref{tab:borewell} reports the values of $\widehat{a}_{\max}^{-1}$, $\widehat{\P}_{W_0}(a(X) < 0)$, and $\widehat{\text{MIV}}$, separately for each target parameter, ATE and ATT, and each of the two specifications considered by \cite{BFS2020}. A small fraction of the estimated weights is negative, which implies that $\widehat{\text{MIV}}$ is zero. The more central problem, however, is that $\widehat{a}_{\max}^{-1}$ is either 72.51\% or 78.55\% when the ATE is the target parameter, and either 26.28\% or 36.00\% for the ATT\@. It may not be immediately obvious whether the values of $\widehat{a}_{\max}^{-1}$ for the ATE are small enough to be concerning. However, the width of the resulting bounds on this parameter is either $0.2145 \cdot (\overline{\tau}-\underline{\tau})$ or $0.2749 \cdot (\overline{\tau}-\underline{\tau})$, where $[\underline{\tau},\overline{\tau}]$ is the assumed support of the CATE function. When the outcome is binary (e.g., whether a household has any outstanding debt), the width of these bounds is either 0.4290 or 0.5498, which is arguably too wide to be informative. Moreover, these bounds are nearly three times wider when the ATT is instead the target. As a result, a cautious researcher should estimate their preferred target parameter directly when faced with these values of our diagnostics.

\subsection{Causal Effects of Health Insurance}

\defcitealias{Finkelsteinetal2012}{Finkelstein et al.~(2012)}

Our third application reanalyzes data from the Oregon Health Insurance Experiment, a prominent randomized controlled trial with noncompliance. \citetalias{Finkelsteinetal2012} use the experiment to study the effects of Medicaid enrollment on health and economic outcomes. In this study, eligible participants were randomly selected for the opportunity to apply for public health insurance and, if selected, could then choose whether to enroll. In this case, Medicaid enrollment is the treatment, and the randomized selection indicator serves as an instrument that is valid conditional on household size and survey wave.

As in \cite{BFS2020}, \citetalias{Finkelsteinetal2012} use the same right-hand-side specification to estimate the causal effects of their treatment on a large number of outcomes. Specifically, their publicly available replication package contains data for 30 such outcomes. This again allows us to use our diagnostics as an ``ex ante'' specification check that ignores outcome data yet remains applicable to all 30 outcomes.

\begin{table}[!b]
	\begin{adjustwidth}{-1in}{-1in}
		\centering
		\begin{threeparttable}
			\caption{Internal Validity of IV and 2SLS Estimands for Effects of Health Insurance}\label{tab:ohie}
			\begin{small}
				\begin{tabular}{c >{\centering\arraybackslash}m{1.8cm} >{\centering\arraybackslash}m{1.8cm}}
					\hline\hline
          & $\beta_\text{IV}$   & $\beta_\text{2SLS}$ \\
					\hline
    $\widehat{a}_{\max}^{-1}$ & 0.9413 & 0.7460 \\
    $\widehat{\P}_{W_0}(a(X) < 0)$ & 0 & 0 \\
    $\hMIV$ (uniformly in $\tau_0$) & 0.9413 & 0.7460 \\
    $\hMR$ (uniformly in $\tau_0$) & 0.2720 & 0.2156 \\
				    \hline\hline
				\end{tabular}
			\end{small}
			\begin{footnotesize}
				\begin{tablenotes}[flushleft]
					\item \textit{Notes:} The data are a sample of 23,741 survey respondents from the Oregon Health Insurance Experiment. The treatment is Medicaid enrollment. The instrument is the randomized eligibility to apply for Medicaid.
				\end{tablenotes}
			\end{footnotesize}
		\end{threeparttable}
	\end{adjustwidth}
\end{table}

Table \ref{tab:ohie} reports the values of $\widehat{a}_{\max}^{-1}$, $\widehat{\P}_{W_0}(a(X) < 0)$, $\widehat{\text{MIV}}$, and $\widehat{\text{MR}}$ for the Oregon Health Insurance Experiment. We separately consider the ``noninteracted'' IV estimand, $\beta_\text{IV}$, which was the object of interest in \citetalias{Finkelsteinetal2012}, and the specification with multiple interacted instruments that targets $\beta_\text{2SLS}$. Because the specification is saturated in covariates and there is no evidence of monotonicity violations, all estimated weights are positive and $\widehat{\text{MIV}} = \widehat{a}_{\max}^{-1}$. The value of $\widehat{\text{MIV}}$ for $\beta_\text{IV}$ is 94.13\%, implying that this estimand corresponds to the average treatment effect for that share of compliers. Arguably, this is sufficiently close to 1 that any major differences between $\beta_\text{IV}$ and the LATE are unlikely. At the same time, the value of $\widehat{\text{MIV}}$ for $\beta_\text{2SLS}$ is 74.60\%, implying a more limited degree of internal validity of the interacted specification with respect to the subpopulation of compliers. Finally, the estimated proportion of compliers is 28.90\%, which implies that the value of $\widehat{\text{MR}}$ is $0.2890 \cdot 0.9413 = 27.20\%$ for $\beta_\text{IV}$ and $0.2890 \cdot 0.7460 = 21.56\%$ for $\beta_\text{2SLS}$. Without using outcome data, we can conclude that $\beta_\text{IV}$ should be close to the LATE but not necessarily the ATE, whereas $\beta_\text{2SLS}$ may not be close to either of these target parameters.

\subsection{Causal Effects of Unilateral Divorce Laws}

Our fourth and final application revisits the effects of unilateral divorce laws in the U.S. on female suicide \citep{SW2006,Goodman-Bacon2021}. The data consist of 41 states observed over the 1964--1996 period. The outcome of interest is the state- and year-specific female suicide rate. The treatment is whether the state allowed unilateral divorce in a given year. The sample includes D.C. but excludes the states excluded by \cite{Goodman-Bacon2021} as well as eight additional always-treated states.

Panel A of Table \ref{tab:divorce} reports our baseline estimates of the average effects of unilateral divorce laws on female suicide. After we drop the eight always-treated states, the TWFE estimate, $-0.604$, becomes much smaller in absolute value than the corresponding estimate in \cite{Goodman-Bacon2021}, $-3.080$. Unlike that estimate, ours is also statistically insignificant. The conclusion changes, however, when we explicitly target the average treatment effect on the treated (ATT), that is, the average effect for the largest subpopulation for which such an effect is identified under standard assumptions. Using the approach of \cite{CallawaySantAnna2021}, we obtain a highly significant estimate of $-10.220$. The approach of \cite{Wooldridge2025} produces a less precise estimate of $-5.530$. These estimates are more strongly suggestive of a causal effect of unilateral divorce laws than the TWFE estimate.

\begin{table}[!tb]
	\begin{adjustwidth}{-1in}{-1in}
		\centering
		\begin{threeparttable}
			\caption{Internal Validity of TWFE Estimand for Effects of Unilateral Divorce Laws}\label{tab:divorce}
			\begin{small}
				\begin{tabular}{>{\centering\arraybackslash}m{4.75cm} >{\centering\arraybackslash}m{4.75cm} >{\centering\arraybackslash}m{4.75cm}}
					\hline\hline
				    \multicolumn{3}{c}{A. Estimates of the effects of unilateral divorce laws} \\
				    \hline
				    \multirow{2}[4]{*}{TWFE} & \multicolumn{2}{c}{ATT} \\
				    \cline{2-3}
				          & \citeauthor{CallawaySantAnna2021} & \citeauthor{Wooldridge2025} \\
				    \hline
				    $-0.604$ & $-10.220$ & $-5.530$ \\
				    (2.622) & (3.086) & (3.650) \\
				    \hline
				          &       &  \\
				    \multicolumn{3}{c}{B. Internal validity of the TWFE estimand based on Proposition \ref{prop:CDH}} \\
				    \hline
				    $\widehat{a}_{\max}^{-1}$ & \multicolumn{2}{c}{0.1180} \\
				    $\widehat{\P}_{W_0}(a(X) < 0)$ & \multicolumn{2}{c}{0.2456} \\
				          &       &  \\
				          & uniformly in $\tau_0$ & given $\tau_0$ \\
				    \cline{2-3}
				    $\hMIV$ & 0 & 0.6216 \\
				    $\hMR$ & 0 & 0.3873 \\
				    \hline
				          &       &  \\
				    \multicolumn{3}{c}{C. Internal validity of the TWFE estimand based on Proposition \ref{prop:GB}} \\
				    \hline
				    $\widehat{a}_{\max}^{-1}$ & \multicolumn{2}{c}{0.2246} \\
				    $\widehat{\P}_{W_0}(a(X) < 0)$ & \multicolumn{2}{c}{0} \\
				          &       &  \\
				          & uniformly in $\tau_0$ & given $\tau_0$ \\
				    \cline{2-3}
				    $\hMIV$ & 0.2246 & 0.7707 \\
				    $\hMR$ & 0.1400 & 0.4802 \\
				    \hline\hline
				\end{tabular}
			\end{small}
			\begin{footnotesize}
				\begin{tablenotes}[flushleft]
					\item \textit{Notes:} The data are a panel of U.S. states over 1964--1996. The outcome is the state- and year-specific female suicide rate (per million women). The treatment is whether the state allowed unilateral divorce in a given year. The sample includes D.C. but excludes the states excluded by \cite{Goodman-Bacon2021} as well as eight additional always-treated states. The measures of internal validity ``given $\tau_0$'' require an estimate of the CATE function, which we obtain using the approach of \cite{Wooldridge2025}.
				\end{tablenotes}
			\end{footnotesize}
		\end{threeparttable}
	\end{adjustwidth}
\end{table}

While the TWFE estimate and the two estimates of the ATT are quite different, we focus on another implication of the nonuniformity of the TWFE weight function. We ask: How representative of the underlying population is the TWFE estimand? What is the internal validity of this estimand if we are interested in the treated subpopulation? Panel B of Table \ref{tab:divorce} reports the values of $\widehat{a}_{\max}^{-1}$, $\widehat{\P}_{W_0}(a(X) < 0)$, $\widehat{\text{MIV}}$, and $\widehat{\text{MR}}$, based on the representation of the TWFE estimand in \cite{ChaisemartinDHaultfoeuille2020}. First, because 24.56\% of the estimated weights are negative, $\widehat{\text{MIV}} = 0$ and $\beta_\text{TWFE}$ does not have a causal interpretation uniformly in $\tau_0$. While the existence of an implicit subpopulation is governed by the presence of negative weights, the difference between the TWFE estimand and the ATT is driven by weight nonuniformity. Here, too, our conclusions are pessimistic: the value of $\widehat{a}_{\max}^{-1}$ is only 0.1180, implying that the bounds on the ATT based on the TWFE estimate would have been very wide. Finally, when we estimate the CATE function and use these estimates in computing $\widehat{\text{MIV}}$ and $\widehat{\text{MR}}$, we conclude that the TWFE estimand corresponds to the average treatment effect for 62.16\% of the treated units or 38.73\% of the entire population.

Panel C of Table \ref{tab:divorce} revisits these questions on the basis of the representation of the TWFE estimand in Proposition \ref{prop:GB}. Here, we assume that the CATE function is constant over time, which eliminates the problem of negative weights. Indeed, we now conclude that the TWFE estimand has a causal interpretation uniformly in $\tau_0$, even if it is still not particularly representative of the underlying population and has limited internal validity with respect to the treated subpopulation. The values of $\widehat{\text{MR}}$ and $\widehat{\text{MIV}}$ are 14.00\% and 22.46\%, respectively. When we use the estimated CATE function in computing these measures, these estimates increase to 48.02\% and 77.07\%. This is much more than our initial estimate of 0, but still substantially less than 1, the MIV of the estimands in \cite{CallawaySantAnna2021}, \cite{Wooldridge2025}, and other recent papers, each of which explicitly targets the ATT\@.

\section{Conclusion} \label{sec:conclusion}

In this paper, we studied the representativeness and internal validity of a class of weighted estimands that includes the popular OLS, 2SLS, and TWFE estimands in additive linear models. We examined the conditions under which such estimands can be written as the average treatment effect for a subpopulation. When a given estimand corresponds to the average treatment effect for a large subset of the population of interest, we say its internal validity is high. Our main results derive sharp upper bounds on the size of that subpopulation under different assumptions on treatment effect heterogeneity. These bounds provide a practical tool for quantifying the internal validity of weighted estimands. The resulting diagnostics are implemented in our companion Stata package \texttt{causalrep}.

\small
\setlength\bibsep{0pt}
\bibliographystyle{econometrica}
\bibliography{WCbibfile}

\newpage
\appendix
\normalsize

\section{Proofs of Key Results}\label{sec:main_appendix}

\setcounter{equation}{0}
\renewcommand{\theequation}{\ref{sec:main_appendix}.\arabic{equation}}

This appendix contains the proofs for our key results in Sections \ref{sec:causal_represent} and \ref{sec:quantifying_int_valid}. Proofs for other results, including Theorem \ref{thm:unc_upperbounds2}, can be found in the Supplemental Appendix.

\begin{proof}[Proof of Proposition \ref{prop:subpops_properties}]
To show the first claim, note that the equation $\E_{W_0}[W^*(Y(1) - Y(0))\mid X] = \E_{W_0}[Y(1) - Y(0)\mid X] \cdot \underline{w}^*(X)$ holds since $W^*$ is a regular subpopulation of $W_0$. Since $\underline{w}^*(X) > 0$, we can obtain
\begingroup
\allowdisplaybreaks
\begin{align}
	\E_{W_0}[Y(1) - Y(0)\mid X] &= \E_{W_0}[W^*(Y(1) - Y(0))\mid X]/\underline{w}^*(X)\label{eq:subpopproperties_proof_eq1}\\
	&= \E_{W_0}[Y(1) - Y(0)\mid W^* = 1,X],\notag\\
	&= \E[Y(1) - Y(0)\mid W^* = 1,X],\notag
\end{align}
\endgroup
where the third equality holds from $W^*$ being a subpopulation of $W_0$.

The proposition's second claim is established below:
\begingroup
\allowdisplaybreaks
\begin{align*}
	\E[Y(1) - Y(0) \mid W^* = 1] &= \E_{W_0}[Y(1) - Y(0)\mid W^* = 1]\\
	&= \frac{\E_{W_0}[\E_{W_0}[W^*(Y(1) - Y(0))\mid X]]}{\E_{W_0}[\Prob_{W_0}(W^*=1\mid X)]}\\
	&= \frac{\E_{W_0}[\underline{w}^*(X) \cdot \E_{W_0}[Y(1) - Y(0)\mid X]]}{\E_{W_0}[\underline{w}^*(X)]}\\
	&= \mu(\underline{w}^*/\E_{W_0}[\underline{w}^*(X)],\tau_0).
\end{align*}
\endgroup
The first equality follows from $W^*$ being a subpopulation of $W_0$, the second from the law of iterated expectations and $W^*\leq W_0$, and the third from equation \eqref{eq:subpopproperties_proof_eq1}.
\end{proof}

\begin{proof}[Proof of Theorem \ref{thm:existence1}]

$(\Longrightarrow)$ First, suppose there exists $W^* \in \mathcal{W}(a;W_0,\mathcal{T}_\text{all})$. By Proposition~\ref{prop:subpops_properties}, we have that $0 = \mu(a,\tau) - \mu(\underline{w}^*/\E_{W_0}[\underline{w}^*(X)],\tau) = \mu\left(a - \frac{\underline{w}^*}{\E_{W_0}[\underline{w}^*(X)]},\tau\right)$ for all $\tau \in \mathcal{T}_\text{all}$. Let $\tau^* = a - \frac{\underline{w}^*}{\E_{W_0}[\underline{w}^*(X)]}$. By Assumption \ref{assn:reg}, $\E[\tau^*(X)^2] \leq 2\E[a(X)^2] + 2\E[\frac{\underline{w}^*(X)^2}{\E_{W_0}[\underline{w}^*(X)]^2}] < \infty$, which implies that $\tau^* \in \mathcal{T}_\text{all}$.

Thus, we must have $0 = \mu(a - \underline{w}^*/\E_{W_0}[\underline{w}^*(X)],\tau^*) = \mu(\tau^*,\tau^*)$. Expanding this equality yields $0 = \E_{W_0}[\tau^*(X)^2]$, which implies that $\P_{W_0}(\tau^*(X) = 0) = 1$ or that
\begin{align}
	1 &= \Prob_{W_0}\left(a(X) = \frac{\underline{w}^*(X)}{\E_{W_0}[\underline{w}^*(X)]}\right).\label{eq:prop_weights_estimands}
\end{align}
Since $\underline{w}^*(X) \geq 0$, this implies $\P_{W_0}(a(X) \geq 0) = 1$.

$(\Longleftarrow)$  Second, suppose that $\Prob_{W_0}(a(X) \geq 0) = 1$. Let $U \sim \text{Unif}(0,1)$ where $U \indep (Y(1),Y(0),X,W_0)$, and define 
$W^* = \1\left(U \leq a(X)/\amax\right)\cdot W_0$. We verify $W^*$ is a regular subpopulation of $W_0$. First, we see that $\Prob(W^* = 1) > 0$ because
\begingroup
\allowdisplaybreaks
\begin{align*}
	\Prob(W^*=1) &= \E_{W_0}\left[\Prob_{W_0}\left(U \leq a(X)/\amax \mid X\right)\right] \cdot \Prob(W_0 = 1)\\
	&= \E_{W_0}\left[a(X)/\amax\right] \cdot \Prob(W_0 = 1)\\
	&= \Prob(W_0 = 1)/\amax>0.
\end{align*}
\endgroup
The first equality follows from $\Prob(W_0 = 1\mid W^* = 1) = 1$ and $U \indep (X,W_0)$, the second from $a(X)/\amax \in [0,1]$ almost surely given $W_0=1$, and the third from Assumption~\ref{assn:reg}. The inequality follows from $\Prob(W_0 = 1) > 0$ and $\amax < \infty$. The two conditions in Definition \ref{def:subpop} hold immediately, so $W^*$ is a regular subpopulation of $W_0$.

Finally, let $\underline{w}^*(X) = \Prob_{W_0}(W^* = 1\mid X) = a(X)/\amax$. Using Proposition \ref{prop:subpops_properties}, we can see that, for a given $\tau \in \mathcal{T}_\text{all}$,
\begin{align*}
	\mu(\underline{w}^*/\E_{W_0}[\underline{w}^*(X)],\tau) &= \mu(a,\tau).
\end{align*}
Since $\tau \in \mathcal{T}_\text{all}$ was arbitrary, we have that $W^* \in \mathcal{W}(a;W_0,\mathcal{T}_\text{all})$.
\end{proof}

\begin{proof}[Proof of Theorem \ref{thm:existence2}]
To simplify the notation in the proof, we let $\mu_0 \coloneqq \mu(a,\tau_0)$.

$(\Longrightarrow)$ First, suppose $\mu_0 \notin \mathcal{S}(\tau_0;W_0)$ and suppose there exists $W^* \in \mathcal{W}(a;W_0,\{\tau_0\})$. Since $\mu_0 \notin \mathcal{S}(\tau_0;W_0)$, we can without loss of generality suppose that  $\Prob_{W_0}(\tau_0(X) > \mu_0) = 1$. Since $W^* \in \mathcal{W}(a;W_0,\{\tau_0\})$, we can write by Proposition \ref{prop:subpops_properties}
\begin{align}
	\mu_0 &= \mu(\underline{w}^*/\E_{W_0}[\underline{w}^*(X)],\tau_0) \geq \mu(\underline{w}^*/\E_{W_0}[\underline{w}^*(X)],\mu_0) = \mu_0.\label{eq:existence2proofineq}
\end{align}
The inequality is strict unless $\E_{W_0}[(\underbrace{\tau_0(X)-\mu_0}_{>0 \text{ w.p.1}})\underbrace{\underline{w}^*(X)}_{\in [0,1]}]= 0$ holds. This holds if $\Prob_{W_0}((\tau_0(X)-\mu_0)\underline{w}^*(X) = 0) = 1$, which in turn occurs if and only if $\Prob_{W_0}(\underline{w}^*(X) = 0) = 1$. This implies $\Prob_{W_0}(W^* = 1) = 0$, a contradiction of $W^* \in \mathcal{W}(a,W_0,\{\tau_0\})$. Therefore, the inequality in \eqref{eq:existence2proofineq} is strict and yields $\mu_0 > \mu_0$, a contradiction. Thus, $\mathcal{W}(a;W_0,\{\tau_0\}) = \emptyset$ when $\mu_0 \notin \mathcal{S}(\tau_0;W_0)$.

$(\Longleftarrow)$ Second, suppose $\mu_0 \in \mathcal{S}(\tau_0;W_0)$. Let 
$\mathcal{X}^- \coloneqq \{x \in \supp(X): \tau_0(x) \leq \mu_0\}$  and $\mathcal{X}^+ \coloneqq \{x \in \supp(X): \tau_0(x) \geq \mu_0\}$.
By $\mu_0 \in \mathcal{S}(\tau_0;W_0)$, $\Prob_{W_0}(X \in \mathcal{X}^-) > 0$ and $\Prob_{W_0}(X \in \mathcal{X}^+) > 0$. 

Let $U \sim \text{Unif}(0,1)$ where $U \indep (Y(1),Y(0),X,W_0)$. For $u\in[0,1]$, let
\begin{align*}
	W^*(u) &\coloneqq (\1(U > u, X \in \mathcal{X}^-) + \1(U \leq u,X \in \mathcal{X}^+))\cdot W_0.
\end{align*}
We show $W^*(u)$ is a regular subpopulation of $W_0$ for all $u \in [0,1]$. We can see that $W^*(u) \in \{0,1\}$, that $W^*(u) \indep (Y(1),Y(0))\mid X,W_0 =1$, and that $\Prob(W_0 = 1 \mid W^*(u) = 1) = 1$. To show that $W^*(u)$ characterizes a regular subpopulation of $W_0$, we also show that it is nonzero with positive probability:
\begin{align*}
	\Prob_{W_0}(W^*(u) = 1) &= (1-u) \Prob_{W_0}(X \in \mathcal{X}^-) + u\Prob_{W_0}(X \in \mathcal{X}^+) > 0
\end{align*}
for all $u \in [0,1]$, which implies $\Prob(W^*(u) = 1) > 0$ by $\Prob(W_0 = 1) > 0$. Hence, $W^*(u) \in \text{SP}(W_0)$ for all $u \in[0,1]$.
For $u \in [0,1]$, we have that $\underline{w}^*(X;u) \coloneqq \Prob_{W_0}(W^*(u) = 1\mid X) = (1-u) \1(X \in \mathcal{X}^-) + u\1(X \in \mathcal{X}^+)$. Therefore, using Proposition \ref{prop:subpops_properties},
\begingroup
\allowdisplaybreaks
\begin{align*}
	\E[Y(1) - &Y(0)\mid  W^*(u) = 1] = \mu(\underline{w}^*(\cdot;u)/\E_{W_0}[\underline{w}^*(X;u)],\tau_0)\\
	&= \frac{(1-u) \E_{W_0}[\1(X \in \mathcal{X}^-)\tau_0(X)] + u\E_{W_0}[\1(X \in \mathcal{X}^+)\tau_0(X)]}{(1-u) \Prob_{W_0}(X \in \mathcal{X}^-) + u\Prob_{W_0}(X \in \mathcal{X}^+)}.
\end{align*}
\endgroup
By construction, $\tau_0(X)\1(X \in \mathcal{X}^-) \leq \mu_0 \1(X \in \mathcal{X}^-)$ and $\tau_0(X)\1(X \in \mathcal{X}^+) \geq \mu_0 \1(X \in \mathcal{X}^+)$ almost surely. Therefore, $\E[Y(1) - Y(0)\mid W^*(0) = 1] \leq \mu_0 \leq \E[Y(1) - Y(0)\mid W^*(1) = 1]$.
By the continuity of $\E[Y(1)-Y(0)\mid W^*(u)=1]$ in $u$ and the intermediate value theorem, there exists $u^* \in [0,1]$ such that $\mu_0 = \E[Y(1) -Y(0)\mid W^*(u^*) = 1]$ and $W^*(u^*) \in \mathcal{W}(a;W_0,\{\tau_0\})$.
\end{proof}

\begin{proof}[Proof of Proposition \ref{prop:linear_CATEs}]	
	$(\Longrightarrow)$ First, we suppose there exists $W^* \in \mathcal{W}(a;W_0,\mathcal{T}_\text{lin})$. By Proposition \ref{prop:subpops_properties}, we have that $\mu(a - \underline{w}^*/\E_{W_0}[\underline{w}^*(X)],\tau) = 0$ for all $\tau \in \mathcal{T}_\text{lin}$. Therefore,
\begingroup
\allowdisplaybreaks
\begin{align*}
		0 &= \mu(a - \underline{w}^*/\E_{W_0}[\underline{w}^*(X)],\tau)\\
		&= \E_{W_0}[a(X)(c + d'X)] - \frac{\E_{W_0}[\underline{w}^*(X)(c + d'X)]}{\E_{W_0}[\underline{w}^*(X)]}\\
		&= d'\left(\E_{W_0}[a(X)X] - \frac{\E_{W_0}[\underline{w}^*(X)X]}{\E_{W_0}[\underline{w}^*(X)]}\right)
\end{align*}
\endgroup
for all $d \in \R^{d_X}$, which implies $\E_{W_0}[a(X)X] = \frac{\E_{W_0}[\underline{w}^*(X)X]}{\E_{W_0}[\underline{w}^*(X)]}$. Let $u(x) = \underline{w}^*(x) \P_{W_0}(X = x)/\E_{W_0}[\underline{w}^*(X)]$. We have that $\frac{\E_{W_0}[\underline{w}^*(X)X]}{\E_{W_0}[\underline{w}^*(X)]} = \sum_{x \in \supp(X \mid W_0 = 1)} xu(x)$, a convex combination of values in $\supp(X\mid W_0=1)$ because $u(\cdot) \geq 0$ and $\sum_{x \in \supp(X \mid W_0 = 1)} u(x) = 1$. Thus, $\E_{W_0}[a(X)X] \in \text{conv}(\supp(X\mid W_0=1))$.

	$(\Longleftarrow)$ Second, suppose that $\E_{W_0}[a(X)X] \in \text{conv}(\supp(X\mid W_0=1))$. By convexity, we can write $\E_{W_0}[a(X)X]$ as $\sum_{x \in \supp(X \mid W_0 = 1)} x u(x) \P_{W_0}(X = x) = \E_{W_0}[X u(X)]$ for some $u(\cdot) \geq 0$ satisfying $\E_{W_0}[u(X)] = 1$. Let $W^* = \1\left( U \leq u(X)/\max(\supp(u(X)\mid W_0=1))\right) \cdot W_0$, where $U \sim \text{Unif}(0,1) \indep (X,Y(1),Y(0),W_0)$. Then $W^*$ is a regular subpopulation of $W_0$ and $\underline{w}^*(X) = \frac{u(X)}{\max(\supp(u(X)\mid W_0=1))}$ since $\frac{u(X)}{\max(\supp(u(X)\mid W_0=1))} \in [0,1]$ with probability 1 given $W_0=1$. Therefore,  for all $\tau(x) = c + d'x \in \mathcal{T}_\text{lin}$, we have that
\begingroup
\allowdisplaybreaks
\begin{align*}
	\mu(\underline{w}^*/\E_{W_0}[\underline{w}^*(X)],\tau) &= \mu(u/\E_{W_0}[u(X)],\tau)\\
	&= \mu(u,\tau)\\
	&= \E_{W_0}\left[u(X)(c + d'X)\right]\\
	&= c + d'\E_{W_0}\left[a(X)X\right]\\
	&= \E_{W_0}\left[a(X)(c + d'X)\right]\\
	&= \mu(a,\tau).
\end{align*}
\endgroup
Therefore, by Proposition \ref{prop:subpops_properties} we have that $W^* \in \mathcal{W}(a;W_0,\mathcal{T}_\text{lin})$.
\end{proof}

\begin{proof}[Proof of Proposition \ref{prop:BV_CATEs}]	
We consider the $\Delta>0$ case first and the $\Delta=0$ case second.
	
\noindent\textbf{Case 1:} $\Delta > 0$.

\noindent
$(\Longrightarrow)$ First, let $\Prob_{W_0}(a(X) \geq 0) < 1$. We will show that $\mathcal{W}(a;W_0,\mathcal{T}_\text{BD}(\Delta)) = \emptyset$ by way of contradiction. 

Suppose there is a $W^* \in \mathcal{W}(a;W_0,\mathcal{T}_\text{BD}(\Delta))$ and let $\underline{w}^*(X) = \Prob_{W_0}(W^*=1\mid X) \in [0,1]$. Let $\tau^*(x) = \Delta \cdot \1(a(x) < 0)$, which implies $\tau^* \in \mathcal{T}_\text{BD}(\Delta)$. We have that $\mu(a,\tau^*) = \Delta \cdot \E_{W_0}[a(X) \1(a(X) < 0) ]$. $\E_{W_0}[a(X) \1(a(X) < 0)]$ is strictly negative because it is weakly negative and because $\E_{W_0}[a(X) \1(a(X) < 0)] = 0$ implies $\P_{W_0}(a(X) < 0) = 0$, a contradiction of $\Prob_{W_0}(a(X) \geq 0) < 1$.  Since $\Delta > 0$, we conclude that $\mu(a,\tau^*) < 0$. However, $\mu(\underline{w}^*/\E_{W_0}[\underline{w}^*(X)],\tau^*) \geq 0$ since $\tau^* \geq 0$. By Proposition \ref{prop:subpops_properties}, $W^* \in \mathcal{W}(a;W_0,\mathcal{T}_\text{BD}(\Delta))$ implies $0 > \mu(a,\tau^*) = \mu(\underline{w}^*/\E_{W_0}[\underline{w}^*(X)],\tau^*) \geq 0$, a contradiction. Therefore $\mathcal{W}(a;W_0,\mathcal{T}_\text{BD}(\Delta)) = \emptyset$.

$(\Longleftarrow)$ Second, suppose $\Prob_{W_0}(a(X) \geq 0) = 1$. By Theorem \ref{thm:existence1}, $\mathcal{W}(a;W_0,\mathcal{T}_\text{all}) \neq \emptyset$. Since $\mathcal{T}_\text{BD}(\Delta) \subseteq \mathcal{T}_\text{all}$, $\mathcal{W}(a;W_0,\mathcal{T}_\text{all}) \subseteq \mathcal{W}(a;W_0,\mathcal{T}_\text{BD}(\Delta))$ holds. Therefore, $\mathcal{W}(a;W_0,\mathcal{T}_\text{BD}(\Delta)) \neq \emptyset$, which means that $\mu(a,\tau)$ has a causal representation uniformly in $\tau \in \mathcal{T}_\text{BD}(\Delta)$.

\noindent \textbf{Case 2:} $\Delta=0$.

\noindent
When $\Delta=0$, the function class $\mathcal{T}_\text{BD}(\Delta)$ is the set of all constant functions. In this case, $\tau(X) = t_0$, where $t_0 \in \R$ denotes a constant. Thus $\mathcal{W}(a;W_0,\mathcal{T}_\text{BD}(0)) \neq \emptyset$ for all weight functions $a(\cdot)$ since $W_0 \in \text{SP}(W_0)$ and because $\mu(a,\tau) = \E_{W_0}[a(X)t_0] = t_0 = \E_{W_0}[Y(1) - Y(0)]$ for any $a(\cdot)$.
\end{proof}

\begin{proof}[Proof of Theorem \ref{thm:unc_upperbounds1}]
First, suppose $\Prob_{W_0}(a(X) \geq 0) = 1$. From Theorem \ref{thm:existence1}, there exists $W^* \in \mathcal{W}(a;W_0,\mathcal{T}_\text{all})$.  From derivations in the proof of Theorem \ref{thm:existence1} (see equation \eqref{eq:prop_weights_estimands}) we have that $\Prob_{W_0}(C \cdot a(X) = \underline{w}^*(X)) = 1$ for some positive constant $C > 0$. Since $\underline{w}^*(X) \leq 1$, we must have $C \cdot a(X) \leq 1$ almost surely given $W_0=1$. Thus $C$ is bounded above by $\inf(\supp(1/a(X)\mid W_0=1)) = 1/\amax$, which is strictly positive by assumption. Therefore,
\begingroup
\allowdisplaybreaks
\begin{align*}
	\Prob_{W_0}(W^*=1) &=\E_{W_0}[\underline{w}^*(X)] = \E_{W_0}[ C \cdot a(X)] \leq \amax^{-1}.
\end{align*}
\endgroup
This upper bound is sharp because it is attained by setting $W^* = \1\left(U \leq a(X)/\amax\right)\cdot W_0$ and noting that $W^* \in \mathcal{W}(a;W_0,\mathcal{T}_\text{all})$ from the proof of Theorem \ref{thm:existence1}. 

Second, suppose $\Prob_{W_0}(a(X) \geq 0) < 1$. By Theorem \ref{thm:existence1}, $\mathcal{W}(a;W_0,\mathcal{T}_\text{all}) = \emptyset$ and therefore $\text{MIV}(a,W_0;\mathcal{T}_\text{all}) = 0$.
\end{proof}

\begin{proof}[Proof of Proposition \ref{prop:ATE_bounds_amax}]
We first show the validity of the bounds. We obtain the upper bound as follows:
\begin{align*}
	\E_{W_0}[\tau_0(X)] &= \E_{W_0}[(a(X)\amax^{-1} + (1 - a(X)\amax^{-1}))\tau_0(X)]\\
	&= \mu(a,\tau_0) \amax^{-1} + \E_{W_0}[(1 - a(X)\amax^{-1})\tau_0(X)]\\
	&\leq \mu(a,\tau_0) \amax^{-1} + \E_{W_0}[(1 - a(X)\amax^{-1})\overline{\tau}]\\
	&= \mu(a,\tau_0) \amax^{-1} + \overline{\tau}(1 - \amax^{-1}),
\end{align*}
where the inequality holds because $\P_{W_0}(a(X) \leq \amax) = 1$, which implies $\P_{W_0}(1 - a(X)\amax^{-1} \geq 0) = 1$. The lower bound is similarly obtained. 
		
	For sharpness, let $a(X)$ equal $\amax$ with probability $\amax^{-1}$ and 0 with probability $1 - \amax^{-1}$, and let $\tau_0(X) = \mu(a,\tau_0) \cdot\1(a(X) = \amax) + \overline{\tau} \cdot\1(a(X) = 0)$. Then $\tau_0(X) \in [\underline{\tau},\overline{\tau}]$, $\E_{W_0}[a(X)\tau_0(X)] = \mu(a,\tau_0)$, and $\E_{W_0}[Y(1) - Y(0)] = \E_{W_0}[\tau_0(X)] = \mu(a,\tau_0) \P_{W_0}(a(X) = \amax) + \overline{\tau} \P_{W_0}(a(X) = 0) = \mu(a,\tau_0)\amax^{-1} + \overline{\tau}(1 - \amax^{-1})$. Therefore, the upper bound is sharp. We can show the lower bound's sharpness by setting $\tau_0(X) = \mu(a,\tau_0) \cdot\1(a(X) = \amax) + \underline{\tau} \cdot\1(a(X) = 0)$ instead.
\end{proof}

\begin{proof}[Proof of Proposition \ref{prop:ATE_variance_bound}]
We first show that $\var_{W_0}(a(X)) \leq \amax - 1$. This follows from
\begin{align*}
	\var_{W_0}(a(X)) &= \E_{W_0}[a(X)^2] - \E_{W_0}[a(X)]^2 = \E_{W_0}[a(X)^2] - 1 \leq \amax \E_{W_0}[a(X)] - 1 = \amax - 1.
\end{align*}
The inequality follows from $a(X)^2 \leq a(X) \cdot \amax$ with probability 1 given $W_0 = 1$, which holds by the nonnegativity of the weights. By Cauchy--Schwarz we have that
\begin{align*}
	|\E_{W_0}[\tau_0(X)] - \mu(a,\tau_0)| &= |\E_{W_0}[(1-a(X))\tau_0(X)]|\\
	&= |\E_{W_0}[(1-a(X))(\tau_0(X) - \E_{W_0}[\tau_0(X)])]|\\
	&\leq \text{SD}_{W_0}(1 - a(X))\text{SD}_{W_0}(\tau_0(X))\\
	&\leq \sqrt{\amax - 1}\cdot\overline{\sigma}_\tau.
\end{align*}
A rearrangement yields the bounds in \eqref{eq:bounds2}.

To show they are sharp, let $a(X) = \amax$ with probability $\amax^{-1}$, and  $a(X) = 0$ with probability $1 - \amax^{-1}$. Sharpness is trivial when $\amax = 1$, so consider the case where $\amax > 1$. Let $\tau_0(X) = \frac{\overline{\sigma}_\tau}{\sqrt{\amax - 1}}(1 - a(X)) + \mu(a,\tau_0) + \sqrt{\amax - 1}\cdot\overline{\sigma}_\tau$. We have that
\begin{align*}
	\E_{W_0}[\tau_0(X)] &= \mu(a,\tau_0) + \sqrt{\amax - 1}\cdot\overline{\sigma}_\tau\\
	\E_{W_0}[a(X)\tau_0(X)] &= \frac{\overline{\sigma}_\tau}{\sqrt{\amax - 1}}(1 - \amax) + \mu(a,\tau_0) + \sqrt{\amax - 1}\cdot\overline{\sigma}_\tau = \mu(a,\tau_0)\\
	\var_{W_0}(\tau_0(X)) &= \frac{\overline{\sigma}_\tau^2}{\amax - 1}\var_{W_0}(a(X))= \overline{\sigma}_\tau^2 .
\end{align*}
Thus the upper bound is attained. The lower bound is attained by setting $\tau_0(X) = \frac{\overline{\sigma}_\tau}{\sqrt{\amax - 1}}(a(X) - 1) + \mu(a,\tau_0) - \sqrt{\amax - 1} \cdot \overline{\sigma}_\tau$. 
\end{proof}

\begin{proof}[Proof of Proposition \ref{prop:negative_weights_bound}]
	We can split $\E_{W_0}[a(X)] = 1$ as follows:
	\begin{align*}
		1 &= \E_{W_0}[a(X)] = \E_{W_0}[a(X) \1(a(X) \geq 0)] + \E_{W_0}[a(X) \1(a(X) < 0)] \leq \amax \P_{W_0}(a(X) \geq 0) + 0.
	\end{align*}
	Rearranging this inequality yields $\P_{W_0}(a(X) \geq 0) \geq \amax^{-1}$, or equivalently, $\P_{W_0}(a(X) < 0) \leq 1 - \amax^{-1}$.
	
	Sharpness of this bound is trivial when $\amax^{-1} = 1$, so consider the case where $\amax^{-1} < 1$. For $\varepsilon \in (0,1-\amax^{-1})$, let $\P_{W_0}(a(X) = \amax) = \amax^{-1} + \varepsilon$ and $\P_{W_0}(a(X) = -\varepsilon \amax/(1 - \amax^{-1} - \varepsilon))= 1 - \amax^{-1} - \varepsilon$. This distribution yields $\E_{W_0}[a(X)] = 1$, $\P_{W_0}(a(X) \leq \amax) = 1$, and $\P_{W_0}(a(X) < 0) = 1 - \amax^{-1} - \varepsilon$. This can be made arbitrarily close to $1 - \amax^{-1}$ by letting $\varepsilon \searrow 0$, so the bound is sharp.
\end{proof}

\newpage

\section*{Supplemental Appendix}

\setcounter{section}{1}
\renewcommand{\thesection}{\Alph{section}}
\renewcommand{\thesubsection}{\Alph{section}.\arabic{subsection}}

\noindent
This document contains appendices that supplement the main text. Appendix \ref{appsec:weaklycausal} formalizes the connection between uniform causal representations and weakly causal estimands, defined in \cite{BlandholBonneyMogstadTorgovitsky2026}. \mbox{Appendix \ref{appsec:estimation}} supplements the estimation and inference section (Section \ref{sec:estimation}) and provides an algorithm for constructing confidence intervals. Appendices \ref{appsec:fixed_t0_intvalidity}--\ref{appsec:proofsforsec6} provide proofs for results in the main text and for results in Appendices \ref{appsec:weaklycausal} and \ref{appsec:estimation}\@. Appendix \ref{appsec:DiD} contains additional calculations related to the weights for the TWFE estimand.

\section{Weakly Causal Estimands and Uniform Causal Representations in $\mathcal{T}_\text{all}$}\label{appsec:weaklycausal}

We now establish equivalence between \textit{weakly causal} estimands as defined in \cite{BlandholBonneyMogstadTorgovitsky2026} (henceforth, BBMT) and estimands that have uniform causal representations as in Theorem \ref{thm:existence1}. As in BBMT, consider the case where $X$ has finite support and, as in this paper, assume the treatment is binary. We also abstract from choice groups denoted by $G$ in BBMT.

Since $X$ has finite support, let $\supp(X\mid W_0=1) = \{x_1,\ldots,x_K\}$ and let $\tau \coloneqq (\tau(x_1),\ldots,\tau(x_K)) \in \R^K$ be the collection of CATEs. For $d \in \{0,1\}$ let $\nu_d(x) \coloneqq \E_{W_0}[Y(d)\mid X = x]$ denote the \textit{average structural function} (ASF) which also conditions on $W_0 = 1$, let $\nu_d \coloneqq (\nu_d(x_1),\ldots,\nu_d(x_K)) \in \R^K$, and let $\mathcal{M} \subseteq \R^{2K}$ be a set of possible ASFs such that $(\nu_0,\nu_1) \in \mathcal{M}$. We now state the definition of weakly causal estimands from BBMT (i.e., their Definition WC) in our setting which features binary treatments.

\begin{definition}
The estimand $\beta$ is \textit{weakly causal} if the following statements are true for all $(\nu_0,\nu_1) \in \mathcal{M}$:
\begin{enumerate}
	\item If $\nu_1 - \nu_0 \geq \mathbf{0}_K$,\footnote{Vector inequalities hold if they hold component-wise.} then $\beta \geq 0$.
	\item If $\nu_1 - \nu_0 \leq \mathbf{0}_K$, then $\beta \leq 0$.
\end{enumerate}
\end{definition}

\noindent
Thus, an estimand is weakly causal if, whenever all CATEs have the same sign, the estimand also has that sign. Whether an estimand satisfies this condition also depends on $\mathcal{M}$, the set of allowed ASFs. To compare weak causality to our result on uniform causal representations, we consider $\mathcal{M}_\text{all} \coloneqq \R^{2K}$, the unrestricted set of ASFs. The corresponding unrestricted set of CATE functions, which we denoted by $\mathcal{T}_{\text{all}}$, allows $\tau$ to be any vector in $\R^K$. We consider estimands characterized by our equation \eqref{eq:weighted_est_def}. We note that these estimands rule out ``level dependence,'' i.e., that the estimand changes if potential outcomes $(Y(0),Y(1))$ are translated to $(Y(0) + c,Y(1) +c)$ for some constant $c \in \R$\@. For example, the IV estimand is generally level dependent when the propensity score $\P(Z = 1 \mid X)$ is nonlinear in $X$\@.\footnote{See p.~17 in BBMT\@.} With these choices, we can show the two definitions are equivalent.

\begin{proposition}\label{prop:weakly_causal_relationship}
	Let $\mu(a,\tau_0)$ be an estimand satisfying equation \eqref{eq:weighted_est_def}. Suppose Assumption \ref{assn:reg} holds and that $a_{\max} < \infty$. Then $\mu(a,\tau_0)$ is weakly causal with  $\mathcal{M} = \mathcal{M}_{\text{all}}$ if and only if it has a causal representation uniformly in $\mathcal{T}_{\text{all}}$.
\end{proposition}

\noindent
The proof of this proposition hinges on the equivalence, under level independence, of weakly causal estimands and estimands with nonnegative weights, as in \mbox{Proposition 4} of BBMT\@. Also, as shown in Theorem \ref{thm:existence1}, estimands with nonnegative weights have a uniform causal representation in $\mathcal{T}_{\text{all}}$. Therefore, a weighted estimand has nonnegative weights if and only if it is weakly causal and if and only if it has a causal representation uniformly in $\mathcal{T}_\text{all}$. Thus, a weakly causal estimand admits a regular subpopulation $W^*$ such that the estimand measures the average effect of treatment over that subpopulation.

\section{Details on Estimation and Inference}\label{appsec:estimation}

This appendix complements Section \ref{sec:estimation} in the main text. In it, we compute the limiting distribution of our estimated measure of internal validity and prove the validity of a nonstandard bootstrap algorithm for constructing confidence intervals around it. This is done for the case where $\mathcal{T} = \mathcal{T}_\text{all}$. Estimation and inference for $\text{MIV}(a,W_0;\{\tau_0\})$ are related to their counterparts in linear programs with estimated constraints. See \cite{AndrewsRothPakes2023}, \cite{CoxShi2023}, \cite{FangSantosShaikhTorgovitsky2023}, and \cite{ChoRussell2024} for recent advances on this topic.

As in Section \ref{sec:estimation}, we consider the case where $X$ is discrete. We assume the existence of estimators for $\alpha(\cdot)$ and $w_0(\cdot)$, but not knowledge of $\supp(X\mid W_0=1)$.  In this case, $\alpha(x)$ and $w_0(x)$ are usually estimated ``cell-by-cell'' and their estimators are $\sqrt{n}$-consistent with limiting Gaussian distributions. We will see that inference on $\amax^{-1}$ is generally nonstandard and, as a result, most common bootstrap procedures fail.

We consider the following simple plug-in estimator:
\begin{align*}
	\widehat{a}_{\max}^{-1} &\coloneqq \frac{\avg \widehat{\alpha}(X_i)\widehat{w}_0(X_i)}{\avg \widehat{w}_0(X_i) \cdot \max_{i: \widehat{w}_0(X_i) > c_n} \widehat{\alpha}(X_i)},
\end{align*}
where $c_n$ is a tuning parameter that converges to 0 as $n\rightarrow \infty$. Note that this tuning parameter can be removed when $w_0$ is known, for example when $W_0 = 1$ almost surely. This estimator does not assume knowledge of the support of $X$ given $W_0=1$, but it can also be implemented by taking the maximum over $\supp(X\mid W_0=1)$ when it is known. 

Let $\supp(X) = \{x_1,\ldots,x_K\}$, denote by $p_j \coloneqq \P(X=x_j)$ the frequency of cell $j$, and let $\widehat{p}_j \coloneqq \avg \1(X_i = x_j)$ denote its sample frequency. Let $\widehat{\theta} \coloneqq (\widehat{\alpha},\widehat{w}_0,\widehat{p})$ where $\widehat{\alpha} \coloneqq (\widehat{\alpha}(x_1),\ldots,\widehat{\alpha}(x_K))$, $\widehat{w}_0 \coloneqq (\widehat{w}_0(x_1),\ldots,\widehat{w}_0(x_K))$, and $\widehat{p} \coloneqq (\widehat{p}_1,\ldots,\widehat{p}_K)$. Let $\theta \coloneqq (\alpha,w_0,p)$ denote their population counterparts.
 
We make the following assumption on the behavior of the first-step estimators.

\begin{assumption}[Preliminary Estimators]\label{assn:prelim_est}
Let
	\begin{align*}
		\sqrt{n}(\widehat{\theta} - \theta) \dconv \mathbb{Z}
	\end{align*}
as $n\rightarrow \infty$, where $\mathbb{Z} \coloneqq (\mathbb{Z}_\alpha,\mathbb{Z}_{w_0},\mathbb{Z}_p) \in \R^{K} \times \R^K \times \R^K$ has a Gaussian distribution. 	
\end{assumption}

\noindent
The above assumption is often satisfied when $X$ has finite support since estimators for $\alpha(x_j)$ and $w_0(x_j)$ can be obtained using only the observations for which $X_i = x_j$. Note that the  distribution of $\mathbb{Z}_{w_0}$ may be degenerate. For example, if $W_0 = 1$ almost surely, then $\widehat{w}_0(x) = w_0(x) = 1$ and thus $\mathbb{Z}_{w_0} = \mathbf{0}_K$, a $K$-vector of zeros.

The next theorem shows the consistency of $\widehat{a}_{\max}^{-1}$ and establishes the limiting distribution of this estimator.

\begin{theorem}[Consistency and Asymptotic Distribution]\label{thm:cons_asynorm}
	Suppose Assumption \ref{assn:prelim_est} holds. Suppose $c_n \rightarrow 0$ and $c_n \sqrt{n} \rightarrow \infty$ as $n \rightarrow \infty$.   Then, $\widehat{a}_{\max}^{-1}$ is consistent for $\amax^{-1}$ and
	\begin{align*}
		\sqrt{n}(\widehat{a}_{\max}^{-1} - \amax^{-1}) &\dconv \psi(\mathbb{Z}),
	\end{align*}
	where $\psi$ is a continuous mapping defined by
	\begin{align}
	\psi(\mathbb{Z}) \; &= \; \sum_{j=1}^K \frac{w_0(x_j)p_j}{\P(W_0=1) \alphamax} \mathbb{Z}_{\alpha}(j) \; - \; \frac{\E_{W_0}[\alpha(X)]}{\alpha_{\max}^2}\max_{j \in \Psi_{\mathcal{X}^+}} \mathbb{Z}_{\alpha}(j)\notag\\
		&\quad + \; \sum_{j=1}^K \frac{(\alpha(x_j) - \E_{W_0}[\alpha(X)])p_j}{\P(W_0=1) \alphamax} \mathbb{Z}_{w_0}(j) \notag\\
		&\quad + \; \sum_{j=1}^K \frac{(\alpha(x_j) - \E_{W_0}[\alpha(X)])w_0(x_j)}{\P(W_0 = 1)\alphamax}\mathbb{Z}_p(j),\label{eq:psi_mapping}
\end{align}
as $n\rightarrow \infty$, where $\Psi_{\mathcal{X}^+} \coloneqq \{j: x_j \in \mathcal{X}^+, \alpha(x_j) = \alphamax\}$ and $\mathcal{X}^+ \coloneqq \supp(X\mid W_0=1)$.
\end{theorem}

\noindent
The mapping $\psi$ is linear if and only if $\alpha(x)$ is maximized at a unique value in $\mathcal{X}^+$, and nonlinear if multiple values maximize $\alpha(x)$. The linearity of this mapping crucially affects the choice of the inference procedure. When $\psi$ is linear, the limiting distribution of $\widehat{a}_{\max}^{-1}$ is Gaussian and common bootstrap procedures, such as the empirical bootstrap, are valid whenever they are valid for $\widehat{\theta}$. 

When $\alpha(x)$ is maximized at more than one value, the limiting distribution of $\widehat{a}_{\max}^{-1}$ is nonlinear in $\mathbb{Z}$ and therefore non-Gaussian. In this case, standard bootstrap approaches are invalid \citep[see Theorem 3.1 of][]{FangSantos2019}. Because $\amax^{-1}$ can be written as a Hadamard directionally differentiable mapping of $\theta$, alternative bootstrap procedures, such as those proposed by \cite{HongLi2018} and \cite{FangSantos2019}, can be used to conduct valid inference on $\amax^{-1}$.

We propose a bootstrap procedure that can be applied regardless of the linearity of $\psi$. In order to show its validity, we assume that the limiting distribution $\mathbb{Z}$ can be approximated via a bootstrap procedure.

\begin{assumption}[Bootstrap for First-Step Estimators]\label{assn:boot_prelim}
	Let $\mathbb{Z}^* \coloneqq (\mathbb{Z}^*_\alpha,\mathbb{Z}^*_{w_0},\mathbb{Z}^*_p)  \in \R^{K} \times \R^{K} \times \R^{K}$ be a random vector such that $\mathbb{Z}^* \overset{p}{\rightsquigarrow} \mathbb{Z}$, where $\overset{p}{\rightsquigarrow}$ denotes convergence in probability conditional on the data used to compute $\widehat{\theta}$ as $n\rightarrow \infty$.
\end{assumption}

\noindent
This assumption is easily satisfied when $\widehat{p}$ are sample frequencies, $(\widehat{\alpha},\widehat{w}_0)$ are cell-by-cell estimators that are asymptotically linear and Gaussian, and when $\mathbb{Z}^*$ follows the distribution of these estimators under a standard bootstrap approach. For example, for the empirical bootstrap we can let $\mathbb{Z}_{p}^* \coloneqq \sqrt{n}(\widehat{p}^* - \widehat{p})$ where $\widehat{p}_j^* = \frac{1}{n}\sum_{i=1}^n \1(X_i^* = x_j)$, and where $(X_1^*,\ldots,X_n^*)$ are drawn from the empirical distribution of $(X_1,\ldots,X_n)$. 

\begin{theorem}[Bootstrap Validity]\label{thm:bootstrap}
		Suppose the assumptions of Theorem \ref{thm:cons_asynorm} hold and that Assumption~\ref{assn:boot_prelim} holds. Then,
		\begin{align*}
			\widehat{\psi}(\mathbb{Z}^*) \overset{p}{\rightsquigarrow} \psi(\mathbb{Z})
		\end{align*}
		as $n\rightarrow \infty$, where $\widehat{\psi}$ is defined by
		\begin{align}\label{eq:psihat_mapping}
	\widehat{\psi}(\mathbb{Z}^*) \; &\coloneqq \; \sum_{j=1}^K \frac{\widehat{w}_0(x_j)\widehat{p}_j}{\widehat{\P}(W_0=1) \widehat{\alpha}_{\max}} \mathbb{Z}^*_{\alpha}(j) \; - \; \frac{\widehat{\E}_{W_0}[\alpha(X)]}{\widehat{\alpha}_{\max}^2} \max_{j \in \widehat{\Psi}_{\widehat{\mathcal{X}}^+}} \mathbb{Z}^*_\alpha(j)\notag\\
		 &\quad + \; \sum_{j=1}^K \frac{(\widehat{\alpha}(x_j) - \widehat{\E}_{W_0}[\alpha(X)])\widehat{p}_j}{\widehat{\P}(W_0=1) \widehat{\alpha}_{\max}} \mathbb{Z}^*_{w_0}(j)\notag\\
		 & \quad + \; \sum_{j=1}^K \frac{(\widehat{\alpha}(x_j) - \widehat{\E}_{W_0}[\alpha(X)])\widehat{w}_0(x_j)}{\widehat{\P}(W_0 = 1)\widehat{\alpha}_{\max}}\mathbb{Z}^*_p(j),
\end{align}
where $\widehat{\E}_{W_0}[\alpha(X)] \coloneqq \frac{\avg \widehat{\alpha}(X_i)\widehat{w}_0(X_i)}{\avg \widehat{w}_0(X_i)}$, $\widehat{\P}(W_0=1) \coloneqq \avg \widehat{w}_0(X_i)$, $ \widehat{\alpha}_{\max} \coloneqq \max_{i: \widehat{w}_0(X_i) > c_n} \widehat{\alpha}(X_i)$, $\widehat{\Psi}_{\widehat{\mathcal{X}}^+} \coloneqq \left\{k: x_k \in \widehat{\mathcal{X}}^+, \widehat{\alpha}(x_k) \geq \max_{x_j \in \widehat{\mathcal{X}}^+} \widehat{\alpha}(x_j) - \xi_n\right\}$, where $\widehat{\mathcal{X}}^+ \coloneqq \{x_j : \widehat{p}_j > 0, \widehat{w}_0(x_j) > c_n\}$ and $\xi_n$ is a positive sequence satisfying $\xi_n \rightarrow 0$ and $\xi_n \sqrt{n} \rightarrow \infty$ as $n\rightarrow \infty$. 
\end{theorem}

\noindent
The proof of Theorem \ref{thm:bootstrap} shows that $\widehat{\Psi}_{\widehat{\mathcal{X}}^+}$ consistently estimates $\Psi_{\mathcal{X}^+}$ and thus satisfies the conditions of Theorem 3.2 in \cite{FangSantos2019}. The bootstrap procedure is valid whether the limiting distribution is Gaussian or not. If we assume $\alpha(x)$ is maximized at a single value, standard bootstrap procedures can also be used to approximate the limiting distribution of $\widehat{a}_{\max}^{-1}$.

We propose the following bootstrap procedure to compute a one-sided $(1-\beta)$ confidence interval for $\amax^{-1}$.

\begin{algorithm}[One-Sided Confidence Interval for $\amax^{-1}$] We compute the confidence interval in three steps:
\begin{enumerate}
	\item Compute $\widehat{\theta}$ and $\widehat{a}_{\max}^{-1}$ using the random sample $\{(W_i,X_i)\}_{i=1}^n$;
	\item For bootstrap samples $b = 1,\ldots,B$, compute $\widehat{\theta}^{\ast,b} \coloneqq (\widehat{\alpha}^{\ast,b},\widehat{w}_0^{\ast,b},\widehat{p}^{\ast,b})$ and $\mathbb{Z}^{\ast,b} \coloneqq \sqrt{n}(\widehat{\theta}^{\ast,b} - \widehat{\theta})$;
	\item Compute $\widehat{q}_{\beta}$, the $\beta$ quantile of $\widehat{\psi}(\mathbb{Z}^{\ast,b})$, and report the interval $\left[0,\widehat{a}_{\max}^{-1} -\widehat{q}_{\beta}/\sqrt{n}\right]$.
\end{enumerate}
\end{algorithm}

We could also view these inferential problems through the lens of intersection or union bounds. For example, we can write
\begin{align*}
	\amax^{-1} &= \inf_{x \in \supp(X\mid W_0=1)}\frac{\E_{W_0}[\alpha(X)]}{\alpha(x)} =: \inf_{x \in \supp(X\mid W_0=1)} \amax(x)^{-1}.
\end{align*}
Computing a one-sided confidence interval for $\amax^{-1}$ of the kind $[0,(\widehat{a}_{\max}^{-1})^+]$ can be cast as doing inference on intersection bounds. \cite{ChernozhukovLeeRosen2013} offer methods for such problems. Equivalently, the computation of a one-sided confidence interval $[(\widehat{a}_{\max}^{-1})^-,1]$ is related to inferential questions in union bounds: see \cite{Bei2024}. We leave all details for future work.

\section{Proof of Theorem \ref{thm:unc_upperbounds2}}\label{appsec:fixed_t0_intvalidity}

We begin with a technical lemma used in the proof of Theorem \ref{thm:unc_upperbounds2}. For brevity, throughout this appendix we omit the subscript $W_0$ from expectations and probabilities when doing so causes no confusion.

\begin{lemma}\label{lem:RLcont}
	Suppose Assumption \ref{assn:reg} holds. Then, 
\begin{enumerate}
	\item The functions $\alpha \mapsto \E_{W_0}[T_\mu \1(T_\mu < \alpha)]$ and $\alpha \mapsto \E_{W_0}[T_\mu \1(T_\mu \geq \alpha)]$ are left-continuous.
	\item The functions $\alpha \mapsto \E_{W_0}[T_\mu \1(T_\mu > \alpha)]$ and $\alpha \mapsto \E_{W_0}[T_\mu \1(T_\mu \leq \alpha)]$ are right-continuous.
\end{enumerate}
\end{lemma}

\begin{proof}[Proof of Lemma \ref{lem:RLcont}]
	 The function $\alpha \mapsto \E[T_\mu \1(T_\mu < \alpha)]$ is left-continuous if for any strictly increasing sequence $\alpha_n \nearrow \alpha$ we have that $\E[T_\mu \1(T_\mu < \alpha_n)] \rightarrow \E[T_\mu \1(T_\mu < \alpha)]$.  To see this holds, note that $f_n(t) \coloneqq t \1(t < \alpha_n) \rightarrow t \1(t<\alpha)$ since $t \1(t<\alpha_n) = 0$ for all $t \geq \alpha$, and $t \1(t< \alpha_n) = t$ whenever $t< \alpha$ for sufficiently large $n$. The random variable $| T_\mu \1(T_\mu < \alpha_n)| $ is dominated by $| T_\mu| $ and $\E[| T_\mu| ]<\infty$ by Assumption \ref{assn:reg}. Therefore, by dominated convergence, $\E[T_\mu \1(T_\mu < \alpha_n)] \rightarrow \E[T_\mu \1(T_\mu < \alpha)]$ hence $\E[T_\mu \1(T_\mu < \alpha)]$ is left-continuous. The function $\alpha \mapsto \E[T_\mu \1(T_\mu \geq \alpha)]$ is also left-continuous because $\E[T_\mu \1(T_\mu \geq \alpha)] = \E[T_\mu] - \E[T_\mu \1(T_\mu < \alpha)]$. The lemma's second claim can be similarly shown by considering a sequence $\alpha_n \searrow \alpha$.
\end{proof}

\begin{proof}[Proof of Theorem \ref{thm:unc_upperbounds2}]
We break down this proof into four cases.

\smallskip
\noindent\textbf{Case 1: $\mu_0 \in \mathcal{S}(\tau_0;W_0)$ and $\mu_0 < E_0$}

\noindent
We want to maximize $\Prob_{W_0}(W^* = 1)$ over the subpopulations $W^*$ in $\mathcal{W}(a;W_0,\{\tau_0\})$. Recall that $W^* \in \mathcal{W}(a;W_0,\{\tau_0\})$ if $\mu_0 = \mu(\underline{w}^*/\E_{W_0}[\underline{w}^*(X)],\tau_0)$ and $W^* \in \text{SP}(W_0)$ hold, where $\underline{w}^*(X) = \Prob_{W_0}(W^*=1\mid X)$. Therefore,
\begingroup
\allowdisplaybreaks
\begin{align*}
	\text{MIV}(a,W_0;\{\tau_0\}) &= \sup_{W^* \in \mathcal{W}(a;W_0,\{\tau_0\})} \Prob_{W_0}(W^* = 1)\\
	&= \sup_{W^* \in \{0,1\}: \mu_0 = \mu(\underline{w}^*/\E_{W_0}[\underline{w}^*(X)],\tau_0), W^* \in \text{SP}(W_0)} \E_{W_0}[\underline{w}^*(X) ]\\
	&\leq \sup_{W^* \in \{0,1\}:\mu_0 = \mu(\underline{w}^*/\E_{W_0}[\underline{w}^*(X)],\tau_0)} \E_{W_0}[\underline{w}^*(X)]\\
	&= \sup_{\underline{w}^*: \mu_0 = \mu(\underline{w}^*/\E_{W_0}[\underline{w}^*(X)],\tau_0), \underline{w}^*(X) \in [0,1]} \E_{W_0}[\underline{w}^*(X)].
\end{align*}
\endgroup
We will first derive a closed-form expression for an upper bound on $\text{MIV}(a,W_0;\{\tau_0\})$. We will then show that it is attained by a corresponding $W^+ \in \mathcal{W}(a;W_0,\{\tau_0\})$ and therefore equals $\text{MIV}(a,W_0;\{\tau_0\})$. To simplify notation, we omit the $W_0$ subscript from expectations and probabilities below.

Before proceeding, let $\alpha^+ \coloneqq \inf\{\alpha \geq 0: R(\alpha) \geq 0\}$ where $R(\alpha) \coloneqq \E[T_\mu \1(T_\mu \leq \alpha)]$. Since
$\E[T_\mu]=E_0-\mu_0>0$, we have that $\alpha^+<\infty$. By construction, we also have that $\alpha^+ \geq 0$.

If $\alpha^+=0$, then $R(0)=0$. Since $T_\mu \1(T_\mu\leq 0)\leq 0$, this implies $\Prob(T_\mu<0)=0$. For any feasible $\underline{w}^*$, the constraint $\E[\underline{w}^*(X)T_\mu]=0$, together with $\underline{w}^*(X)\geq 0$ and $T_\mu\geq 0$ almost surely, implies $\underline{w}^*(X)T_\mu = 0$ almost surely. This implies that $\underline{w}^*(X) = \underline{w}^*(X)\1(T_\mu = 0)$. Hence $\E[\underline{w}^*(X)] = \E[\underline{w}^*(X)\1(T_\mu = 0)] \leq \Prob(T_\mu=0)$. This upper bound is attained by setting $W^+\coloneqq \1(T_\mu=0)W_0$, so in this boundary case $\text{MIV}(a,W_0;\{\tau_0\})=\Prob(T_\mu=0)$, which coincides with \eqref{eq:HC_Uppbound} when $\alpha^+  = 0$.

We therefore assume for the remainder of Case 1 that $\alpha^+>0$. On $[0,\infty)$, $R$ is nondecreasing and right-continuous, so
$R(\alpha^+ )\geq0$. Moreover, by the definition of $\alpha^+$, $R(\alpha^+ -\varepsilon)<0$ for all $\varepsilon\in(0,\alpha^+)$.

Second, we show an upper bound for $\text{MIV}(a,W_0;\{\tau_0\})$. For all $\underline{w}^*$ such that $\mu_0 = \mu(\underline{w}^*/\E_{W_0}[\underline{w}^*(X)],\tau_0)$ and $\underline{w}^*(X) \in [0,1]$, we have that
\begingroup
\allowdisplaybreaks
\begin{align*}
	\E[\underline{w}^*(X)] &= \frac{\E[\underline{w}^*(X)(\alpha^+ - T_\mu)]}{\alpha^+} + \frac{\E[\underline{w}^*(X)T_\mu]}{\alpha^+}\\
	&= \frac{\E[\underline{w}^*(X)(\alpha^+ - T_\mu)]}{\alpha^+}\\
	&= \frac{\E[\underline{w}^*(X)(\alpha^+ - T_\mu)\1(T_\mu \leq \alpha^+)]}{\alpha^+} + \frac{\E[\underline{w}^*(X)(\alpha^+ - T_\mu)\1(T_\mu > \alpha^+)]}{\alpha^+}\\
	&\leq \frac{\E[1\cdot (\alpha^+ - T_\mu)\1(T_\mu \leq \alpha^+)]}{\alpha^+} + \frac{\E[0\cdot (\alpha^+ - T_\mu)\1(T_\mu > \alpha^+)]}{\alpha^+}\\
	&= \E[\1(T_\mu \leq \alpha^+)] - \frac{\E[ T_\mu \1(T_\mu \leq \alpha^+)]}{\alpha^+}\\
	&=: P^+.
\end{align*}
\endgroup
The second equality follows from $\mu_0 = \mu(\underline{w}^*/\E[\underline{w}^*(X)],\tau_0)$. The inequality follows from $\{0,1\}$ being lower/upper bounds for $\underline{w}^*(X)$. Therefore, $\text{MIV}(a,W_0;\{\tau_0\}) \leq P^+$. 

Third, and finally, we show this inequality is binding by finding $W^+ \in \mathcal{W}(a;W_0,\{\tau_0\})$ such that $\Prob(W^+ = 1 \mid W_0 = 1) = P^+$.

We start by defining
\begin{align*}
	\underline{w}^+(X) &\coloneqq \1(T_\mu < \alpha^+) + \left(1 - \frac{R(\alpha^+)\1(\Prob(T_\mu = \alpha^+) \neq 0)}{\alpha^+ \Prob(T_\mu = \alpha^+)}\right)\1(T_\mu = \alpha^+).
\end{align*}
This function is bounded above by 1 because $R(\alpha^+) \geq 0$ and $\alpha^+ > 0$. To show $\underline{w}^+$ is bounded below by 0, consider cases where $\Prob(T_\mu = \alpha^+)$ or $R(\alpha^+)$ equal and differ from 0. If $\Prob(T_\mu = \alpha^+) = 0$ or $R(\alpha^+) = 0$, then $\underline{w}^+(X) \in \{0,1\}$ and it is therefore bounded below by 0. If $\Prob(T_\mu = \alpha^+)>0$ and $R(\alpha^+) > 0$, then
\begingroup
\allowdisplaybreaks
\begin{align*}
	1 - \frac{R(\alpha^+)}{\alpha^+ \Prob(T_\mu = \alpha^+)} &=  \frac{\alpha^+ \Prob(T_\mu = \alpha^+)  - \E[T_\mu \1(T_\mu \leq \alpha^+)]}{\alpha^+ \Prob(T_\mu = \alpha^+)}\\
	&= \frac{\E[T_\mu \1(T_\mu = \alpha^+)]  - \E[T_\mu \1(T_\mu \leq \alpha^+)]}{\alpha^+ \Prob(T_\mu = \alpha^+)}\\
	&= \frac{ - \E[T_\mu \1(T_\mu < \alpha^+)]}{\alpha^+ \Prob(T_\mu = \alpha^+)}.
\end{align*}
\endgroup
By the definition of $\alpha^+$ as an infimum, we must have that $R(\alpha^+ - \varepsilon) < 0$ for all $\varepsilon > 0$ such that $\alpha^+ - \varepsilon \geq 0$. By Lemma \ref{lem:RLcont}, $\alpha \mapsto \E[T_\mu \1(T_\mu < \alpha)]$ is left-continuous and satisfies $\E[T_\mu \1(T_\mu < \alpha)] \leq R(\alpha)$. Therefore, since $R(\alpha^+ - \varepsilon) < 0$ for all $\varepsilon \in (0,\alpha^+)$, we must have that $\E[T_\mu \1(T_\mu < \alpha^+ - \varepsilon)] < 0$ for all $\varepsilon \in (0,\alpha^+)$. Letting $\varepsilon \searrow 0$ yields that $\E[T_\mu \1(T_\mu < \alpha^+)] \leq 0$. Therefore $- \E[T_\mu \1(T_\mu < \alpha^+)]/(\alpha^+ \Prob(T_\mu = \alpha^+)) \geq 0$ and $\underline{w}^+(X) \geq 0$.

Next, we compute
\vspace{-0.25cm}
\begingroup
\allowdisplaybreaks
\begin{align*}
	\E[\underline{w}^+(X)] &= \P(T_\mu < \alpha^+) + \left(1 - \frac{R(\alpha^+)\1(\Prob(T_\mu = \alpha^+) \neq 0)}{\alpha^+ \Prob(T_\mu = \alpha^+)}\right) \P(T_\mu = \alpha^+)\\
	&= 	\P(T_\mu \leq \alpha^+) - \frac{\E[T_\mu \1(T_\mu \leq \alpha^+)]}{\alpha^+}\1(\Prob(T_\mu = \alpha^+) \neq 0)\\
	&= \Prob(T_\mu \leq \alpha^+) - \frac{\E[T_\mu \1(T_\mu \leq \alpha^+)]}{\alpha^+}\\
	&= P^+.
\end{align*}
\endgroup
The indicator function disappears in the third equality because $\Prob(T_\mu = \alpha^+) = 0$ implies $\E[T_\mu \1(T_\mu \leq \alpha^+)] = 0$ as shown above.

We next verify $\mu(\underline{w}^+/\E[\underline{w}^+(X)],\tau_0) = \mu_0$. This condition is equivalent to $\E[\underline{w}^+(X)T_\mu] = 0$, which we verify here:
\begingroup
\allowdisplaybreaks
\begin{align*}
	\E[\underline{w}^+(X)T_\mu] 
	&=  \E[T_\mu \1(T_\mu \leq \alpha^+)]\\
	&\quad - \frac{R(\alpha^+)\1(\Prob(T_\mu = \alpha^+) \neq 0)}{\alpha^+ \Prob(T_\mu = \alpha^+)} \alpha^+ \Prob(T_\mu = \alpha^+)\\
	&= R(\alpha^+)- R(\alpha^+)\1(\Prob(T_\mu = \alpha^+) \neq 0)\\
	&= R(\alpha^+)\1(\Prob(T_\mu = \alpha^+) = 0).
\end{align*}
\endgroup
Therefore, $\E[\underline{w}^+(X)T_\mu]$ equals 0 when $R(\alpha^+)  = 0$. When $R(\alpha^+) > 0$, we have that $\Prob(T_\mu = \alpha^+) > 0$ as shown earlier. So $\E[\underline{w}^+(X)T_\mu]$ is also equal to 0 in this case.

We conclude by showing that $\underline{w}^+(X)$ corresponds to $\Prob(W^+ = 1\mid X)$ for some $W^+ \in \text{SP}(W_0)$. Let $U \sim \text{Unif}(0,1)$ satisfy $U \indep (Y(1),Y(0),X,W_0)$ and define
\begin{align*}
	W^+ &\coloneqq \left(\1(T_\mu < \alpha^+) + \1\left(T_\mu = \alpha^+, U \leq 1 - \frac{R(\alpha^+)\1(\Prob(T_\mu = \alpha^+) \neq 0)}{\alpha^+ \Prob(T_\mu = \alpha^+)}\right)\right) \cdot W_0.
\end{align*}
By construction, $W^+ \in \{0,1\}$, $\Prob(W^+ = \; 1\mid X, W_0 = 1) \; = \; \underline{w}^+(X)$, $\Prob(W_0 \; = \; 1 \mid W^+ = \; 1) \; = \; 1$, and $W^+ \indep (Y(1),Y(0))\mid X,W_0 = 1$. Also, since $\mu_0 \in \mathcal{S}(\tau_0;W_0)$, $\Prob(T_\mu \leq 0) = \Prob(\tau_0(X) \leq \mu_0) > 0$. Since $\alpha^+ > 0$ we have that $\Prob(W^+ = 1) \geq \Prob(T_\mu < \alpha^+) \geq \Prob(T_\mu \leq 0) > 0$. Therefore $W^+$ is a regular subpopulation of $W_0$ for which $\Prob(W^+ = 1 \mid W_0 = 1) = P^+$, hence $P^+$ is the maximum, which coincides with \eqref{eq:HC_Uppbound} when $\alpha^+  > 0$.

\medskip
\noindent\textbf{Case 2: $\mu_0 \in \mathcal{S}(\tau_0;W_0)$ and $\mu_0 > E_0$}

\noindent
As in Case 1, $\text{MIV}(a,W_0;\{\tau_0\}) \leq \sup_{\underline{w}^*: \mu(\underline{w}^*/\E_{W_0}[\underline{w}^*(X)],\tau_0) = \mu_0, \underline{w}^*(X) \in [0,1]} \E_{W_0}[\underline{w}^*(X)]$.

Let $\alpha^- \coloneqq \sup\{\alpha \leq 0: L(\alpha) \leq 0\}$ where $L(\alpha) \coloneqq \E[T_\mu \1(T_\mu \geq \alpha)]$.  By construction, $\alpha^- \in (-\infty,0]$. Similarly to \mbox{Case 1}, when $\alpha^- = 0$, we obtain as a special case that $\text{MIV}(a,W_0;\{\tau_0\}) = \P(T_\mu = 0)$ by setting $W^-\coloneqq \1(T_\mu=0)W_0$. 

The rest of this case assumes that $\alpha^- < 0$. By Lemma \ref{lem:RLcont}, $L$ is a left-continuous function, and therefore $L(\alpha^-) = \lim_{\alpha \nearrow \alpha^-} L(\alpha) \leq 0$. 

We now show an upper bound for $\text{MIV}(a,W_0;\{\tau_0\})$. For all $\underline{w}^*$ such that $\mu(\underline{w}^*/\E[\underline{w}^*(X)],\tau_0) = \mu_0$ and $\underline{w}^*(X) \in [0,1]$, we have that
\begingroup
\allowdisplaybreaks
\begin{align*}
	&\E[\underline{w}^*(X)]\\
	&\leq \frac{\E[1\cdot (\alpha^- - T_\mu)\1(T_\mu \geq \alpha^-)]}{\alpha^-} + \frac{\E[0\cdot (\alpha^- - T_\mu)\1(T_\mu < \alpha^-)]}{\alpha^-}\\
	&= \E[\1(T_\mu \geq \alpha^-)] - \frac{\E[ T_\mu \1(T_\mu \geq \alpha^-)]}{\alpha^-}\\
	&=: P^-,
\end{align*}
\endgroup
which follows from a similar argument as above. This implies $\text{MIV}(a,W_0;\{\tau_0\}) \leq P^- $. We now show that this inequality is an equality by finding $W^- \in \mathcal{W}(a;W_0,\{\tau_0\})$ such that $\Prob(W^- = 1 ) = P^-$. Let
\begin{align*}
	\underline{w}^-(X) &\coloneqq \1(T_\mu > \alpha^-) + \left(1 - \frac{L(\alpha^-)\1(\Prob(T_\mu = \alpha^-) \neq 0)}{\alpha^- \Prob(T_\mu = \alpha^-)}\right)\1(T_\mu = \alpha^-).
\end{align*}
The rest of the proof follows symmetrically from the previous case.

\medskip
\noindent \textbf{Case 3: $\mu_0 = E_0 \in \mathcal{S}(\tau_0;W_0)$}

\noindent
Let $W^* = W_0$. Then $W^* \in \text{SP}(W_0)$, we have that $\mu_0 = \E_{W_0}[Y(1) - Y(0)]$ and thus $\P_{W_0}(W^* = 1)$ is maximized at 1.

\medskip
\noindent \textbf{Case 4:} $\mu_0 \notin \mathcal{S}(\tau_0;W_0)$

\noindent
By Theorem \ref{thm:existence2}, there does not exist a regular subpopulation $W^*$ satisfying $\mu_0 = \E[Y(1) - Y(0)\mid W^* = 1]$ and therefore the supremum equals 0 by its definition.
\end{proof}

\section{Proofs for Section \ref{sec:examples2}}\label{appsec:proofsforsec5}

\begin{proof}[Proof of Proposition \ref{prop:CDH}]
We begin by noting that
\begin{align*}
	\beta_\text{TWFE} &= \frac{\avgt\E[\ddot{D}_t Y_t]}{\avgt\E[\ddot{D}_t^2]} = \frac{\sum_{t=1}^T \E[\ddot{D}_t Y_t\mid P=t]\Prob(P = t)}{\sum_{t=1}^T \E[\ddot{D}_t^2\mid P=t] \Prob(P=t)} = \frac{\E[\ddot{D} Y]}{\E[\ddot{D}^2]},
\end{align*} 
where the second equality follows from the uniform distribution of $P$ which is independent from $(\ddot{D}_t,Y_t)$ for all $t \in \{1,\ldots,T\}$. The third equality follows from defining $\ddot{D} \coloneqq \ddot{D}_P$. We also note that
\begingroup
\allowdisplaybreaks
\begin{align*}
	\ddot{D} &=  D_P - \frac{1}{T}\sum_{s=1}^T D_s - \sum_{t=1}^T\E[D_t]\1(P=t) + \sum_{s=1}^T \E[D_s]\E[\1(P=s)]\\
	&= D - \frac{1}{T}\sum_{s=1}^T \1(G \leq s) - \E[D\mid P] + \E[D]\\
	&= D - \E[D\mid G] - \E[D\mid P] + \E[D].
\end{align*}
\endgroup
The third equality follows from $\E[D\mid G] = \E[\1(G \leq P)\mid G] = \frac{1}{T}\sum_{s=1}^T \1(G \leq s) = \frac{1}{T}\sum_{s=1}^T D_s$. We break down the rest of this proof into four steps.

\bigskip
\noindent\textbf{Step 1: Splitting the Numerator in Two}

\noindent
We have that
\begin{align*}
	\E[\ddot{D}Y] &= \E[\ddot{D}(Y(0) + D (Y(1) - Y(0)))]\\
	&= \E[\ddot{D}\E[Y(0)\mid G,P]] + \E[\ddot{D}D \E[Y(1) - Y(0)\mid G,P]].	
\end{align*}
The first equality follows from $Y = Y(0) + D(Y(1) - Y(0))$ and the second from iterated expectations and $\E[D\mid G,P] = D$.

\bigskip
\noindent \textbf{Step 2: First Numerator Term}

\noindent
We have that
\begin{align*}
	\E[\ddot{D}\E[Y(0)\mid G,P]] &= \E[\ddot{D} \theta(G,P)] = \E[\ddot{D} \ddot{\theta}(G,P)],
\end{align*}
where $\theta(G,P) \coloneqq \E[Y(0)\mid G,P]$. The second equality follows by properties of projections and from defining $\ddot\theta(G,P)$ as follows:
\begin{align*}
	\ddot{\theta}(G,P) &\coloneqq \theta(G,P) - \E[\theta(G,P)\mid G] - \E[\theta(G,P)\mid P] + \E[\theta(G,P)]\\
	&= \E[Y(0)\mid G,P] - \E[Y(0)\mid G] - \E[Y(0)\mid P] + \E[Y(0)].
\end{align*}
Then, we note that
\begingroup
\allowdisplaybreaks
\begin{align*}
	\ddot{\theta}(g',t') &= \E[Y(0)\mid G =g',P = t'] - \E[Y(0)\mid G = g'] - \E[Y(0)\mid P = t'] + \E[Y(0)]\\
	&= \E[Y_{t'}(0)\mid G =g'] - \avgt \E[Y_t(0)\mid G=g']\\
	&\quad - \sum_{g \in \supp(G)} \left( \E[Y_{t'}(0)\mid G=g] - \avgt  \E[Y_t(0)\mid G=g]\right) \Prob(G=g)\\
	&= \avgt \E[Y_{t'}(0) - Y_t(0)\mid G=g'] - \E\left[\avgt \E[Y_{t'}(0) - Y_t(0)\mid G]\right].
\end{align*}
\endgroup
Assumption \ref{assn:DiD}.2 implies that for any pair $t,t' \in \{1,\ldots,T\}$ and any $g' \in \supp(G)$
\begin{align*}
	\E[Y_{t'}(0) - Y_t(0)\mid G=g'] = \E[Y_{t'}(0) - Y_t(0)].
\end{align*}
This can be shown for $t' > t$ by writing $\E[Y_{t'}(0) - Y_t(0)\mid G=g'] = \sum_{s=t+1}^{t'} \E[Y_{s}(0) - Y_{s-1}(0)\mid G=g'] = \sum_{s=t+1}^{t'} \E[Y_{s}(0) - Y_{s-1}(0)] = \E[Y_{t'}(0) - Y_t(0)]$. Similar derivations show this holds for $t' < t$. The case where $t' = t$ is trivial. Therefore,
\begin{align*}
	\ddot{\theta}(g',t') &= \avgt \E[Y_{t'}(0) - Y_t(0)\mid G=g'] - \E\left[\avgt \E[Y_{t'}(0) - Y_t(0)\mid G]\right]\\
	&= \avgt \E[Y_{t'}(0) - Y_t(0)] - \E\left[\avgt \E[Y_{t'}(0) - Y_t(0)]\right]\\
	&= 0
\end{align*}
for all $(g',t') \in \supp(G) \times \{1,\ldots,T\}$, which implies $\E[\ddot{D} \E[Y(0)\mid G,P]] = 0$.

\bigskip
\noindent \textbf{Step 3: Second Numerator Term}

\noindent
We can write
\begin{align}
	\ddot{D} D &= (D - \E[D\mid G] - \E[D\mid P] + \E[D])D\notag\\
	&= (1 - \E[D\mid G] - \E[D\mid P] + \E[D])\Prob(D=1\mid G,P)\label{eq:ddot_d_equation}
\end{align}
by $D^2 = D = \E[D\mid G,P]$. Thus,
\begin{align*}
	&\E[\ddot{D} D \E[Y(1) - Y(0)\mid G,P]]\\
	&= \E[(1 - \E[D\mid G] - \E[D\mid P] + \E[D]) \E[Y(1) - Y(0)\mid G,P, D = 1]\Prob(D=1\mid G,P)].
\end{align*}

\noindent\textbf{Step 4: Denominator}

\noindent
In this step, we show that
\begin{align*}
	\E[\ddot{D}^2] &= \E[\ddot{D} D] = \E[(1 - \E[D\mid G] - \E[D\mid P] + \E[D])\Prob(D=1\mid G,P)].
\end{align*}
The first equality is obtained from properties of linear projections and the second follows from equation \eqref{eq:ddot_d_equation}.

\medskip
\noindent
We conclude the proof by noting the equivalence of integrating over the distribution of $P$ and averages over time periods, which shows the equivalence between $\beta_\text{TWFE}$, the expression in Proposition \ref{prop:CDH}, and
\begin{align*}
	 \frac{\E[(1 - \E[D\mid G] - \E[D\mid P] + \E[D]) \cdot \E[Y(1) - Y(0)\mid G,P, D = 1] \cdot \Prob(D=1\mid G,P)]}{\E[(1 - \E[D\mid G] - \E[D\mid P] + \E[D]) \cdot \Prob(D=1\mid G,P)]}.
\end{align*}
\end{proof}

\begin{proof}[Proof of Proposition \ref{prop:GB}]
Proposition \ref{prop:CDH} and $\Prob(D=1\mid G,P) = D$ yield
\begin{align*}
	\beta_\text{TWFE} &= \frac{\E[(1 - \E[D\mid G] - \E[D\mid P] + \E[D]) \cdot D \cdot \E[Y(1) - Y(0)\mid G,P,D=1]	]}{\E[(1 - \E[D\mid G] - \E[D\mid P] + \E[D]) \cdot D]}.
\end{align*}
Since $\E[Y(1) - Y(0)\mid G,P,D=1] = \E[Y(1) - Y(0)\mid G,D=1]$ by assumption, we can use the law of iterated expectations to obtain
\begin{align*}
	\beta_\text{TWFE} &= \frac{\E[\E[(1 - \E[D\mid G] - \E[D\mid P] + \E[D]) \cdot D\mid G] \cdot \E[Y(1) - Y(0)\mid G,D=1]	]}{\E[\E[(1 - \E[D\mid G] - \E[D\mid P] + \E[D]) \cdot D\mid G]]}.
\end{align*}
We now calculate the conditional expectation $\E[(1 - \E[D\mid G] - \E[D\mid P] + \E[D]) \cdot D\mid G = g]$ for $g \in \supp(G)$.  If $g = +\infty$, then this conditional expectation is 0 by construction, so we focus on the case where $g \in \{2,\ldots,T\}$. For these derivations, we let $F_G(p) \coloneqq \P(G \leq p)$ denote the cdf of $G$ at $p$. We have that:
\begingroup
\allowdisplaybreaks
\begin{align*}
	&\E[(1 - \E[D\mid G] - \E[D\mid P] + \E[D])D\mid G = g]\\
	&= \E[D\mid G = g]\left(1 - \E[D\mid G = g] - \frac{\E[\E[D\mid P]D\mid G=g]}{\E[D\mid G=g]} + \E[D]\right)\\
	&= \E[D\mid G=g]\left(1 - \E[\1(G \leq P)\mid G=g] - \frac{\E[F_G(P)\1(G \leq P)\mid G = g]}{\E[\1(G \leq P)\mid G=g]}\right.\\
	&\quad \left. + \; \E[\E[\1(G \leq P)\mid P]]\right)\\
	&= \E[D\mid G=g]\left(1 - \E[\1(g \leq P)] - \frac{\E[F_G(P)\1(g \leq P)]}{\E[\1(g \leq P)]} + \E[F_G(P)]\right)\\
	&= \E[D\mid G=g]\left(1 - \E[\1(g \leq P)] - \E[F_G(P)\mid g \leq P]\right.\\
	&\quad \left. + \; \E[F_G(P)\mid g \leq P]\E[\1(g \leq P)] + \E[F_G(P)\mid g > P]\E[\1(g > P)]\right)\\
	&= \E[D\mid G=g]\E[\1(g > P)](1 + \E[F_G(P)\mid g > P] - \E[F_G(P)\mid g \leq P])\\
	&= \E[D\mid G=g](1 - \E[D\mid G=g])(1 + \E[D\mid g > P] - \E[D\mid g \leq P])\\
	&= \Prob(D=1\mid G=g) \cdot \Prob(D=0\mid G=g) \cdot (\Prob(D=1\mid P < g) + \Prob(D=0\mid P \geq g)).
\end{align*}
\endgroup
The first equality follows from $\E[D\mid G=g] > 0$ for $g \in \{2,\ldots,T\}$, the second from $D = \1(G \leq P)$ and the law of iterated expectations, the third from $G \indep P$, the fourth from the law of iterated expectations, the fifth from combining terms, the sixth from the law of iterated expectations again, and the last line is obtained by the fact that $D \in \{0,1\}$. The representation in Proposition \ref{prop:GB} follows.
\end{proof}

\section{Proofs for Appendix \ref{appsec:weaklycausal}}
\label{appsec:proofsforappa}

\begin{proof}[Proof of Proposition \ref{prop:weakly_causal_relationship}]
By Theorem \ref{thm:existence1}, $\mu(a,\tau_0)$ has a uniform causal representation in $\mathcal{T}_\text{all}$ if and only if $a(x_k) \geq 0$ for $k \in \{1,\ldots,K\}$. Therefore, it is sufficient to show the equivalence between weakly causal estimands and estimands with nonnegative weights. A similar result was shown in Proposition 4 of BBMT, but we nevertheless provide a proof here to account for slight differences in notation.

	Suppose $\mu(a,\tau_0)$ is weakly causal. Let $\nu_1 = (\1(a(x_1) < 0),\ldots,\1(a(x_K) < 0))$ and $\nu_0 = \mathbf{0}_K$. Trivially, $(\nu_0,\nu_1) \in \mathcal{M}_{\text{all}}$ and $\tau^- \coloneqq \nu_1 - \nu_0 \in \mathcal{T}_{\text{all}}$, where $\tau^- \geq \mathbf{0}_K$. Since $\mu(a,\tau_0)$ is weakly causal, $\mu(a,\tau^-) \geq 0$ where
	\begin{align*}
		\mu(a,\tau^-) &= \E_{W_0}[a(X)\tau^-(X)] = \sum_{k=1}^K a(x_k) \1(a(x_k) < 0) \P_{W_0}(X = x_k) \geq 0.
	\end{align*}
	This implies $a(x_k) \geq 0$ for all $k \in \{1,\ldots,K\}$. Thus, $\mu(a,\tau_0)$ has a uniform causal representation in $\mathcal{T}_\text{all}$.
	
	Now suppose $\mu(a,\tau_0)$ has a uniform causal representation in $\mathcal{T}_\text{all}$, or that $a(x_k) \geq 0$ for $k \in \{1,\ldots,K\}$. Then, for any $(\nu_0,\nu_1) \in \mathcal{M}_{\text{all}}$ such that $\tau \coloneqq \nu_1 - \nu_0 \geq \mathbf{0}_K$, 
	\begin{align*}
		\mu(a,\tau) &= \sum_{k=1}^K a(x_k) \tau(x_k) \P_{W_0}(X = x_k) \geq 0.
	\end{align*}
	The inequality holds because $a(x_k)$ and $\tau(x_k)$ are nonnegative for all $k \in \{1,\ldots,K\}$. This last inequality is reversed if we instead assume that $\tau \leq \mathbf{0}_K$. Thus, $\mu(a,\tau_0)$ is weakly causal.
\end{proof}

\section{Proofs for Appendix \ref{appsec:estimation}}\label{appsec:proofsforsec6}

We use the following lemma in the proof of Theorem \ref{thm:cons_asynorm}.

\begin{lemma}\label{lemma:HDD_of_max}
	Let $\theta = (\theta(1),\ldots,\theta(K)) \in \R^K$ and let the mapping $\phi:\R^K \to \R$ be defined by $\phi(\theta) \coloneqq \max_{j \in \{1,\ldots,K\}} \theta(j)$. Then, $\phi$ is Hadamard directionally differentiable for all $\theta \in \R^K$ tangentially to $\R^K$ with directional derivative at $\theta$ in direction $h \in \R^K$ equal to
\begin{align*}
	\phi_\theta'(h) &= \max_{j \in \argmax_{k \in \{1,\ldots,K\}} \theta(k)}h(j).
\end{align*}	
\end{lemma}

\begin{proof}[Proof of Lemma \ref{lemma:HDD_of_max}]
	Let $h_n \rightarrow h \in \R^K$ and $t_n \searrow 0$ as $n\rightarrow \infty$. Then,
	\begin{align*}
		t_n^{-1}(\phi(\theta + t_n h_n) - \phi(\theta)) &= t_n^{-1}\left(\max_{k \in \{1,\ldots,K\}}(\theta(k) + t_n h_n(k)) - \max_{k \in \{1,\ldots,K\}} \theta(k)\right).
	\end{align*}
Let $\Theta_{\max} \coloneqq \{j \in \{1,\ldots,K\}: \theta(j) = \max_{k \in \{1,\ldots,K\}} \theta(k)\}$ and let $j_{\max}$ be an element of $\Theta_{\max}$. Then, $\max_{k \in \{1,\ldots,K\}} \theta(k) = \theta(j_{\max})$ and thus
\begin{align*}
	\frac{\phi(\theta + t_n h_n) - \phi(\theta)}{t_n} &= \max\left\{\frac{\theta(1) - \theta(j_{\max})}{t_n} + h_n(1), \ldots , \frac{\theta(K) - \theta(j_{\max})}{t_n} + h_n(K)\right\}.
\end{align*}
For each $j \in \Theta_{\max}$, $(\theta(j) - \theta(j_{\max}))/t_n + h_n(j) = h_n(j) \rightarrow h(j)$. For each $j \notin \Theta_{\max}$, $(\theta(j) - \theta(j_{\max}))/t_n \rightarrow -\infty$ since $\theta(j) - \theta(j_{\max}) < 0$ and $t_n \searrow 0$. Therefore, by continuity of the maximum operator in its arguments, $t_n^{-1}(\phi(\theta + t_n h_n) - \phi(\theta)) \rightarrow \max_{j \in \Theta_{\max}} h(j)$.
\end{proof}

\begin{proof}[Proof of Theorem \ref{thm:cons_asynorm}]
We begin by showing the consistency of $\widehat{a}_{\max}^{-1}$ for $\amax^{-1}$.

\smallskip
\noindent\textbf{Part 1: Consistency}

\noindent
The estimator $\avg \widehat{\alpha}(X_i)\widehat{w}_0(X_i)$ is consistent for $\E[\alpha(X)w_0(X)]$ since its components are assumed consistent by Assumption \ref{assn:prelim_est}, and by the continuous mapping theorem. The consistency of $\avg \widehat{w}_0(X_i)$ for $\E[w_0(X)]$ is similarly established.

We now consider the maximum term in the denominator. We can write
\begin{align*}
	\max_{i: \widehat{w}_0(X_i) > c_n} \widehat{\alpha}(X_i) &= \max_{x \in \widehat{\mathcal{X}}^+} \widehat{\alpha}(x).
\end{align*}
Recall that $\mathcal{X}^+ = \{x_j : p_j > 0, w_0(x_j) > 0\}$. We first show that $\P(\widehat{\mathcal{X}}^+ = \mathcal{X}^+) \rightarrow 1$ as $n \rightarrow \infty$. To see this, first suppose that $x_j \in \mathcal{X}^+$. Then, 
\begingroup
\allowdisplaybreaks
\begin{align*}
	\P(x_j \in \widehat{\mathcal{X}}^+) &= \P\left(\left\{\widehat{p}_j > 0\right\} \cap \left\{\widehat{w}_0(x_j) > c_n\right\}\right)\\
	&\geq \P\left(\widehat{p}_j > 0\right) + \P\left(\widehat{w}_0(x_j) > c_n\right) - 1,
\end{align*}
\endgroup
following an application of Bonferroni's inequality.

We have that $\P(\widehat{p}_j > 0) \rightarrow 1$ since $\widehat{p}_j \pconv p_j > 0$. We also have that $\P(\widehat{w}_0(x_j) > c_n) = \P(\widehat{w}_0(x_j) - c_n > 0) \rightarrow 1$ because $\widehat{w}_0(x_j) - c_n \pconv w_0(x_j) > 0$ by $c_n = o(1)$ and $w_0(x_j) > 0$, which follows from $x_j \in \mathcal{X}^+$. Therefore, $\P(x_j \notin \widehat{\mathcal{X}}^+) = 1 - \P(x_j \in \widehat{\mathcal{X}}^+) \leq 1 - (\P(\avg \1(X_i = x_j) > 0) + \P(\widehat{w}_0(x_j) > c_n) -1 )  \rightarrow 0$ as $n\rightarrow\infty$.

Now suppose that $x_j \notin \mathcal{X}^+$. Then
\begin{align*}
	\P(x_j \notin \widehat{\mathcal{X}}^+) &= \P\left(\left\{\widehat{p}_j = 0\right\} \cup \left\{\widehat{w}_0(x_j) \leq c_n\right\}\right)\\
	&\geq \P(\widehat{w}_0(x_j) \leq c_n) = \P(\sqrt{n}\widehat{w}_0(x_j) \leq \sqrt{n}c_n).
\end{align*}
By Assumption \ref{assn:prelim_est}, $\sqrt{n}\widehat{w}_0(x_j) = \sqrt{n}(\widehat{w}_0(x_j) - w_0(x_j)) \dconv \mathbb{Z}_{w_0}(j)$, since $w_0(x_j) = 0$ for $x_j \notin \mathcal{X}^+$. Therefore, $\sqrt{n}\widehat{w}_0(x_j) = O_p(1)$. Also $\sqrt{n}c_n \rightarrow +\infty$ by the theorem's assumption. Therefore, $\P(\sqrt{n}\widehat{w}_0(x_j) \leq \sqrt{n}c_n) \rightarrow 1$ and $\P(x_j \in \widehat{\mathcal{X}}^+) \rightarrow 0$ as $n\rightarrow \infty$. Because of this,
\begingroup
\allowdisplaybreaks
\begin{align*}
	\P(\widehat{\mathcal{X}}^+ = \mathcal{X}^+) &= \P\left(\bigcap_{x_j \in \mathcal{X}^+} \{x_j \in \widehat{\mathcal{X}}^+\} \cap \bigcap_{x_j \notin \mathcal{X}^+} \{x_j \notin \widehat{\mathcal{X}}^+\}\right)\\
	&= 1 - \P\left(\bigcup_{x_j \in \mathcal{X}^+} \{x_j \notin \widehat{\mathcal{X}}^+\} \cup \bigcup_{x_j \notin \mathcal{X}^+} \{x_j \in \widehat{\mathcal{X}}^+\}\right)\\
	&\geq 1 - \left(\sum_{j: x_j \in \mathcal{X}^+} \P(x_j \notin \widehat{\mathcal{X}}^+) + \sum_{j: x_j \notin \mathcal{X}^+} \P(x_j \in \widehat{\mathcal{X}}^+)\right)\\
	&\rightarrow 1 - 0 = 1.
\end{align*}
\endgroup
Thus, we obtain $\P\left(\max_{x \in \widehat{\mathcal{X}}^+} \widehat{\alpha}(x) = \max_{x \in \mathcal{X}^+} \widehat{\alpha}(x)\right) \geq \P( \widehat{\mathcal{X}}^+ = \mathcal{X}^+) \rightarrow 1$, which yields
\begin{align*}
	\max_{i: \widehat{w}_0(X_i) > c_n} \widehat{\alpha}(X_i) &= \max_{x \in \widehat{\mathcal{X}}^+} \widehat{\alpha}(x) = \max_{x \in \mathcal{X}^+} \widehat{\alpha}(x) + o_p(1).
\end{align*}
By the consistency of $\widehat{\alpha}$ for $\alpha$, the continuity of the maximum operator, and the continuous mapping theorem, $\max_{x \in \mathcal{X}^+} \widehat{\alpha}(x) \pconv \max_{x \in \mathcal{X}^+} \alpha(x)$. Because $\mathcal{X}^+ = \supp(X\mid W_0=1)$ is a finite set, we also have that  $\max_{x \in \mathcal{X}^+} \alpha(x) = \sup(\supp(\alpha(X)\mid W_0=1)) =\alphamax$. Another application of the continuous mapping theorem suffices to show that $\widehat{a}_{\max}^{-1}$ is consistent for $\amax^{-1}$.

\bigskip
\noindent \textbf{Part 2: Asymptotic Distribution}

\noindent
We first establish the joint limiting distribution of terms (i) $\sqrt{n}(\avg \widehat{\alpha}(X_i)\widehat{w}_0(X_i) - \E[\alpha(X)w_0(X)])$, (ii) $\sqrt{n}(\avg \widehat{w}_0(X_i) - \E[w_0(X_i)])$, and (iii) $\sqrt{n}(\max_{i: \widehat{w}_0(X_i) > c_n} \widehat{\alpha}(X_i) - \max_{x \in \mathcal{X}^+} \alpha(x))$. The terms (i) and (ii) can be expanded as follows:
\begingroup
\allowdisplaybreaks
\begin{align}
	\sqrt{n}&\left(\avg \widehat{\alpha}(X_i)\widehat{w}_0(X_i) - \E[\alpha(X)w_0(X)]\right)\notag\\
	 &= \sqrt{n}\left(\sum_{j=1}^K \widehat{\alpha}(x_j)\widehat{w}_0(x_j) \widehat{p}_j - \sum_{j=1}^K \alpha(x_j)w_0(x_j)p_j\right)\notag\\
		&= \sum_{j=1}^K w_0(x_j)p_j \sqrt{n}(\widehat{\alpha}(x_j) - \alpha(x_j)) + \sum_{j=1}^K \alpha(x_j)p_j \sqrt{n}(\widehat{w}_0(x_j) - w_0(x_j))\notag\\
		&\quad + \sum_{j=1}^K \alpha(x_j)w_0(x_j)\sqrt{n}(\widehat{p}_j - p_j) + o_p(1)\label{eq:numerator1}
\end{align}
\endgroup
and
\begingroup
\allowdisplaybreaks
\begin{align}
	\sqrt{n}&\left(\avg \widehat{w}_0(X_i) - \E[w_0(X)]\right) = \sqrt{n}\left(\sum_{j=1}^K \widehat{w}_0(x_j) \widehat{p}_j - \sum_{j=1}^K w_0(x_j)p_j\right)\notag\\
	&= \sum_{j=1}^K \left( p_j \sqrt{n}(\widehat{w}_0(x_j) - w_0(x_j)) + w_0(x_j)\sqrt{n}(\widehat{p}_j - p_j)\right) + o_p(1).\label{eq:denominator1}
\end{align}
\endgroup
For term (iii), we use the expansion
\begingroup
\allowdisplaybreaks
\begin{align}
	\sqrt{n}\left(\max_{i: \widehat{w}_0(X_i) > c_n} \widehat{\alpha}(X_i) - \max_{x \in \mathcal{X}^+} \alpha(x)\right) &= \sqrt{n}\left(\max_{i: \widehat{w}_0(X_i) > c_n} \widehat{\alpha}(X_i) - \max_{x \in \mathcal{X}^+} \widehat{\alpha}(x)\right)\label{eq:term1}\\
	&+ \sqrt{n}\left(\max_{x \in \mathcal{X}^+} \widehat{\alpha}(x) - \max_{x \in \mathcal{X}^+}\alpha(x)\right).\label{eq:term2}
\end{align}
\endgroup
The term in \eqref{eq:term1} is of order $o_p(1)$ because
\begingroup
\allowdisplaybreaks
\begin{align*}
	\P\left(\sqrt{n}\left(\max_{i: \widehat{w}_0(X_i) > c_n} \widehat{\alpha}(X_i) - \max_{x \in \mathcal{X}^+} \widehat{\alpha}(x)\right) = 0\right) &= \P\left(\max_{x \in \widehat{\mathcal{X}}^+} \widehat{\alpha}(x) = \max_{x \in \mathcal{X}^+} \widehat{\alpha}(x)\right)\\
	&\geq \P( \widehat{\mathcal{X}}^+ = \mathcal{X}^+) \rightarrow 1,
\end{align*}
\endgroup
as shown earlier.

The term in \eqref{eq:term2} can be analyzed using Theorem 2.1 in \cite{FangSantos2019}, which generalizes the delta method to the class of Hadamard directionally differentiable functions. Using Lemma \ref{lemma:HDD_of_max}, we have that
\begingroup
\allowdisplaybreaks
\begin{align}
	\sqrt{n}\left(\max_{x \in \mathcal{X}^+} \widehat{\alpha}(x) - \max_{x \in \mathcal{X}^+}\alpha(x)\right) &= \max_{x_j \in \argmax_{x \in \mathcal{X}^+} \alpha(x)}\sqrt{n}\left(\widehat{\alpha}(x_j) - \alpha(x_j)\right) + o_p(1)\notag\\
	&=: \max_{j \in \Psi_{\mathcal{X}^+}}\sqrt{n}\left(\widehat{\alpha}(x_j) - \alpha(x_j)\right) + o_p(1). \label{eq:denominator2}
\end{align}
\endgroup
Combining the expressions in \eqref{eq:numerator1}, \eqref{eq:denominator1}, and \eqref{eq:denominator2} with the delta method yields
\begingroup
\allowdisplaybreaks
\begin{align*}
	\sqrt{n}(\widehat{a}_{\max}^{-1} - \amax^{-1}) &= \frac{1}{\P(W_0=1) \max_{x \in \mathcal{X}^+} \alpha(x)}\sum_{j=1}^K \left(w_0(x_j)p_j \sqrt{n}(\widehat{\alpha}(x_j) - \alpha(x_j)) \right.\\
	&\qquad \left.  + \alpha(x_j)p_j \sqrt{n}(\widehat{w}_0(x_j) - w_0(x_j)) + \alpha(x_j)w_0(x_j)\sqrt{n}(\widehat{p}_j - p_j) \right)\\
	&\quad - \frac{\E_{W_0}[\alpha(X)]}{\P(W_0=1)\max_{x \in \mathcal{X}^+} \alpha(x)}\sum_{j=1}^K \left( p_j \sqrt{n}(\widehat{w}_0(x_j) - w_0(x_j)) + w_0(x_j)\sqrt{n}(\widehat{p}_j - p_j)\right)\\
	&\quad - \frac{\E_{W_0}[\alpha(X)]}{\max_{x \in \mathcal{X}^+} \alpha(x)^2}\max_{j \in \Psi_{\mathcal{X}^+}}\sqrt{n}\left(\widehat{\alpha}(x_j) - \alpha(x_j)\right) + o_p(1)\\
	&=\psi(\sqrt{n}(\widehat{\alpha} - \alpha),\sqrt{n}(\widehat{w}_0 - w_0),\sqrt{n}(\widehat{p} - p)) + o_p(1)\\
	&\dconv \psi(\mathbb{Z})
\end{align*}
\endgroup
by the continuity of $\psi$ and Assumption \ref{assn:prelim_est}.
\end{proof}

\begin{proof}[Proof of Theorem \ref{thm:bootstrap}]

We verify the validity of the bootstrap by appealing to Theorem 3.2 in \cite{FangSantos2019}. By their Remark 3.4, it is sufficient to show that the mapping $\widehat{\psi}$ satisfies the Lipschitz condition $|\widehat{\psi}(h') - \widehat{\psi}(h)| \leq C_n \|h' - h\|$ for any $h', h \in \R^{3K}$ and for $C_n = O_p(1)$, and that $\widehat{\psi}(h) \pconv \psi(h)$ for all $h \in \R^{3K}$.
	
We first verify the Lipschitz condition. Let $h \coloneqq (h_1,h_2,h_3)$ and $h' \coloneqq (h_1',h_2',h_3')$.
\begingroup
\allowdisplaybreaks
\begin{align}
		|\widehat{\psi}(h') - \widehat{\psi}(h)| &\leq \left|\sum_{j=1}^K \frac{\widehat{w}_0(x_j)\widehat{p}_j}{\widehat{\P}(W_0=1) \widehat{\alpha}_{\max}} (h_1'(j) - h_1(j))\right|\notag\\
		& \quad + \left|\frac{\widehat{\E}_{W_0}[\alpha(X)]}{\widehat{\alpha}_{\max}^2} \left(\max_{j \in \widehat{\Psi}_{\widehat{\mathcal{X}}^+}} h_1'(j) - \max_{j \in \widehat{\Psi}_{\widehat{\mathcal{X}}^+}} h_1(j)\right)\right|\notag\\
		&\quad + \left|\sum_{j=1}^K \frac{(\widehat{\alpha}(x_j) - \widehat{\E}_{W_0}[\alpha(X)])\widehat{p}_j}{\widehat{\P}(W_0=1) \widehat{\alpha}_{\max}} (h_2'(j) - h_2(j))\right|\notag\\
		&\quad + \left|\sum_{j=1}^K \frac{(\widehat{\alpha}(x_j) - \widehat{\E}_{W_0}[\alpha(X)])\widehat{w}_0(x_j)}{\widehat{\P}(W_0 = 1)\widehat{\alpha}_{\max}}(h_3'(j) - h_3(j))\right|\notag\\
		&\leq  \left(\sum_{j=1}^K \frac{\widehat{w}_0(x_j)^2\widehat{p}_j^2}{\widehat{\P}(W_0=1)^2 \widehat{\alpha}_{\max}^2}\right)^{1/2} \|h_1' - h_1\|\label{eq:lips_term1}\\
		&\quad + \frac{|\widehat{\E}_{W_0}[\alpha(X)]|}{\widehat{\alpha}_{\max}^2} \left|\max_{j \in \widehat{\Psi}_{\widehat{\mathcal{X}}^+}} h_1'(j) - \max_{j \in \widehat{\Psi}_{\widehat{\mathcal{X}}^+}} h_1(j)\right|\\
		&\quad + \left(\sum_{j=1}^K \frac{(\widehat{\alpha}(x_j) - \widehat{\E}_{W_0}[\alpha(X)])^2\widehat{p}_j^2}{\widehat{\P}(W_0=1)^2 \widehat{\alpha}_{\max}^2}\right)^{1/2} \|h_2' - h_2\|\\
		&\quad + \left(\sum_{j=1}^K \frac{(\widehat{\alpha}(x_j) - \widehat{\E}_{W_0}[\alpha(X)])^2\widehat{w}_0(x_j)^2}{\widehat{\P}(W_0 = 1)^2\widehat{\alpha}_{\max}^2}\right)^{1/2}\|h_3' - h_3\|, \label{eq:lips_term2}
\end{align}
\endgroup
where we applied the Cauchy--Schwarz inequality several times. Note that the maximum function is Lipschitz with Lipschitz constant one and therefore
\begin{align*}
	 \left|\max_{j \in \widehat{\Psi}_{\widehat{\mathcal{X}}^+}} h_1'(j) - \max_{j \in \widehat{\Psi}_{\widehat{\mathcal{X}}^+}} h_1(j)\right| &\leq \sum_{j \in \widehat{\Psi}_{\widehat{\mathcal{X}}^+}} |h_1'(j) - h_1(j)|\notag\\
	 &\leq \sum_{j=1}^K |h_1'(j) - h_1(j)|\notag\\
	 &\leq \sqrt{K}\|h_1' - h_1\|. 
\end{align*}
Combining equations \eqref{eq:lips_term1}--\eqref{eq:lips_term2} with the consistency of $(\widehat{\alpha},\widehat{w}_0,\widehat{p})$ established in Theorem \ref{thm:cons_asynorm} shows that $|\widehat{\psi}(h') - \widehat{\psi}(h)| \leq C_n \|h' - h\|$ for any $h', h \in \R^{3K}$ and for $C_n = O_p(1)$. 

We finish this proof by considering the consistency of the different components of $\widehat{\psi}(h)$. Applying Theorem~\ref{thm:cons_asynorm} and the continuous mapping theorem, the following terms, all present in $\widehat{\psi}(h)$, are consistent for their counterparts in $\psi(h)$:
\begin{align*}
	&\left(\sum_{j=1}^K \frac{\widehat{w}_0(x_j)\widehat{p}_j}{\widehat{\P}(W_0=1) \widehat{\alpha}_{\max}}h_1(j), \frac{\widehat{\E}_{W_0}[\alpha(X)]}{\widehat{\alpha}_{\max}^2},\right.\\
	 &\quad\left. \sum_{j=1}^K \frac{(\widehat{\alpha}(x_j) - \widehat{\E}_{W_0}[\alpha(X)])\widehat{p}_j}{\widehat{\P}(W_0=1) \widehat{\alpha}_{\max}}h_2(j),\sum_{j=1}^K \frac{(\widehat{\alpha}(x_j) - \widehat{\E}_{W_0}[\alpha(X)])\widehat{w}_0(x_j)}{\widehat{\P}(W_0 = 1)\widehat{\alpha}_{\max}}h_3(j)\right).
\end{align*}
To finish proving that $\widehat{\psi}(h) \pconv \psi(h)$, we show that $\max_{j \in \widehat{\Psi}_{\widehat{\mathcal{X}}^+}} h_1(j) \pconv \max_{j \in \Psi_{\mathcal{X}^+}} h_1(j)$. This holds if the set $\widehat{\Psi}_{\widehat{\mathcal{X}}^+}$ is consistent for $\Psi_{\mathcal{X}^+}$, which we establish here. Suppose that $k \in \Psi_{\mathcal{X}^+}$. Then,
\begin{align*}
	\P(k \in \widehat{\Psi}_{\widehat{\mathcal{X}}^+}) &= \P\left(k \in \widehat{\mathcal{X}}^+, \widehat{\alpha}(x_k) \geq \max_{j \in \widehat{\mathcal{X}}^+} \widehat{\alpha}(x_j) - \xi_n \right)\\
	&= \P\left(k \in \widehat{\mathcal{X}}^+, \sqrt{n}(\widehat{\alpha}(x_k) - \max_{j \in \widehat{\mathcal{X}}^+} \widehat{\alpha}(x_j)) \geq - \sqrt{n}\xi_n\right)\\
	&= \P\left(k \in \widehat{\mathcal{X}}^+, \sqrt{n}\left( \widehat\alpha(x_k) - \alpha(x_k)\right) - \sqrt{n}(\max_{j \in \widehat{\mathcal{X}}^+} \widehat{\alpha}(x_j) - \alpha(x_k)) \geq - \sqrt{n}\xi_n\right).
\end{align*}
To evaluate the limit of this last expression, first note that $\sqrt{n}\left( \widehat\alpha(x_k) - \alpha(x_k)\right) = O_p(1)$ by Assumption~\ref{assn:prelim_est}. Moreover, by the proof of Theorem \ref{thm:cons_asynorm}, $ \sqrt{n}\left(\max_{j \in \widehat{\mathcal{X}}^+} \widehat{\alpha}(x_j) - \alpha(x_k)\right) = O_p(1)$, because $k \in \Psi_{\mathcal{X}^+}$. Moreover, $k \in  \Psi_{\mathcal{X}^+}$ implies $k \in \mathcal{X}^+$, so $\P(k \in \widehat{\mathcal{X}}^+) \to \P(k \in \mathcal{X}^+) = 1$. Since $-\sqrt{n}\xi_n \rightarrow -\infty$, $\P(k \in \widehat{\Psi}_{\widehat{\mathcal{X}}^+}) \to 1$ when $k \in \Psi_{\mathcal{X}^+}$. 

Now suppose that $k \notin \Psi_{\mathcal{X}^+}$. Then, $\P(k \in \widehat{\Psi}_{\widehat{\mathcal{X}}^+}) = \P(k \in \widehat{\mathcal{X}}^+,\widehat{\alpha}(x_k) \geq \max_{j \in \widehat{\mathcal{X}}^+} \widehat{\alpha}(x_j) - \xi_n )$. If $k \notin \mathcal{X}^+$, then $\P(k \in \widehat{\Psi}_{\widehat{\mathcal{X}}^+}) \leq \P(k \in \widehat{\mathcal{X}}^+) \rightarrow 0$ as $n\rightarrow \infty$, so $\P(k \in \widehat{\Psi}_{\widehat{\mathcal{X}}^+}) \to 0$. 

If $k \in \mathcal{X}^+$, then
\begingroup
\allowdisplaybreaks
\begin{align*}
	\P(k \in \widehat{\Psi}_{\widehat{\mathcal{X}}^+}) &\leq \P(\widehat{\alpha}(x_k) \geq \max_{j \in \widehat{\mathcal{X}}^+} \widehat{\alpha}(x_j) - \xi_n )\rightarrow \P(\alpha(x_k) \geq \max_{j \in \mathcal{X}^+} \alpha(x_j) - 0) = 0,
\end{align*}
\endgroup
where the last equality holds from $k \notin \Psi_{\mathcal{X}^+}$. Therefore, $\P(\widehat{\Psi}_{\widehat{\mathcal{X}}^+} = \Psi_{\mathcal{X}^+}) \rightarrow 1$ as $n\rightarrow\infty$. This implies $\widehat{\psi}(h) \pconv \psi(h)$, which concludes the proof.
\end{proof}

\section{Additional Derivations for Difference-in-Differences}\label{appsec:DiD}

\cite{Goodman-Bacon2021} provides the following representation of the two-way fixed effects estimand under the assumption that group-time average treatment effects are constant over time:
\begin{align*}
	\beta_\text{TWFE} &= \sum_{k: \; \var(D\mid G=k)>0} \Bigg[ \sum_{j=1}^{k-1} \sigma_{jk}^{k} + \sum_{j=k+1}^{K} \sigma_{kj}^{k} \Bigg] \cdot \E[Y(1) - Y(0)\mid G=k,D=1],
\end{align*}
where
\begin{displaymath}
\sigma_{jk}^{k}  \coloneqq  \frac{\Prob(G=j) \cdot \Prob(G=k) \cdot \Prob(D=1 \mid  G=k) \cdot \big[ \Prob(D=1 \mid  G=j) - \Prob(D=1 \mid  G=k) \big]}{\var(D^{\perp \left( G_{t_1}, \ldots, G_{t_{K-1}}, P_1, \ldots, P_T \right)})}
\end{displaymath}
and
\begin{displaymath}
\sigma_{kj}^{k}  \coloneqq  \frac{\Prob(G=j) \cdot \Prob(G=k) \cdot \big[ 1 - \Prob(D=1 \mid  G=k) \big] \big[ \Prob(D=1 \mid  G=k) - \Prob(D=1 \mid  G=j) \big]}{\var(D^{\perp \left( G_{t_1}, \ldots, G_{t_{K-1}}, P_1, \ldots, P_T \right)})}.
\end{displaymath}
Here $A^{\perp B}$ is used to denote the residual in the linear projection of $A$ on $(1,B)$. It is also the case that $\sum_{k: \; \var(D\mid G=k)>0} \sum_{l>k} \big( \sigma_{kl}^{k} + \sigma_{kl}^{l} \big) = 1$.\footnote{The result in \cite{Goodman-Bacon2021} technically also includes a weight $\sigma_{kU}$ attached to the contrast between group $k$ and the never-treated group. We subsume this weight under $\sigma_{kj}^{k}$, and likewise subsume the weight on the contrast with the always-treated group under $\sigma_{jk}^{k}$.} When we compare this representation with Proposition \ref{prop:GB}, that is,
\begingroup
\allowdisplaybreaks
\begin{eqnarray*}
		\beta_\text{TWFE} & = & \frac{\E [ a_\text{TWFE,H}(G) \cdot \Prob(D=1\mid G) \cdot \tau_0(G) ]}{\E [ a_\text{TWFE,H}(G) \cdot \Prob(D=1\mid G) ]} \\
		& = & \frac{\sum_{k: \; \var(D\mid G=k)>0} \Prob(G=k) \cdot a_\text{TWFE,H}(k) \cdot \Prob(D=1\mid G=k) \cdot \tau_0(k)}{\sum_{k: \; \var(D\mid G=k)>0} \Prob(G=k) \cdot a_\text{TWFE,H}(k) \cdot \Prob(D=1\mid G=k)},
\end{eqnarray*}
\endgroup
where $\tau_0(k) = \E[Y(1) - Y(0) \mid G = k, D=1]$, it becomes clear that, for each group $k$ other than the always treated and the never treated,
\begin{align*}
&a_\text{TWFE,H}(k) \cdot \Prob(D=1\mid G=k)\\
&=  \sum_{j=1}^{k-1} \Prob(G=j) \cdot \Prob(D=1 \mid  G=k) \cdot \big[ \Prob(D=1 \mid  G=j) - \Prob(D=1 \mid  G=k) \big] \\
&\quad + \sum_{j=k+1}^{K} \Prob(G=j) \cdot \big[ 1 - \Prob(D=1 \mid  G=k) \big] \big[ \Prob(D=1 \mid  G=k) - \Prob(D=1 \mid  G=j) \big],
\end{align*}
and this, in turn, implies that
\begin{align}
&a_\text{TWFE,H}(k)\notag\\
&= \sum_{j=1}^{k-1} \Prob(G=j) \cdot \big[ \Prob(D=1 \mid  G=j) - \Prob(D=1 \mid  G=k) \big] \nonumber\\
&\quad + \sum_{j=k+1}^{K} \Prob(G=j) \cdot \big[ \Prob(D=1 \mid  G=k) - \Prob(D=1 \mid  G=j) \big] \cdot \frac{1 - \Prob(D=1 \mid  G=k)}{\Prob(D=1 \mid  G=k)}.\label{eq:DiD_weights_GB}
\end{align}

\subsection*{Equivalence of Weight Functions}

We now show that the weights obtained in equation \eqref{eq:DiD_weights_GB} are equivalent to those in Proposition \ref{prop:GB}. First, we rewrite the weights in \eqref{eq:DiD_weights_GB} as follows:
\begingroup
\allowdisplaybreaks
\begin{align*}
	&a_\text{TWFE,H}(k)\\
	&= \sum_{j=1}^{k-1} \Prob(G=j) \cdot \big[ \Prob(D=1 \mid  G=j) - \Prob(D=1 \mid  G=k) \big] \\
	&\quad + \; \sum_{j=k+1}^{K} \Prob(G=j) \cdot \big[ \Prob(D=1 \mid  G=k) - \Prob(D=1 \mid  G=j) \big] \cdot \frac{1 - \Prob(D=1 \mid  G=k)}{\Prob(D=1 \mid  G=k)} \\
	&= \Prob(D=1,G < k) - \Prob(G < k)\E[D\mid G=k] + (1 - \E[D\mid G=k])\Prob(G > k)\\
	&\quad  - \frac{1 - \E[D\mid G=k]}{\E[D\mid G=k]}\Prob(D=1,G>k)\\
	&= \Prob(D=1,G < k) - \Prob(G < k)\E[D\mid G=k] + \Prob(G>k) - \E[D\mid G=k]\Prob(G > k)\\
	&\quad - \frac{1}{\E[D\mid G=k]}\Prob(D=1,G>k) + \Prob(D=1,G>k)\\
	&= \Prob(D=1,G \neq k) - \E[D\mid G=k]\Prob(G \neq k) + \Prob(G > k)\left(1 - \frac{\E[D\mid G>k]}{\E[D\mid G=k]}\right)\\
	&= (\E[D] - \E[D\mid G=k]\Prob(G=k))  - \E[D\mid G=k]\Prob(G\neq k)\\
	&\quad + \Prob(G > k)\left(1 - \frac{\E[D\mid G>k]}{\E[D\mid G=k]}\right)\\
	&= \E[D] - \E[D\mid G=k] + \Prob(G > k)\left(1 - \frac{\E[D\mid G>k]}{\E[D\mid G=k]}\right).
\end{align*}
\endgroup
For $k \in \{2,\ldots,T\}$, the weights in Proposition \ref{prop:GB} are equal to
\begingroup
\allowdisplaybreaks
\begin{align}
  \E[1-\E[D\mid G]-\E[D\mid P]+\E[D]\mid G=k,D=1]
  &= \E[1-\E[D\mid G]-\E[D\mid P]+\E[D]\mid G=k,P\geq k]\notag\\
  &= 1-\E[D\mid G=k]-\E[D\mid P\geq k]+\E[D],
  \label{eq:DID_weights_GB_2}
\end{align}
\endgroup
because they are the average of the weights in Proposition \ref{prop:CDH} conditional on $G=k$ and $D=1$. The proof of Proposition \ref{prop:GB} explicitly shows that 
\begin{align*}
	&\E[1-\E[D\mid G]-\E[D\mid P]+\E[D]\mid G=k,D=1]\\
	&= \Prob(D=0\mid G=k)\cdot(\Prob(D=0\mid P \geq k) + \Prob(D=1\mid P < k)). 
\end{align*}
Let us look at the difference between the weights in \eqref{eq:DiD_weights_GB} and \eqref{eq:DID_weights_GB_2}. Fix $k \in \{2,\ldots,T\}$ and write:
\begingroup
\allowdisplaybreaks
\begin{align*}
	& \left(1 - \E[D\mid G=k] - \E[D\mid P \geq k] + \E[D]\right)\\
	&\quad  - \left(\E[D] - \E[D\mid G=k] + \Prob(G > k)\left(1 - \frac{\E[D\mid G>k]}{\E[D\mid G=k]}\right)\right)\\
	&= 1 - \E[D\mid P \geq k] - \Prob(G > k) + \frac{\E[D \1(G > k)]}{\E[D\mid G=k]}\\
	&= \E[\1(G \leq k)] - \frac{\E[D \1(P \geq k)]}{\E[\1(P \geq k)]} + \frac{\E[D \1(G > k)]}{\E[\1(k \leq P)]}\\
	&= \frac{1}{\E[\1(k \leq P)]}\left(F_G(k) \E[\1(k \leq P)] + \E[D \1(G > k)]  - \E[D\1(P \geq k)]\right)\\
	&= \frac{1}{\E[\1(k \leq P)]}\left(F_G(k) \E[\1(k \leq P)] + \E[\1(k < G \leq P)]  - \E[\E[D\mid P]\1(P \geq k)]\right)\\
	&= \frac{1}{\E[\1(k \leq P)]}\left(F_G(k) \E[\1(k \leq P)] + \E[\E[\1(k < G \leq P)\mid P]]  - \E[F_G(P)\1(P \geq k)]\right)\\
	&= \frac{1}{\E[\1(k \leq P)]}\left(F_G(k) \E[\1(k \leq P)] + \E[(F_G(P) - F_G(k))\1(P \geq k)]  - \E[F_G(P)\1(P \geq k)]\right)\\
	&= \frac{1}{\E[\1(k \leq P)]}\left(F_G(k) \E[\1(k \leq P)] + \E[F_G(P)\1(P \geq k)] - F_G(k)\E[\1(P \geq k)]\right.\\
	&\left. \quad  - \E[F_G(P)\1(P \geq k)]\right)\\
	&= 0.
\end{align*}
\endgroup
Therefore, the weights in Proposition \ref{prop:GB} and equations \eqref{eq:DiD_weights_GB} and \eqref{eq:DID_weights_GB_2} are all equal to one another.

\end{document}